\documentclass[11pt]{article}
\usepackage{vitercik}
\usepackage{subcaption}
\usepackage{url}
\usepackage{todonotes}
\usepackage{multirow}
\usepackage{tabularx}

\newcommand{\hyp}{\mathcal{H}}
\newcolumntype{L}[1]{>{\hsize=#1\hsize\raggedright\arraybackslash}X}%

\begin{document}
\title{Estimating Approximate Incentive Compatibility}
\author{
Maria-Florina Balcan \\ \small Carnegie Mellon University \\ \small \texttt{ninamf@cs.cmu.edu}
\and 
Tuomas Sandholm \\ \small Carnegie Mellon University \\
\small Optimized Markets, Inc.\\
\small Strategic Machine, Inc.\\
\small Strategy Robot, Inc.\\
 \small \texttt{sandholm@cs.cmu.edu}
\and 
Ellen Vitercik \\ \small Stanford University \\ \small \texttt{vitercik@stanford.edu}}
\date{\today}

\maketitle

\begin{abstract}
In practice, most mechanisms for selling, buying, matching, voting, and so on are not incentive compatible. We present techniques for estimating how far a mechanism is from incentive compatible. Given samples from the agents' type distribution, we show how to estimate the extent to which an agent can improve his utility by misreporting his type.  We do so by first measuring the maximum utility an agent can gain by misreporting his type on average over the samples, assuming his true and reported types are from a finite subset---which our technique constructs---of the type space. The challenge is that by measuring utility gains over a finite subset of the type space, we might miss type pairs $\theta$ and $\hat{\theta}$ where an agent with type $\theta$ can greatly improve his utility by reporting type $\hat{\theta}$. Indeed, our primary technical contribution is proving that the maximum utility gain over this finite subset nearly matches the maximum utility gain overall, despite the volatility of the utility functions we study. We apply our tools to the single-item and combinatorial first-price auctions, generalized second-price auction, discriminatory auction, uniform-price auction, and second-price auction with spiteful bidders.
\end{abstract}

\section{Introduction}
Incentive compatibility~\citep{Hurwicz72:Informationally} is a fundamental concept in mechanism design. Under an incentive compatible mechanism, it is in every agent's best interest to report their type truthfully.
Nonetheless, practitioners have long employed mechanisms that are not incentive compatible, also called \emph{manipulable} mechanisms. This is the case in many settings for selling, buying, matching
(such as school choice),
voting, and so on. For example, most real-world auctions are implemented using the first-price mechanism. In multi-unit sales settings, the U.S. Treasury has used discriminatory auctions, a variant of the first-price auction, to sell treasury bills since 1929~\citep{Krishna02:Auction}. Similarly, electricity generators in the U.K. use discriminatory auctions to sell their output~\citep{Krishna02:Auction}. Sponsored search auctions are typically implemented using variants of the generalized second-price auction~\citep{Edelman07:Internet, Varian07:Position}. In the past year, many major display ad exchanges including AppNexus, Rubicon, and Index Exchange 
have transitioned to first-price auctions, driven by ad buyers who believe it offers a higher degree of transparency~\citep{Parkin18:Year, Harada18:Ad}. Finally, nearly all fielded combinatorial auctions are manipulable. Essentially all combinatorial sourcing auctions are implemented using the first-price mechanism~\citep{Sandholm13:Very-Large-Scale}. Combinatorial spectrum auctions are conducted using a variety of manipulable mechanisms. Even the ``incentive auction'' used to source spectrum licenses back from low-value broadcasters---which has sometimes been hailed as obviously incentive compatible---is manipulable once one takes into account the fact that many owners own multiple broadcasting stations, or the fact that stations do not only have the option to keep or relinquish their license, but also the option to move to (two) less desirable spectrum ranges~\citep{Nguyen15:Multi-option}.

A variety of reasons have been suggested to help explain why manipulable mechanisms are used in practice. To name a few,
\begin{enumerate} 
	\item Oftentimes, the rules are easier to explain, as exemplified by the exceptionally simple first-price auction.
	\item Incentive compatibility can cease to hold when determining one's own valuation is costly (for example, due to computation or information gathering effort), even for the classic Vickrey auction~\citep{Sandholm00:Issues}. This is also true of the Vickrey-Clarke-Groves (VCG) mechanism, which requires bidders to submit bids for every bundle, a task that generally requires a prohibitive amount of valuation computation or information acquisition~\citep{Sand93,Parkes99:Optimal,Conen01:Preference,Sandholm06:Preference}.
	\item Single-shot incentive compatible mechanisms are generally not incentive compatible when the bid-taker uses bids from one auction to adjust the parameters of later auctions~\citep{Acquisti:Varian:05,Sandholm13:Very-Large-Scale,Amin13:Learning}.
	\item Incentive compatible mechanisms may leak the agents' sensitive private information~\citep{Rothkopf90:Why}. 
	\item Incentive compatibility typically ceases to hold if agents are not risk neutral (i.e., the utility functions are not quasi-linear). 
\end{enumerate}

Due in part to the ubiquity of manipulable mechanisms, a growing body of additional research has explored mechanisms that are not incentive compatible~\citep{Kothari03:Approximately, Archer04:Approximate,Conitzer07:Incremental,Dekel10:Incentive, Lubin12:Approximate, Mennle14:Axiomatic,Dutting2015:Payment, Azevedo18:Strategy, Golowich18:Deep, Duetting19:Optimal, Feng18:Deep}.
A popular and widely-studied relaxation of incentive compatibility is \emph{$\gamma$-incentive compatibility}~\citep{Kothari03:Approximately, Archer04:Approximate, Dekel10:Incentive, Lubin12:Approximate, Mennle14:Axiomatic,Dutting2015:Payment, Azevedo18:Strategy}, which requires that no agent can improve his utility by more than $\gamma$ when he misreports his type.

\subsection{Our contributions}

Much of the literature on $\gamma$-incentive compatibility rests
on the strong assumption that the agents' type distribution is known.
In reality, this information is rarely available.
We relax this assumption and instead assume we only have samples from the distribution~\citep{Likhodedov04:Boosting, Likhodedov05:Approximating,Sandholm15:Automated}.
We present techniques with provable guarantees that the mechanism designer can use to estimate how far a mechanism is from incentive compatible. We analyze the \emph{ex-interim}\footnote{We do not study \emph{ex-post} or \emph{dominant-strategy} approximate incentive compatibility because they are worst-case, distribution-independent notions. Therefore, we cannot hope to measure the \emph{ex-post} or dominant-strategy approximation factors using samples from the agents' type distribution.}
setting: we bound the amount any agent can improve his utility by misreporting his type, in expectation over the other agents' types, no matter his true type. We adopt the common assumption that each agent's type is drawn independently from the other agents' types.

 Our estimate is simple: it measures the maximum utility an agent can gain by misreporting his type on average over the samples, whenever his true and reported types are from a finite subset of the type space. We bound the difference between our incentive compatibility estimate and the true incentive compatibility approximation factor $\gamma$.
 
Accurately estimating the approximation factor $\gamma$ provides the mechanism designer with a data-dependent metric for choosing among manipulable mechanisms, which can be used in conjunction with other metrics such as expected revenue.
 When the distribution over agents' types is unknown, as we assume it is, one cannot use a distribution-dependent, analytic approach to deriving the approximation factor $\gamma$. However, a distribution-independent bound on the approximation factor that holds even for a worst-case distribution over types may be uninformative when the underlying distribution is not worst-case. The estimate of the distribution's incentive compatibility approximation factor that we provide allows a mechanism designer to make a more informed decision when selecting a manipulable mechanism to field. As a concrete example, machine learning is increasingly being used to learn high-revenue mechanisms from within classes of manipulable mechanisms (as in the work by~\citet{Duetting19:Optimal} in the multi-agent, multi-item setting\footnote{In contrast, in the single-item setting,~\citet{Duetting19:Optimal} use a neural network architecture that is incentive compatible by design.}). Our guarantees can be used to better understand how susceptible the resulting mechanism is to strategic manipulation.

We apply our estimation technique to a variety of auction classes. We begin with the first-price auction, in both single-item and combinatorial settings. Our guarantees can be used by display ad exchanges, for instance, to measure the extent to which incentive compatibility will be compromised if the exchange transitions to using first-price auctions~\citep{Parkin18:Year, Harada18:Ad}. In the single-item setting, we prove that the difference between our estimate and the true incentive compatibility approximation factor is $\tilde O \left(\left(n + \kappa^{-1}\right) / \sqrt{N}\right)$, where $n$ is the number of bidders, $N$ is the number of samples, and $[0, \kappa]$ contains the range of the density functions defining the agents' type distribution. We prove the same bound for the second-price auction with \emph{spiteful bidders}~\citep{Brandt07:Spiteful, Morgan03:Spite, Sharma10:Asymmetric, Tang12:Optimal}, where each bidder's utility not only increases when his surplus increases but also decreases when the other bidders' surpluses increase.

In a similar direction, we analyze the class of generalized second-price auctions~\citep{Edelman07:Internet}, where $m$ sponsored search slots are for sale. The mechanism designer assigns a real-valued weight per bidder, collects a bid per bidder indicating their value per click, and allocates the slots in order of the bidders' weighted bids. In this setting, we prove that the difference between our incentive compatibility estimate and the true incentive compatibility approximation bound is $\tilde O \left(\left(n^{3/2} + \kappa^{-1}\right) / \sqrt{N}\right)$.

We also analyze multi-parameter mechanisms beyond the first-price combinatorial auction, namely, the uniform-price and discriminatory auctions, which are used extensively in markets around the world~\citep{Krishna02:Auction}. In both, the auctioneer has $m$ identical units of a single good to sell. Each bidder submits $m$ bids indicating their value for each additional unit of the good. The number of goods the auctioneer allocates to each bidder equals the number of bids that bidder has in the top $m$ bids. Under a discriminatory auction, each agent pays the amount she bid for the number of units she is allocated. Under a uniform-price auction, the agents pay the market-clearing price, meaning the total amount demanded equals the total amount supplied. In both cases, we prove that the difference between our incentive compatibility estimate and the true incentive compatibility approximation bound is $\tilde O \left(\left(nm^2 + \kappa^{-1}\right) / \sqrt{N}\right)$.

A strength of our estimation techniques is that they are application-agnostic.
For example, they can be used as a tool in \emph{incremental mechanism design}~\citep{Conitzer07:Incremental}, a subfield of \emph{automated mechanism design}~\citep{Conitzer02:Mechanism, Sandholm03:Automated}, where the mechanism designer gradually adds incentive compatibility constraints to her optimization problem until she has met a desired incentive compatibility guarantee.

\paragraph{Key challenges.} To prove our guarantees, we must estimate the value $\gamma$ defined such that no agent can misreport his type in order to improve his expected utility by more than $\gamma$, no matter his true type. We propose estimating $\gamma$ by measuring the extent to which an agent can improve his utility by misreporting his type on average over the samples, whenever his true and reported types are from a finite subset $\cG$ of the type space. We denote this estimate as $\hat{\gamma}$.
The challenge is that by searching over a subset of the type space, we might miss pairs of types $\theta$ and $\hat{\theta}$ where an agent with type $\theta$ can greatly improve his expected utility by misreporting his type as $\hat{\theta}$. Indeed, utility functions are often volatile in mechanism design settings.  For example, under the first- and second-price auctions, nudging an agent's bid from below the other agents' largest bid to above will change the allocation, causing a jump in utility. Thus, there are two questions we must address: which finite subset should we search over and how do we relate our estimate $\hat{\gamma}$ to the true approximation factor $\gamma$?

We provide two approaches to constructing the cover $\cG$. The first is to run a greedy procedure based off a classic algorithm from theoretical machine learning. This approach is extremely versatile: it provides strong guarantees no matter the setting. However, it may be difficult to implement. 

Meanwhile, implementing our second approach is straightforward: the cover is simply a uniform grid over the type space (assuming the type space equals $[0,1]^m$ for some integer $m$). The efficacy of this approach depends on a ``niceness'' property that holds under mild assumptions. To analyze this second approach, we must understand how the edge-length of the grid effects our error bound relating the estimate $\hat{\gamma}$ to the true approximation factor $\gamma$. To do so, we rely on the notion of \emph{dispersion}~\citep{Balcan18:Dispersion}. Roughly speaking, a set of piecewise Lipschitz functions is $(w,k)$-dispersed if every ball of radius $w$ in the domain contains at most $k$ of the functions' discontinuities. Given a set of $N$ samples from the distribution over agents' types, we analyze the dispersion of function sequences from two related families. In both families, the function sequences are of length $N$: each function is defined by one of the $N$ samples. In the first family, each sequence is defined by a true type $\vec{\theta}$ and measures a fixed agent's utility as a function of her reported type $\hat{\vec{\theta}}$, when the other agents' bids are defined by the samples. In the second family, each sequence is defined by a reported type $\vec{\hat{\theta}}$ and measures the agent's utility as a function of her true type $\vec{\theta}$.
We show that if all function sequences from both classes are $(w,k)$-dispersed, then we can use a grid with edge-length $w$ to discretize the agents' type space.

We show that for a wide range of mechanism classes, dispersion holds under mild assumptions. As we describe more in Section~\ref{sec:grid}, this requires us to prove that with high probability, each function sequence from several infinite families of sequences is dispersed. This facet of our analysis is notably different from prior research by \citet{Balcan18:Dispersion}: in their applications, it is enough to show that with high probability, a single, finite sequence of functions is dispersed. Our proofs thus necessitate that we carefully examine the structure of the utility functions that we analyze. In particular, we prove that across all of the infinitely-many sequences, the functions' discontinuities are invariant. As a result, it is enough to prove dispersion for single generic sequence in order to guarantee dispersion for all of the sequences.

Finally, we prove that for both choices of the cover $\cG$, so long as the intrinsic complexities of the agents' utility functions are not too large (as measured by the learning-theoretic notion of \emph{pseudo-dimension}~\citep{Pollard84:Convergence}), then our estimate $\hat{\gamma}$ quickly converges to the true approximation factor $\gamma$ as the number of samples grows.

\subsection{Additional related research}\label{sec:related}

\paragraph{Strategy-proofness in the large.}
\citet{Azevedo18:Strategy} propose a variation on approximate incentive compatibility called \emph{strategy-proofness
in the large (SP-L)}. SP-L requires that it is approximately optimal for agents to report their types truthfully
in sufficiently large markets. As in our paper, SP-L is a condition on \emph{ex-interim} $\gamma$-incentive compatibility. The authors argue that SP-L approximates, in large markets, attractive properties of a mechanism such as strategic simplicity and robustness. They categorize a number of mechanisms as either SP-L or not. For example, they show that the discriminatory auction is manipulable in the large whereas the uniform-price auction is SP-L. Measuring a mechanism's SP-L approximation factor requires knowledge of the distribution over agents' types, whereas we only require sample access to this distribution. Moreover, we do not make any large-market assumptions: our guarantees hold regardless of the number of agents.

\paragraph{Comparing mechanisms by their vulnerability to manipulation.} \citet{Pathak13:School} analyze \emph{ex-post} incentive compatibility without any connection to approximate incentive compatibility. They say that one mechanism $M$ is at least as manipulable as another $M'$ if every type profile that is vulnerable to manipulation under $M$ is also vulnerable to manipulation under $M'$. They apply their formalism in the context of school assignment mechanisms, the uniform-price auction, the discriminatory auction, and several keyword auctions. We do not study \emph{ex-post} approximate incentive compatibility because it is a worst-case, distribution-independent notion. Therefore, we cannot hope to measure an \emph{ex-post} approximation factor using samples from the agents' type distribution. Rather, we are concerned with providing data-dependent bounds on the \emph{ex-interim} and \emph{ex-ante} approximation factors. Another major difference is that our work provides quantitative results on manipulability while theirs provides boolean comparisons as to the relative manipulability of mechanisms. Finally, our measure applies to all mechanisms while theirs cannot rank all mechanisms because in many settings, pairs of mechanisms are incomparable according to their boolean measure.

\paragraph{Mechanism design via deep learning.} In mechanism design via deep learning~\citep{Feng18:Deep, Duetting19:Optimal, Golowich18:Deep}, the learning algorithm receives samples from the distribution over agents' types. The resulting allocation and payment functions are characterized by neural networks, and thus the corresponding mechanism may not be incentive compatible. In an attempt to make these mechanisms nearly incentive compatible, the authors of these works add constraints to the deep learning optimization problem enforcing that the resulting mechanism be incentive compatible over a set of buyer values sampled from the underlying, unknown distribution.

In research that was conducted simultaneously, \citet{Duetting19:Optimal} provide guarantees on the extent to which the resulting mechanism---characterized by a neural network---is approximately incentive compatible, in the sense of \emph{expected ex-post incentive compatibility}: in expectation over all agents' types, what is the maximum amount that an agent can improve her utility by misreporting her type? This is a different notion of approximate incentive compatibility than those we study. In short, we bound the maximum expected change in utility an agent can induce by misreporting her type (Definition~\ref{def:interim}), no matter her true type, where the expectation is over the other agents' types. In contrast, \citet{Duetting19:Optimal} bound the expected maximum change in utility an agent can induce, where the expectation is over all agents' types. As a result, our bounds are complementary, but not overlapping. Moreover, the sets of mechanisms we analyze and \citet{Duetting19:Optimal} analyze are disjoint: \citet{Duetting19:Optimal} provide bounds for mechanisms represented as neural networks, whereas we instantiate our guarantees for single-item and combinatorial first-price auctions, generalized second-price auction, discriminatory auction, uniform-price auction, and second-price auction with spiteful bidders.

\paragraph{Sample complexity of revenue maximization.} A long line of research has studied revenue maximization via machine learning from a theoretical perspective (including, for example, work by~\citet{Elkind07:Designing,
Balcan08:Reducing},
and \citet{Cole14:Sample}). The mechanism designer receives samples from the type distribution which she uses to find a mechanism that is, ideally, nearly optimal in expectation. That research has only studied incentive compatible mechanism classes.
Moreover, in this paper, it is not enough to provide generalization guarantees; we must both compute our estimate of the incentive compatibility approximation factor and bound our estimate's error. This type of error has nothing to do with revenue functions, but rather utility functions. These factors in conjunction mean that our research differs significantly from prior research on generalization guarantees in mechanism design.

In a related direction, \citet{Chawla14:Mechanism, Chawla16:A, Chawla17:Mechanism} study counterfactual revenue estimation. Given two
auctions, they provide techniques for estimating one of the auction's equilibrium revenue from the other auction's equilibrium bids. They also study social welfare in this context. Thus, their research is tied to selling mechanisms, whereas we study more general mechanism design problems from an application-agnostic perspective.

\paragraph{Incentive compatibility from a buyer's perspective.}
\citet{Lahaie18:Testing} also provide tools for estimating approximate incentive compatibility, but from the buyer's perspective rather than the mechanism designer's perspective. As such, the type of information available to their estimation tools versus ours is different. Moreover, they focus on ad auctions, whereas we study mechanism design in general and apply our techniques to a wide range of settings and mechanisms.

\section{Preliminaries and notation}
There are $n$ agents who each have a \emph{type}. We denote agent $i$'s type as $\theta_i$, which is an element of a (potentially infinite) set $\Theta_i$. A mechanism takes as input the agents' \emph{reported types}, which it uses to choose an outcome. We denote agent $i$'s reported type as $\hat{\theta}_i \in \Theta_i$. We denote all $n$ agents' types as $\vec{\theta} = \left(\theta_1, \dots, \theta_n\right)$ and reported types as $\hat{\vec{\theta}} = \left(\hat{\theta}_1, \dots, \hat{\theta}_n\right)$. For $i \in [n]$, we use the standard notation $\vec{\theta}_{-i} \in \times_{j \not= i}\Theta_j$ to denote all $n-1$ agents' types except agent $i$. Using this notation, we denote the type profile $\vec{\theta}$ representing all $n$ agents' types as $\vec{\theta} = \left(\theta_i, \vec{\theta}_{-i}\right)$. Similarly, $\hat{\vec{\theta}} = \left(\hat{\theta}_i, \hat{\vec{\theta}}_{-i}\right)$. We assume there is a distribution $\dist$ over all $n$ agents' types, and thus the support of $\dist$ is contained in $\times_{i = 1}^n \Theta_i$. We use $\dist|_{\theta_i}$ to denote the conditional distribution given $\theta_i$, so the support of $\dist|_{\theta_i}$ is contained in $\times_{j \not= i}\Theta_j$.
We assume that we can draw samples $\vec{\theta}^{(1)}, \vec{\theta}^{(2)}, \dots$ independently from $\dist$. This is the same assumption made in a long line of work on mechanism design via machine learning (including, for example, work by~\citet{Elkind07:Designing,
Balcan08:Reducing}, and \citet{Cole14:Sample}).
%~\citep{Elkind07:Designing,Balcan08:Reducing,  Cole14:Sample, Medina14:Learning, Huang15:Making, Morgenstern15:Pseudo, Roughgarden15:Ironing,Devanur16:Sample,  Morgenstern16:Learning, Balcan16:Sample,Syrgkanis17:Sample, Gonczarowski17:Efficient, Bubeck17:Online,Medina17:Revenue,Alon17:Submultiplicative,Cai17:Learning, Balcan18:General, Gonczarowski18:Sample}.

Given a mechanism $M$ and agent $i \in [n]$, we use the notation $u_{i,M}\left(\vec{\theta}, \hat{\vec{\theta}}\right)$ to denote the utility agent $i$ receives when the agents have types $\vec{\theta}$ and reported types $\hat{\vec{\theta}}$. We assume\footnote{Our estimation error only increases by a multiplicative factor of $H$ if the range of the utility functions is $[-H,H]$ instead of $[-1,1]$.} it maps to $[-1,1]$. When $\vec{\theta}_{-i} = \hat{\vec{\theta}}_{-i}$, we use  the simplified notation $u_{i, M}\left(\theta_i, \hat{\theta}_i, \vec{\theta}_{-i}\right) := u_{i, M}\left(\left(\theta_i, \vec{\theta}_{-i}\right), \left(\hat{\theta}_i, \vec{\theta}_{-i}\right)\right)$.

At a high level, a mechanism is incentive compatible if no agent can ever increase her utility by misreporting her type. A mechanism is \emph{$\gamma$-incentive compatible} if each agent can increase her utility by an additive factor of at most $\gamma$ by misreporting her type~\citep{Kothari03:Approximately, Archer04:Approximate,Conitzer07:Incremental,Dekel10:Incentive, Lubin12:Approximate, Mennle14:Axiomatic,Dutting2015:Payment, Azevedo18:Strategy, Feng18:Deep, Golowich18:Deep, Duetting19:Optimal}. In the main body of this paper, we concentrate on \emph{ex-interim} approximate incentive compatibility~\citep{Azevedo18:Strategy,Lubin12:Approximate}.
\begin{definition}\label{def:interim}
A mechanism $M$ is \emph{\emph{ex-interim} $\gamma$-incentive compatible} if for each $i\in [n]$ and all $\theta_i, \hat{\theta}_i \in \Theta_i$, agent $i$ with type $\theta_i$ can increase her expected utility by an additive factor of at most $\gamma$ by reporting her type as $\hat{\theta}_i$, so long as the other agents report truthfully. In other words, $\E_{\vec{\theta}_{-i} \sim \dist|_{\theta_i}}\left[u_{i, M}\left(\theta_i, \theta_i, \vec{\theta}_{-i}\right)\right] \geq E_{\vec{\theta}_{-i} \sim \dist|_{\theta_i}}\left[u_{i, M}\left(\theta_i, \hat{\theta}_i, \vec{\theta}_{-i}\right)\right] - \gamma.$
\end{definition}

We do not study \emph{ex-post}\footnote{A mechanism $M$ is \emph{ex-post incentive compatible} if for each $i\in [n]$ and all $\theta_i, \hat{\theta}_i \in \Theta_i$, and all $\vec{\theta}_{-i} \in \times_{j \not=i} \Theta_j$, agent $i$ with type $\theta_i$ cannot increase her utility by reporting her type as $\hat{\theta}_i$, so long as the other agents report truthfully. In other words, $u_{i, M}\left(\theta_i, \theta_i, \vec{\theta}_{-i}\right) \geq u_{i, M}\left(\theta_i, \hat{\theta}_i, \vec{\theta}_{-i}\right).$} or \emph{dominant-strategy}\footnote{A mechanism $M$ is \emph{dominant-strategy incentive compatible} if for each $i\in [n]$ and all $\theta_i, \hat{\theta}_i \in \Theta_i$, agent $i$ with type $\theta_i$ cannot increase her utility by reporting her type as $\hat{\theta}_i$, no matter which types the other agents report. In other words, for all $\vec{\theta}_{-i},\hat{\vec{\theta}}_{-i} \in \times_{j \not=i} \Theta_j$, $u_{i, M}\left(\left(\theta_i, \vec{\theta}_{-i}\right), \left(\theta_i, \hat{\vec{\theta}}_{-i}\right)\right) \geq u_{i, M}\left(\left(\theta_i, \vec{\theta}_{-i}\right), \left(\hat{\theta}_i, \hat{\vec{\theta}}_{-i}\right)\right).$} approximate incentive compatibility because they are worst-case, distribution-independent notions. Therefore, we cannot hope to measure the \emph{ex-post} or dominant-strategy approximation factors using samples from the agents' type distribution.

\section{Estimating approximate ex-interim incentive compatibility}\label{sec:interim}
In this section, we show how to estimate the \emph{ex-interim} incentive compatibility approximation guarantee using data. We assume there is an unknown
distribution $\dist$ over agents' types, and we operate under the common assumption~\citep{Lubin12:Approximate,Azevedo18:Strategy, Cai17:Learning, Yao14:n, Cai16:Duality, Goldner16:Prior, Babaioff17:Menu,  Hart12:Approximate} that the agents' types are independently distributed. In other words, for each agent $i \in [n]$, there exists a distribution $\phi_i$ over their type space $\Theta_i$ such that $\dist = \times_{i = 1}^n \phi_i$. For each agent $i \in [n]$, we receive a set $\sample_{-i}$ of samples independently drawn from $\dist_{-i} = \times_{j \not= i} \phi_j$. For each mechanism $M \in \mclass$, we show how to use the samples to estimate a value $\hat{\gamma}_M$ such that:
\smallskip
\begin{quote}\emph{With probability $1-\delta$ over the draw of the $n$ sets of samples $\sample_{-1}, \dots, \sample_{-n}$, for any agent $i \in [n]$ and all pairs $\theta_i,\hat{\theta}_i \in \Theta_i$, $\E_{\vec{\theta}_{-i} \sim \dist_{-i}}\left[u_{i, M}\left(\theta_i, \hat{\theta}_i, \vec{\theta}_{-i}\right) - u_{i, M}\left(\theta_i, \theta_i, \vec{\theta}_{-i}\right)
\right] \leq \hat{\gamma}_M.$}\end{quote}

To this end, one simple approach, informally, is to estimate $\hat{\gamma}_M$ by measuring the extent to which any agent $i$ with any type $\theta_i$ can improve his utility by misreporting his type, averaged over all profiles in $\sample_{-i}$. In other words, we can estimate $\hat{\gamma}_M$ by solving the following optimization problem:
\begin{equation}\max_{\theta_i, \hat{\theta}_i \in \Theta_i}\left\{ \frac{1}{N} \sum_{j = 1}^N u_{i,M}\left(\theta_i, \hat{\theta}_i, \vec{\theta}_{-i}^{(j)}\right) -u_{i,M}\left(\theta_i, \theta_i, \vec{\theta}_{-i}^{(j)}\right)\right\}.\label{eq:nonconvex}\end{equation} Unfortunately, in full generality, there might not be a finite-time procedure to solve this optimization problem, so in Section~\ref{sec:cover}, we propose more nuanced approaches based on optimizing over finite subsets of $\Theta_i \times \Theta_i$.
As a warm-up and a building block for our main theorems in that section, we prove that with probability $1-\delta$, for all mechanisms $M \in \cM$, \begin{align}&\max_{\theta_i, \hat{\theta}_i \in \Theta_i}\left\{\E_{\vec{\theta}_{-i} \sim \dist_{-i}}\left[ u_{i,M}\left(\theta_i, \hat{\theta}_i, \vec{\theta}_{-i}\right) - u_{i,M}\left(\theta_i, \theta_i, \vec{\theta}_{-i}\right)\right]\right\}\nonumber\\
 \leq &\max_{\theta_i, \hat{\theta}_i \in \Theta_i}\left\{\frac{1}{N} \sum_{j = 1}^N u_{i,M}\left(\theta_i, \hat{\theta}_i, \vec{\theta}_{-i}^{(j)}\right) -u_{i,M}\left(\theta_i, \theta_i, \vec{\theta}_{-i}^{(j)}\right)\right\} + \epsilon_{\mclass}(N, \delta),\label{eq:epsilon}\end{align} where $\epsilon_{\mclass}(N, \delta)$ is an error term that converges to zero as the number $N$ of samples grows. Its convergence rate depends on the \emph{intrinsic complexity} of the utility functions corresponding to the mechanisms in $\cM$, which we formalize using the learning-theoretic tool \emph{pseudo-dimension.} We define pseudo-dimension below for an abstract class $\cA$ of functions mapping a domain $\cX$ to $[-1,1]$.

\begin{definition}[Pseudo-dimension \citep{Pollard84:Convergence}]
Let $\sample = \left\{x^{\left(1\right)}, \dots, x^{\left(N\right)}\right\} \subseteq \cX$ be a set of elements from the domain of the functions in $\cA$ and let $z^{\left(1\right)}, \dots, z^{\left(N\right)} \in \R$ be a set of \emph{targets}. We say that $z^{\left(1\right)}, \dots, z^{\left(N\right)}$ \emph{witness} the shattering of $\sample$ by $\cA$ if for all subsets $T \subseteq \sample$, there exists some function $a_T \in \cA$ such that for all $x^{\left(j\right)} \in T$, $a_T\left(x^{\left(j\right)}\right) \leq z^{\left(j\right)}$ and for all $x^{\left(j\right)} \not\in T$, $a_T\left(x^{\left(j\right)}\right) > z^{\left(j\right)}$. If there exists some vector $\vec{z} \in \R^N$ that witnesses the shattering of $\sample$ by $\cA$, then we say that $\sample$ is \emph{shatterable} by $\cA$. Finally, the pseudo-dimension of $\cA$, denoted $\pdim\left(\cA\right)$, is the size of the largest set that is shatterable by $\cA$.
\end{definition}
Theorem~\ref{thm:pdim} provides an abstract generalization bound in terms of pseudo-dimension.
\begin{theorem}[\citet{Pollard84:Convergence}]\label{thm:pdim}
Let $\Phi$ be a distribution over $\cX$. With probability $1-\delta$ over $x^{(1)}, \dots, x^{(N)} \sim \Phi$, for all $a \in \cA$, $\left|\frac{1}{N} \sum_{j = 1}^N a\left(x^{(j)}\right) - \E_{x \sim \Phi}[a(x)]\right| \leq 2\sqrt{\frac{2 d}{N} \ln \frac{eN}{d}} + 2\sqrt{\frac{1}{2N}\ln\frac{2}{\delta}},$ where $d = \pdim(\cA)$.
\end{theorem}

We now use Theorem~\ref{thm:pdim} to prove that the error term $\epsilon_{\cM}(N, \delta)$ in Equation~\eqref{eq:epsilon} converges to zero as $N$ increases.
To this end, for any mechanism $M$, any agent $i \in [n]$, and any pair of types $\theta_i, \hat{\theta}_i \in \Theta_i$, let $u_{i, M, \theta_i, \hat{\theta}_i} : \times_{j \not= i} \Theta_j \to [-1,1]$ be a function that maps the types $\vec{\theta}_{-i}$ of the other agents to the utility of agent $i$ with type $\theta_i$ and reported type $\hat{\theta}_i$ when the other agents report their types truthfully. In other words,  $u_{i, M, \theta_i, \hat{\theta}_i}\left(\vec{\theta}_{-i}\right) =  u_{i,M}\left(\theta_i, \hat{\theta}_i, \vec{\theta}_{-i}\right)$. Let $\cF_{i, \cM}$ be the set of all such functions defined by mechanisms $M$ from the class $\cM$. In other words, $\cF_{i, \cM} = \left\{\left. u_{i, M, \theta_i, \hat{\theta}_i} \right| \theta_i, \hat{\theta}_i \in \Theta_i, M \in \cM\right\}$. We now analyze the convergence rate of the error term $\epsilon_{\cM}(N, \delta)$. The full proof is in Appendix~\ref{app:interim}.

\begin{restatable}{theorem}{AIC}\label{thm:AIC}
With probability $1-\delta$ over the draw of the $n$ sets $\sample_{-i} = \left\{\vec{\theta}_{-i}^{(1)}, \dots, \vec{\theta}_{-i}^{(N)}\right\} \sim \dist_{-i}^N$, for all mechanisms $M \in \cM$ and agents $i \in [n]$, \begin{align*}&\max_{\theta_i, \hat{\theta}_i \in \Theta_i}\left\{\E_{\vec{\theta}_{-i} \sim \dist_{-i}}\left[ u_{i,M}\left(\theta_i, \hat{\theta}_i, \vec{\theta}_{-i}\right) - u_{i,M}\left(\theta_i, \theta_i, \vec{\theta}_{-i}\right)\right]\right\}\\
 \leq &\max_{\theta_i, \hat{\theta}_i \in \Theta_i}\left\{\frac{1}{N} \sum_{j = 1}^N u_{i,M}\left(\theta_i, \hat{\theta}_i, \vec{\theta}_{-i}^{(j)}\right) -u_{i,M}\left(\theta_i, \theta_i, \vec{\theta}_{-i}^{(j)}\right)\right\} + \epsilon_{\cM}(N, \delta),\end{align*} where $\epsilon_{\cM}(N, \delta) = 4\sqrt{\frac{2 d_i}{N} \ln \frac{eN}{d_i}} + 4\sqrt{\frac{1}{2N}\ln\frac{2n}{\delta}}$ and $d_i = \pdim\left(\cF_{i, \cM}\right)$.
\end{restatable}

\begin{proof}[Proof sketch]
Fix an arbitrary bidder $i \in [n]$. By Theorem~\ref{thm:pdim}, we know that with probability $1- \delta/n$, for every pair of types $\theta_i, \hat{\theta}_i \in \Theta_i$ and every mechanism $M \in \cM$, the expected utility of agent $i$ with type $\theta_i$ and reported type $\hat{\theta}_i$ (in expectation over $\vec{\theta}_{-i} \sim \dist_{-i}$) is $\epsilon$-close to his average utility (averaged over $\vec{\theta}_{-i} \in \sample_{-i}$), where $\epsilon = \tilde{O}\left(\sqrt{d_i/N}\right).$ We use this fact to show that when agent $i$ misreports his type, he cannot increase his expected utility by more than a factor of $2\epsilon + \max_{\theta_i, \hat{\theta}_i \in \Theta_i}\left\{\frac{1}{N} \sum_{j = 1}^N u_{i,M}\left(\theta_i, \hat{\theta}_i, \vec{\theta}_{-i}^{(j)}\right) -u_{i,M}\left(\theta_i, \theta_i, \vec{\theta}_{-i}^{(j)}\right)\right\}$.
\end{proof}

\subsection{Incentive compatibility guarantees via finite covers}\label{sec:cover}
In the previous section, we presented an empirical estimate of the \emph{ex-interim} incentive compatibility approximation factor (Equation~\eqref{eq:nonconvex}) and we showed that it quickly converges to the true approximation factor. However, there may not be a finite-time procedure for computing Equation~\eqref{eq:nonconvex} in its full generality, restated below: \begin{equation}\max_{\theta_i, \hat{\theta}_i \in \Theta_i}\left\{\frac{1}{N} \sum_{j = 1}^N u_{i,M}\left(\theta_i, \hat{\theta}_i, \vec{\theta}_{-i}^{(j)}\right) -u_{i,M}\left(\theta_i, \theta_i, \vec{\theta}_{-i}^{(j)}\right)\right\}.\label{eq:nonconvex2}\end{equation} In this section, we address that challenge. A simple alternative approach 
is to fix a finite \emph{cover} of $\Theta_i \times \Theta_i$, which we denote as $\cG \subset \Theta_i \times \Theta_i$, and approximate Equation~\eqref{eq:nonconvex2} by measuring the extent to which any agent $i$ can improve his utility by misreporting his type when his true and reported types are elements of the cover $\cG$, averaged over all profiles in $\sample_{-i}$. In other words, we estimate Equation~\eqref{eq:nonconvex2} as: \begin{equation}\max_{\left(\theta_i, \hat{\theta}_i\right) \in \cG} \left\{\frac{1}{N} \sum_{j = 1}^N u_{i,M}\left(\theta_i, \hat{\theta}_i, \vec{\theta}_{-i}^{(j)}\right) -u_{i,M}\left(\theta_i, \theta_i, \vec{\theta}_{-i}^{(j)}\right)\right\}.\label{eq:grid}\end{equation} This raises two natural questions: how do we select the cover $\cG$ and how close are the optimal solutions to Equations~\eqref{eq:nonconvex2} and \eqref{eq:grid}? We provide two simple, intuitive approaches to selecting the cover $\cG$. The first is to run a greedy procedure (see Section~\ref{sec:greedy}) and the second is to create a uniform grid over the type space (assuming $\Theta_i = [0,1]^m$ for some integer $m$; see Section~\ref{sec:grid}).

\subsubsection{Covering via a greedy procedure}\label{sec:greedy}
In this section, we show how to construct the cover $\cG$ of $\Theta_i \times \Theta_i$ greedily, based off a classic learning-theoretic algorithm. We then show that when we use the cover $\cG$ to estimate the incentive compatibility approximation factor (via Equation~\eqref{eq:grid}), the estimate quickly converges to the true approximation factor. This greedy procedure is summarized by Algorithm~\ref{alg:greedy}.
\begin{algorithm}[t]
\caption{Greedy cover construction}\label{alg:greedy}
\begin{algorithmic}[1]
   \Require Mechanism $M \in \cM$, set of samples $\sample_{-i} = \left\{\vec{\theta}_{-i}^{(1)}, \dots, \vec{\theta}_{-i}^{(N)}\right\}$, accuracy parameter $\epsilon >0$.
\State Let $U$ be the set of vectors $U \leftarrow \left\{\frac{1}{N} \begin{pmatrix} u_{i,M}\left(\theta_i, \hat{\theta}_i, \vec{\theta}_{-i}^{(1)}\right) - u_{i,M}\left(\theta_i, \theta_i, \vec{\theta}_{-i}^{(1)}\right)\\
\vdots\\
u_{i,M}\left(\theta_i, \hat{\theta}_i, \vec{\theta}_{-i}^{(N)}\right) - u_{i,M}\left(\theta_i, \theta_i, \vec{\theta}_{-i}^{(N)}\right)
\end{pmatrix}: \theta_i, \hat{\theta}_i \in \Theta_i\right\}.$
\State Let $V \leftarrow \emptyset$ and $\cG \leftarrow \emptyset$.
\While{$U \setminus \left(\bigcup_{\vec{v} \in V} B_1\left(\vec{v}, \epsilon\right)\right) \not= \emptyset$}
	\State Select an arbitrary vector $\vec{v}' \in U \setminus \left(\bigcup_{\vec{v} \in V} B_1\left(\vec{v}, \epsilon\right)\right)$.
	\State Let $\theta_i, \hat{\theta}_i \in \Theta_i$ be the types such that $\vec{v}' = \frac{1}{N} \begin{pmatrix} u_{i,M}\left(\theta_i, \hat{\theta}_i, \vec{\theta}_{-i}^{(1)}\right) - u_{i,M}\left(\theta_i, \theta_i, \vec{\theta}_{-i}^{(1)}\right)\\
\vdots\\
u_{i,M}\left(\theta_i, \hat{\theta}_i, \vec{\theta}_{-i}^{(N)}\right) - u_{i,M}\left(\theta_i, \theta_i, \vec{\theta}_{-i}^{(N)}\right)
\end{pmatrix}.$
	\State Add $\vec{v}'$ to $V$ and $\left(\theta_i, \hat{\theta}_i\right)$ to $\cG$.
\EndWhile
 \Ensure The cover $\cG \subseteq \Theta_i \times \Theta_i$.
\end{algorithmic}
\end{algorithm}
For any $\vec{v} \in \R^N$, we use the notation $B_1\left(\vec{v}, \epsilon\right) = \left\{\vec{v}' : \norm{\vec{v}' - \vec{v}}_1 \leq \epsilon\right\}$. To simplify notation, let $U$ be the set of vectors defined in Algorithm~\ref{alg:greedy}: \[U = \left\{\frac{1}{N} \begin{pmatrix} u_{i,M}\left(\theta_i, \hat{\theta}_i, \vec{\theta}_{-i}^{(1)}\right) - u_{i,M}\left(\theta_i, \theta_i, \vec{\theta}_{-i}^{(1)}\right)\\
\vdots\\
u_{i,M}\left(\theta_i, \hat{\theta}_i, \vec{\theta}_{-i}^{(N)}\right) - u_{i,M}\left(\theta_i, \theta_i, \vec{\theta}_{-i}^{(N)}\right)
\end{pmatrix}: \theta_i, \hat{\theta}_i \in \Theta_i\right\}.\] Note that the solution to Equation~\eqref{eq:nonconvex2} equals $\max_{\vec{v} \in U} \sum_{i = 1}^N v[i]$.
The algorithm greedily selects a set of vectors $V \subseteq U$, or equivalently, a set of type pairs $\cG \subseteq \Theta_i \times \Theta_i$ as follows: while $U \setminus \left(\bigcup_{\vec{v} \in V} B_1\left(\vec{v}, \epsilon\right)\right)$ is non-empty, it chooses an arbitrary vector
$\vec{v}'$ in the set, adds it to $V$, and adds the pair 
$\left(\theta_i, \hat{\theta}_i\right) \in \Theta_i \times \Theta_i$ defining the vector $\vec{v}'$ to $\cG$.
 Classic results from learning theory~\citep{Anthony09:Neural} guarantee that this greedy procedure will repeat for at most $\left(8eN/(\epsilon d_i)\right)^{2d_i}$ iterations, where $d_i = \pdim\left(\cF_{i,\cM}\right)$.
 
 We now relate the true incentive compatibility approximation factor to the solution to Equation~\eqref{eq:grid} when the cover is constructed using Algorithm~\ref{alg:greedy}.
 The full proof is in Appendix~\ref{app:interim}.
 
 \begin{restatable}{theorem}{greedy}\label{thm:greedy}
 Given a set $\sample_{-i} = \left\{\vec{\theta}_{-i}^{(1)}, \dots, \vec{\theta}_{-i}^{(N)}\right\}$, a mechanism $M \in \cM$, and accuracy parameter $\epsilon > 0$, let $\cG\left(\sample_{-i}, M, \epsilon\right)$ be the cover returned by Algorithm~\ref{alg:greedy}.
With probability $1-\delta$ over the draw of the $n$ sets $\sample_{-i} \sim \dist_{-i}^N$, for every mechanism $M \in \cM$ and every agent $i \in [n]$, \begin{align}&\max_{\theta_i, \hat{\theta}_i \in \Theta_i}\left\{\E_{\vec{\theta}_{-i} \sim \dist_{-i}}\left[ u_{i,M}\left(\theta_i, \hat{\theta}_i, \vec{\theta}_{-i}\right) - u_{i,M}\left(\theta_i, \theta_i, \vec{\theta}_{-i}\right)\right]\right\}\nonumber\\
 \leq &\max_{\left(\theta_i, \hat{\theta}_i\right) \in \cG\left(\sample_{-i}, M, \epsilon\right)}\left\{\frac{1}{N} \sum_{j = 1}^N u_{i,M}\left(\theta_i, \hat{\theta}_i, \vec{\theta}_{-i}^{(j)}\right) -u_{i,M}\left(\theta_i, \theta_i, \vec{\theta}_{-i}^{(j)}\right)\right\} + \epsilon + \tilde O\left(\sqrt{\frac{d_i}{N}}\right),\label{eq:greedy_inequality}\end{align} where $d_i = \pdim\left(\cF_{i, \cM}\right)$. Moreover, with probability 1, $\left|\cG\left(\sample_{-i}, M, \epsilon\right)\right| \leq \left(8eN/\left(\epsilon d_i\right)\right)^{2d_i}$.
\end{restatable}

\begin{proof}[Proof sketch]
Inequality~\eqref{eq:greedy_inequality} follows from Theorem~\ref{thm:AIC}, the fact that $V$ is an $\epsilon$-cover of $U$ (where $U$ and $V$ are the sets defined in Algorithm~\ref{alg:greedy}), and the equivalence of $\cG\left(\sample_{-i}, M, \epsilon\right)$ and $V$. The bound on $\left|\cG\left(\sample_{-i}, M, \epsilon\right)\right|$ follows from classic learning-theoretic results~\citep{Anthony09:Neural}.
\end{proof}

\subsubsection{Covering via a uniform grid}\label{sec:grid}
The greedy approach in Section~\ref{sec:greedy} is extremely versatile: no matter the type space $\Theta_i$, when we use the resulting cover to estimate the incentive compatibility approximation factor (via Equation~\eqref{eq:grid}), the estimate quickly converges to the true approximation factor. However, implementing the greedy procedure (Algorithm~\ref{alg:greedy}) might be computationally challenging because at each round, it is necessary to check if $U \setminus \left(\bigcup_{\vec{v} \in V} B_1\left(\vec{v}, \epsilon\right)\right)$ is nonempty and if so, select a vector from the set.  In this section, we propose an alternative, extremely simple approach to selecting a cover $\cG$: using a uniform grid over the type space. The efficacy of this approach depends on a ``niceness'' assumption that holds under mild assumptions, as we prove in Section~\ref{sec:applications}. Throughout this section, we assume that $\Theta_i = [0,1]^m$ for some integer $m$. We propose covering the type space using a \emph{$w$-grid} $\cG_w$ over $[0,1]^m$, by which we mean a finite set of vectors in $[0,1]^m$ such that for all $\vec{p} \in [0,1]^m$, there exists a vector $\vec{p}' \in \cG_w$ such that $\norm{\vec{p} - \vec{p}'}_1 \leq w$. For example, if $m = 1$, we could define $\cG_w = \left\{0, \frac{1}{\lfloor 1/ w\rfloor}, \frac{2}{\lfloor 1/ w\rfloor}, \dots, 1\right\}$. We will estimate the expected incentive compatibility approximation factor using Equation~\eqref{eq:grid} with $\cG = \cG_w \times \cG_w$. Throughout the rest of the paper, we discuss how the choice of $w$ effects the error bound. To do so, we will use the notion of \emph{dispersion}, defined below.
\begin{definition}[\citet{Balcan18:Dispersion}]
  \label{def:dispersion}
  Let $a_1, \dots, a_N : \R^d \to \R$ be a set of functions where
  each $a_i$ is piecewise Lipschitz with respect to the $\ell_1$-norm over a partition $\cP_i$ of $\R^d$. We
  say that $\cP_i$ splits a set $A \subseteq \R^d$ if $A$ intersects with at least two
  sets in $\cP_i$. The set of functions is
  \emph{$(w,k)$-dispersed} if for every point $\vec{p} \in \R^d$, the ball $\left\{\vec{p}' \in
\R^d : \norm{\vec{p}-\vec{p}'}_1 \leq w\right\}$ is split by at most $k$
  of the partitions $\cP_1, \dots, \cP_N$.
\end{definition}

The smaller $w$ is and the larger $k$ is, the more ``dispersed'' the functions' discontinuities are. Moreover, the more jump discontinuities a set of functions has, the more difficult it is to approximately optimize its average using a grid. We illustrate this phenomenon in Example~\ref{ex:FP}.

\begin{example}\label{ex:FP}
Suppose there are two agents and $M$ is the first-price single-item auction. For any trio of types $\theta_1, \theta_2, \hat{\theta}_1 \in [0,1]$, $u_{1,M}\left(\theta_1, \hat{\theta}_1, \theta_2\right) = \textbf{1}_{\left\{\hat{\theta}_1 > \theta_2\right\}}\left(\theta_1 - \hat{\theta}_1\right)$. Suppose $\theta_1 = 1$. Figure~\ref{fig:good_apx} displays the function $\frac{1}{4}\sum_{\theta_2 \in \sample_2}u_{1,M}\left(1,\hat{\theta}_1, \theta_2\right)$ where $\sample_2 = \left\{\frac{1}{5}, \frac{2}{5}, \frac{3}{5}, \frac{4}{5}\right\}$ and Figure~\ref{fig:bad_apx} displays the function $\frac{1}{4}\sum_{\theta_2 \in \sample_2'}u_{1,M}\left(1, \hat{\theta}_1, \theta_2\right)$ where $\sample_2' = \left\{\frac{31}{40}, \frac{32}{40}, \frac{33}{40}, \frac{34}{40}\right\}$.

\begin{figure}
\centering
\begin{subfigure}{0.47\textwidth}
\includegraphics{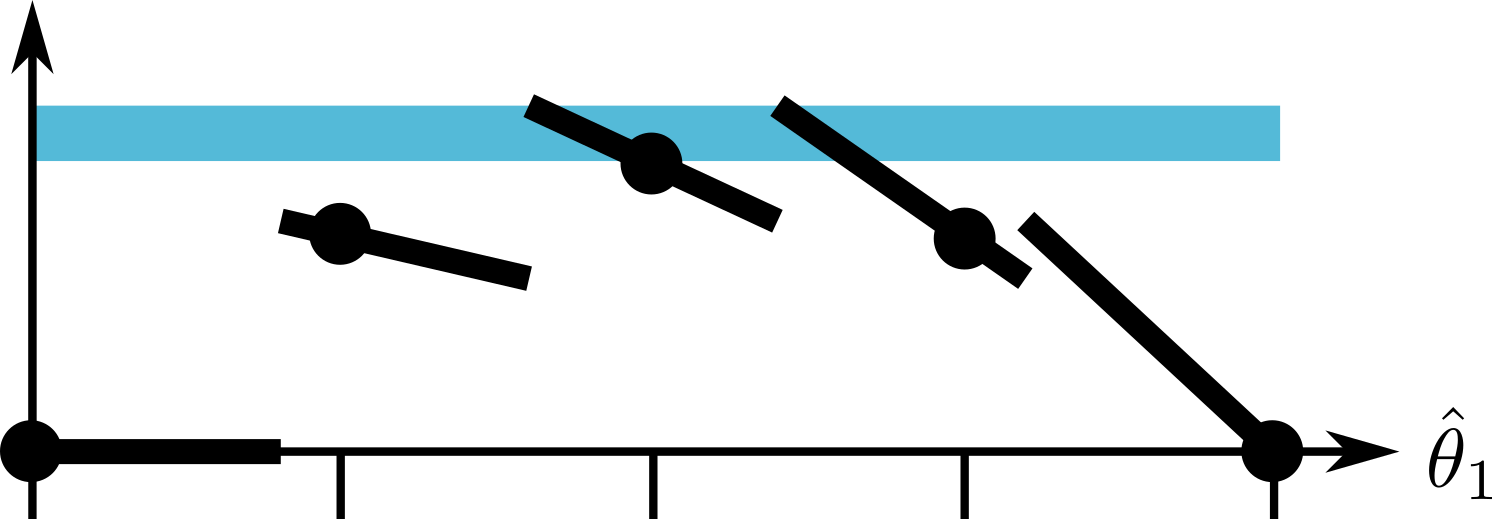} \centering
\caption{Average of dispersed functions.}
\label{fig:good_apx}
\end{subfigure}\qquad
\begin{subfigure}{0.47\textwidth}
\includegraphics{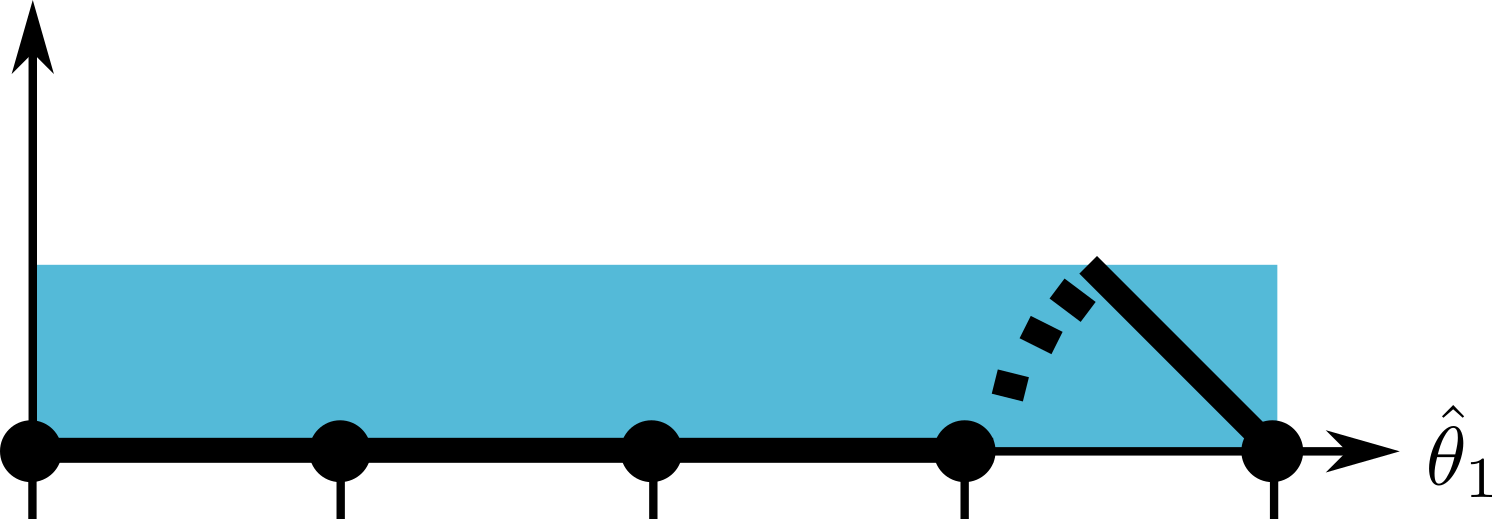}\centering
\caption{Average of functions that are not dispersed.}
\label{fig:bad_apx}
\end{subfigure}\qquad
\caption{Illustration of Example~\ref{ex:FP}. The lines in Figure~\ref{fig:good_apx} depict the average of four utility functions corresponding to the first-price auction. The maximum of this average falls at the top of the blue region. We evaluate the function in Figure~\ref{fig:good_apx} on the grid $\left\{0, \frac{1}{4}, \frac{1}{2}, \frac{3}{4}, 1\right\}$, as depicted by the dots. The maximum over the grid falls at the bottom of the blue region. In Figure~\ref{fig:bad_apx}, we illustrate the same concepts but for a different set of utility functions that are not as dispersed as those in Figure~\ref{fig:good_apx}. As illustrated by the blue regions, the approximation over the grid is better for the dispersed functions than the non-dispersed functions.} 
\label{fig:apx_comparison}
\end{figure}
In Figure~\ref{fig:apx_comparison}, we evaluate each function on the grid $\cG = \left\{0, \frac{1}{4}, \frac{1}{2}, \frac{3}{4}, 1\right\}$. In Figure~\ref{fig:good_apx}, the maximum over $\cG$ better approximates the the maximum over $[0,1]$ compared to Figure~\ref{fig:bad_apx}. In other words, \begin{align}&\max_{\hat{\theta}_1 \in [0,1]}\sum_{\theta_2 \in \sample_2}u_{1,M}\left(1, \hat{\theta}_1, \theta_2\right) - \max_{\hat{\theta}_1 \in \cG}\sum_{\theta_2 \in \sample_2}u_{1,M}\left(1, \hat{\theta}_1, \theta_2\right)\label{eq:good_difference}\\
< &\max_{\hat{\theta}_1 \in [0,1]}\sum_{\theta_2 \in \sample_2'}u_{1,M}\left(1, \hat{\theta}_1, \theta_2\right) - \max_{\hat{\theta}_1 \in \cG}\sum_{\theta_2 \in \sample_2'}u_{1,M}\left(1, \hat{\theta}_1, \theta_2\right).\label{eq:bad_difference}\end{align} Intuitively, this is because the functions we average over in Figure~\ref{fig:good_apx} are more dispersed than the functions we average over in Figure~\ref{fig:bad_apx}. The differences described by Equations~\eqref{eq:good_difference} and \eqref{eq:bad_difference} are represented by the shaded regions in Figures~\ref{fig:good_apx} and \ref{fig:bad_apx}, respectively.
\end{example}

We now state a helpful lemma which we use to prove this section's main theorem. Informally, it shows that we can measure the average amount that any agent can improve his utility by misreporting his type, even if we discretize his type space, so long as his utility function applied to the samples demonstrates dispersion. The full proof is in Appendix~\ref{app:interim}.

\begin{restatable}{lemma}{dispersion}\label{lem:grid_dispersion_prob}
Suppose that for each agent $i \in [n]$, there exist $L_i, k_i, w_i \in \R$ such that with probability $1-\delta$ over the draw of the $n$ sets $\sample_{-i} = \left\{\vec{\theta}^{(1)}_{-i}, \dots, \vec{\theta}^{(N)}_{-i}\right\} \sim \dist_{-i}^N$, for each mechanism $M \in \cM$ and agent $i \in [n]$, the following conditions hold:
\begin{enumerate}
\item For any type $\vec{\theta}_i \in [0,1]^m$, the functions $u_{i, M}\left(\vec{\theta}_i, \cdot, \vec{\theta}_{-i}^{(1)}\right), \dots, u_{i, M}\left(\vec{\theta}_i, \cdot, \vec{\theta}_{-i}^{(N)}\right)$ are piecewise $L_i$-Lipschitz and $(w_i,k_i)$-dispersed.
\item For any reported type $\hat{\vec{\theta}}_i \in [0,1]^m$, the functions $u_{i, M}\left(\cdot, \hat{\vec{\theta}}_i, \vec{\theta}_{-i}^{(1)}\right), \dots, u_{i, M}\left(\cdot, \hat{\vec{\theta}}_i, \vec{\theta}_{-i}^{(N)}\right)$ are piecewise $L_i$-Lipschitz and $(w_i,k_i)$-dispersed.
\end{enumerate}
Then with probability $1-\delta$ over the draw of the $n$ sets $\sample_{-i} \sim \dist_{-i}^N$, for all mechanisms $M \in \cM$ and agents $i \in [n]$, \begin{align}\max_{\vec{\theta}_i, \hat{\vec{\theta}}_i \in [0,1]^m}&\left\{\frac{1}{N} \sum_{j = 1}^N u_{i,M}\left(\vec{\theta}_i,\hat{\vec{\theta}}_i, \vec{\theta}_{-i}^{(j)}\right)
 - u_{i,M}\left(\vec{\theta}_i, \vec{\theta}_i, \vec{\theta}_{-i}^{(j)}\right)\right\}\nonumber\\
\leq \max_{\vec{\theta}_i, \hat{\vec{\theta}}_i \in \cG_{w_i}}&\left\{\frac{1}{N} \sum_{j = 1}^N u_{i,M}\left(\vec{\theta}_i,\hat{\vec{\theta}}_i, \vec{\theta}_{-i}^{(j)}\right)
 - u_{i,M}\left(\vec{\theta}_i, \vec{\theta}_i, \vec{\theta}_{-i}^{(j)}\right)\right\} + 4L_iw_i + \frac{8k_i}{N}\label{eq:grid_dispersion_prob}.\end{align}
\end{restatable}

\begin{proof}[Proof sketch]
By definition of dispersion, the following conditions hold with probability $1-\delta$:

\paragraph{Condition 1.} For all mechanisms $M \in \cM$, agents $i \in [n]$, types $\vec{\theta}_i \in [0,1]^m$, and pairs of reported types $\hat{\vec{\theta}}_i, \hat{\vec{\theta}}_i' \in [0,1]^m$, if $\norm{\hat{\vec{\theta}}_i - \hat{\vec{\theta}}_i'}_1 \leq w_i$, then $\left|\frac{1}{N}\sum_{j = 1}^N u_{i, M}\left(\vec{\theta}_i, \hat{\vec{\theta}}_i, \vec{\theta}_{-i}^{(j)}\right) - u_{i, M}\left(\vec{\theta}_i, \hat{\vec{\theta}}_i', \vec{\theta}_{-i}^{(j)}\right)\right| \leq L_iw_i + \frac{2k_i}{N}.$

\paragraph{Condition 2.} For all mechanisms $M \in \cM$, agents $i \in [n]$, reported types $\hat{\vec{\theta}}_i \in [0,1]^m$, and pairs of types $\vec{\theta}_i, \vec{\theta}_i' \in [0,1]^m$, if $\norm{\vec{\theta}_i - \vec{\theta}_i'}_1 \leq w_i$, then $\left|\frac{1}{N}\sum_{j = 1}^N u_{i, M}\left(\vec{\theta}_i, \hat{\vec{\theta}}_i, \vec{\theta}_{-i}^{(j)}\right) - u_{i, M}\left(\vec{\theta}_i', \hat{\vec{\theta}}_i, \vec{\theta}_{-i}^{(j)}\right)\right| \leq L_iw_i + \frac{2k_i}{N}.$

We claim that Inequality~\eqref{eq:grid_dispersion_prob} holds so long as Conditions 1 and 2 hold.
To see why, suppose they do both hold, and
fix an arbitrary agent $i \in [n]$, mechanism $M \in \cM$, and pair of types $\vec{\theta}_i, \hat{\vec{\theta}}_i \in [0,1]^m$. Consider the average amount agent $i$ with type $\vec{\theta}_i$ can improve his utility by misreporting his type as $\hat{\vec{\theta}}_i$: \begin{equation}\frac{1}{N} \sum_{j = 1}^N u_{i,M}\left(\vec{\theta}_i,\hat{\vec{\theta}}_i, \vec{\theta}_{-i}^{(j)}\right)
 - u_{i,M}\left(\vec{\theta}_i, \vec{\theta}_i, \vec{\theta}_{-i}^{(j)}\right).\label{eq:fixed_types}\end{equation} By definition of the grid $\cG_{w_i}$, we know there are points $\vec{p}, \hat{\vec{p}} \in \cG_{w_i}$ such that $\norm{\vec{\theta}_i - \vec{p}}_1 \leq w_i$ and $\norm{\hat{\vec{\theta}}_i - \hat{\vec{p}}}_1 \leq w_i$. Based on Conditions 1 and 2, we show that if we ``snap'' $\vec{\theta}_i$ to $\vec{p}$ and $\hat{\vec{\theta}}_i$ to $\hat{\vec{p}}$, we will not increase Equation~\eqref{eq:fixed_types} by more than an additive factor of $O\left(L_iw_i + k_i/N\right)$.
Since $\vec{p}$ and $\hat{\vec{p}}$ are elements of $\cG_{w_i}$, Inequality~\eqref{eq:grid_dispersion_prob} holds so long as Conditions 1 and 2 hold. Since Conditions 1 and 2 hold with probability $1-\delta$, the lemma statement holds.
\end{proof}

In Section~\ref{sec:applications}, we prove that under mild assumptions, for a wide range of mechanisms, Conditions 1 and 2 in Lemma~\ref{lem:grid_dispersion_prob} hold with $L_i = 1$, $w_i = O \left(1 / \sqrt{N}\right)$, and $k_i = \tilde O \left(\sqrt{N}\right)$, ignoring problem-specific multiplicands. Thus, we find that $4L_iw_i + 8k_i/N$ quickly converges to 0 as $N$ grows.

Theorem~\ref{thm:AIC} and Lemma~\ref{lem:grid_dispersion_prob} immediately imply this section's main theorem. At a high level, it states that any agent's average utility gain when restricted to types on the grid is a close estimation of the true incentive compatibility approximation factor, so long as his utility function applied to the samples demonstrates dispersion. The full proof is in Appendix~\ref{app:interim}.

\begin{restatable}{theorem}{main}\label{thm:main_prob}
Suppose that for each agent $i \in [n]$, there exist $L_i, k_i, w_i \in \R$ such that with probability $1-\delta$ over the draw of the $n$ sets $\sample_{-i} = \left\{\vec{\theta}^{(1)}_{-i}, \dots, \vec{\theta}^{(N)}_{-i}\right\} \sim \dist_{-i}^N$, for each mechanism $M \in \cM$ and agent $i \in [n]$, the following conditions hold:
\begin{enumerate}
\item For any $\vec{\theta}_i \in [0,1]^m$, the functions $u_{i, M}\left(\vec{\theta}_i, \cdot, \vec{\theta}_{-i}^{(1)}\right), \dots, u_{i, M}\left(\vec{\theta}_i, \cdot, \vec{\theta}_{-i}^{(N)}\right)$ are piecewise $L_i$-Lipschitz and $(w_i,k_i)$-dispersed.
\item For any $\hat{\vec{\theta}}_i \in [0,1]^m$, the functions $u_{i, M}\left(\cdot, \hat{\vec{\theta}}_i, \vec{\theta}_{-i}^{(1)}\right), \dots, u_{i, M}\left(\cdot, \hat{\vec{\theta}}_i, \vec{\theta}_{-i}^{(N)}\right)$ are piecewise $L_i$-Lipschitz and $(w_i,k_i)$-dispersed.
\end{enumerate} Then with probability $1-2\delta$, for every mechanism $M \in \cM$ and every agent $i \in [n]$, \begin{align*}&\max_{\vec{\theta}_i, \hat{\vec{\theta}}_i \in \Theta_i}\left\{\E_{\vec{\theta}_{-i} \sim \dist_{-i}}\left[u_{i,M}\left(\vec{\theta}_i, \hat{\vec{\theta}}_i, \vec{\theta}_{-i}\right) - u_{i,M}\left(\vec{\theta}_i,\vec{\theta}_i, \vec{\theta}_{-i}\right)\right]\right\}\\
 \leq &\max_{\vec{\theta}_i, \hat{\vec{\theta}}_i \in \cG_{w_i}}\left\{\frac{1}{N} \sum_{j = 1}^N u_{i,M}\left(\vec{\theta}_i,\hat{\vec{\theta}}_i, \vec{\theta}_{-i}^{(j)}\right)
 - u_{i,M}\left(\vec{\theta}_i, \vec{\theta}_i, \vec{\theta}_{-i}^{(j)}\right)\right\} + \epsilon,\end{align*} where $\epsilon = 4L_iw_i + \frac{8k_i}{N} + 4\sqrt{\frac{2 d_i}{N} \ln \frac{eN}{d_i}} + 4\sqrt{\frac{1}{2N}\ln\frac{2n}{\delta}}$ and $d_i = \pdim\left(\cF_{i, \cM}\right)$.
\end{restatable}

\begin{proof}[Proof sketch]
Theorem~\ref{thm:AIC} guarantees that with probability at most $\delta$, the extent to which any agent $i$ can improve his utility by misreporting his type, averaged over all profiles in $\sample_{-i}$, does not approximate the true incentive compatibility approximation factor, as summarized below:

\paragraph{Bad event 1.} For some mechanism $M \in \cM$ and agent $i \in [n]$,\begin{align}&\max_{\vec{\theta}_i, \hat{\vec{\theta}}_i \in \Theta_i}\left\{\E_{\vec{\theta}_{-i} \sim \dist_{-i}}\left[u_{i,M}\left(\vec{\theta}_i, \hat{\vec{\theta}}_i, \vec{\theta}_{-i}\right) - u_{i,M}\left(\vec{\theta}_i,\vec{\theta}_i, \vec{\theta}_{-i}\right)\right]\right\}\nonumber\\
 > &\max_{\vec{\theta}_i, \hat{\vec{\theta}}_i \in [0,1]^m}\left\{\frac{1}{N} \sum_{j = 1}^N u_{i,M}\left(\vec{\theta}_i,\hat{\vec{\theta}}_i, \vec{\theta}_{-i}^{(j)}\right)
 - u_{i,M}\left(\vec{\theta}_i, \vec{\theta}_i, \vec{\theta}_{-i}^{(j)}\right)\right\} + \epsilon_{\cM}(N, \delta),\label{eq:full_type}\end{align} where $\epsilon_{\cM}(N, \delta) = 4\sqrt{\frac{2 d_i}{N} \ln \frac{eN}{d_i}} + 4\sqrt{\frac{1}{2N}\ln\frac{2n}{\delta}}$ and $d_i = \pdim\left(\cF_{i, \cM}\right)$.
 
By Lemma~\ref{lem:grid_dispersion_prob}, we also know that with probability at most $\delta$, we cannot approximate Equation~\eqref{eq:full_type} by discretizing the agent's type space, as summarized by the following bad event:
 
 \paragraph{Bad event 2.} For some mechanism $M \in \cM$ and agent $i \in [n]$, \begin{align}\max_{\vec{\theta}_i, \hat{\vec{\theta}}_i \in [0,1]^m}&\left\{\frac{1}{N} \sum_{j = 1}^N u_{i,M}\left(\vec{\theta}_i,\hat{\vec{\theta}}_i, \vec{\theta}_{-i}^{(j)}\right)
 - u_{i,M}\left(\vec{\theta}_i, \vec{\theta}_i, \vec{\theta}_{-i}^{(j)}\right)\right\}\nonumber\\
> \max_{\vec{\theta}_i, \hat{\vec{\theta}}_i \in \cG_{w_i}}&\left\{\frac{1}{N} \sum_{j = 1}^N u_{i,M}\left(\vec{\theta}_i,\hat{\vec{\theta}}_i, \vec{\theta}_{-i}^{(j)}\right)
 - u_{i,M}\left(\vec{\theta}_i, \vec{\theta}_i, \vec{\theta}_{-i}^{(j)}\right)\right\} + 4L_iw_i + \frac{8k_i}{N}.\label{eq:discretize_proof}\end{align}

With probability $1-2\delta$, neither bad event occurs, so Equation~\eqref{eq:discretize_proof} approximates Equation~\eqref{eq:full_type}.
\end{proof}

In Appendix~\ref{app:interim}, we show that if, given a set of samples $\sample_{-i}$, one can measure the dispersion parameters $L_i, k_i, w_i \in \R$ such that Conditions 1 and 2 in Theorem~\ref{thm:main_prob} hold, then the theorem's inequality holds.
%In Section~\ref{sec:experiments}, we demonstrate this technique via experiments.

We conclude with one final comparison of the greedy approach in Section~\ref{sec:greedy} and the uniform grid approach in this section. In Section~\ref{sec:greedy}, we use the functional form of the utility functions in order to bound the cover size, as quantified by pseudo-dimension. When we use the approach based on dispersion, we do not use these functional forms to the fullest extent possible; we only use simple facts about the functions' discontinuities and Lipschitzness.
 
\subsection{Dispersion and pseudo-dimension guarantees}\label{sec:applications}
\begin{table}
						\begin{tabularx}{\textwidth}{L{1} L{1}}
\toprule
Mechanism(s) & Upper bound on estimation error $\left|\hat{\gamma} - \gamma\right|$ based on dispersion \\\midrule
First-price single-item auction & $\tilde{O}\left(\frac{\kappa^{-1} + n}{\sqrt{N}}\right)$
\\\midrule
First-price combinatorial auction over $\ell$ items & $\tilde{O}\left(\frac{\kappa^{-1} + (n+1)^{2\ell}\sqrt{\ell}}{\sqrt{N}}\right)$\\\midrule
Generalized second-price auction & $\tilde{O}\left(\frac{\kappa^{-1} + n^{3/2}}{\sqrt{N}}\right)$\\\midrule
Discriminatory auction over $m$ units of a single good & $\tilde{O}\left(\frac{\kappa^{-1} + nm^2}{\sqrt{N}}\right)$\\\midrule
Uniform-price auction over $m$ units of a single good & $\tilde{O}\left(\frac{\kappa^{-1} + nm^2}{\sqrt{N}}\right)$\\\midrule
Second-price auction with spiteful bidders & $\tilde{O}\left(\frac{\kappa^{-1} + n}{\sqrt{N}}\right)$\\\bottomrule
			\end{tabularx}
			\caption{Our estimation error upper bounds based on dispersion. The value $\kappa$ is defined such that $[0, \kappa]$ contains the range of the density functions defining the agents' type distribution.}\label{tab:results}
\end{table}
We now provide dispersion and pseudo-dimension guarantees for a variety of mechanism classes. The theorems in this section allow us to instantiate the bounds from the previous section and thus understand how well our empirical incentive compatibility approximation factor matches the true approximation factor. We summarize our guarantees in Table~\ref{tab:results}.

Given an agent $i \in [n]$, a true type $\vec{\theta}_i \in [0,1]^m$, and a set of samples $\vec{\theta}_{-i}^{(1)}, \dots, \vec{\theta}_{-i}^{(N)} \sim \dist_{-i}$, we often find that the discontinuities of each function $u_{i, M}\left(\vec{\theta}_i, \cdot, \vec{\theta}_{-i}^{(j)}\right)$ are highly dependent on the vector $\vec{\theta}_{-i}^{(j)}$. For example, under the first-price single-item auction, the discontinuities of the function $u_{i, M}\left(\theta_i, \cdot, \vec{\theta}_{-i}^{(j)}\right)$ occur when $\hat{\theta}_i = \norm{\vec{\theta}_{-i}^{(j)}}_{\infty}$. Since each vector $\vec{\theta}_{-i}^{(j)}$ is a random draw from the distribution $\dist_{-i}$, these functions will not be dispersed if these random draws are highly-concentrated. For this reason, we focus on type distributions with \emph{$\kappa$-bounded density functions}. A
density function $\phi:\R \to \R$ is
$\kappa$-bounded if $\max\{\phi(x)\} \leq \kappa$. For example, the density function of any normal distribution with standard deviation $\sigma \in \R$ is $\frac{1}{\sigma\sqrt{2\pi}}$-bounded.

The challenge we face in this section is that it is not enough to show that for some $\vec{\theta}_i \in [0,1]^m$, the functions $u_{i, M}\left(\vec{\theta}_i, \cdot, \vec{\theta}_{-i}^{(1)}\right), \dots, u_{i, M}\left(\vec{\theta}_i, \cdot, \vec{\theta}_{-i}^{(N)}\right)$ are dispersed. Rather, we must prove that for all type vectors, the dispersion property holds. This facet of our analysis is notably different from prior work by \citet{Balcan18:Dispersion}: in their applications, it is enough to show that with high probability, a single, finite sequence of functions is dispersed. In contrast, we show that under mild assumptions, with high probability, each function sequence from an infinite family is dispersed.

In a bit more detail, there are two types of sequences we must prove are dispersed. The first are defined by functions over reported types $\hat{\vec{\theta}}_i$, with one sequence for each true type $\vec{\theta}_i$; the second are defined by functions over true types $\vec{\theta}_i$ with one sequence for each reported type $\hat{\vec{\theta}}_i$. In other words, sequences of the first type have the form $u_{i, M}\left(\vec{\theta}_i, \cdot, \vec{\theta}_{-i}^{(1)}\right), \dots, u_{i, M}\left(\vec{\theta}_i, \cdot, \vec{\theta}_{-i}^{(N)}\right)$ and sequences of the second type have the form $u_{i, M}\left(\cdot, \hat{\vec{\theta}}_i, \vec{\theta}_{-i}^{(1)}\right), \dots, u_{i, M}\left(\cdot, \hat{\vec{\theta}}_i, \vec{\theta}_{-i}^{(N)}\right)$. For the first family of sequences, we show that discontinuities occur only when agent $i$'s shift in her reported type causes the allocation to change. These change-points have nothing to do with the true type $\vec{\theta}_i$ (which the mechanism does not see), so all the sequences have the same discontinuities (see, for example in Lemma~\ref{lem:FP_discontinuities}). As a result, it is enough prove dispersion for a single generic sequence, defined by an arbitrary $\vec{\theta}_i$---as we do, for example, in Theorem~\ref{thm:FP_disp}---in order to guarantee dispersion for all of the sequences. Throughout our applications, the second type of sequence is even simpler.  These functions have no discontinuities at all, because the allocation is invariant as we vary the true type $\vec{\theta}_i$ (because, again, the mechanism does not see $\vec{\theta}_i$). Therefore, these functions are either constant if the player was not allocated anything, or linear otherwise (see, for example, Theorem~\ref{thm:FP_disp2}).  As a result, all sequences from this second family are immediately dispersed.

\subsubsection{First-price single-item auction}
To develop intuition, we begin by analyzing the first-price auction for a single item. We then analyze the first-price combinatorial auction in Section~\ref{sec:FP_combinatorial}.
Under this auction, each agent $i \in [n]$ has a value $\theta_i \in [0,1]$ for the item and submits a bid $\hat{\theta}_i \in [0,1]$. The agent with the highest bid wins the item and pays his bid. Therefore, agent $i$'s utility function is $u_{i, M}\left(\vec{\theta}, \hat{\vec{\theta}}\right) = \textbf{1}_{\left\{\hat{\theta}_i > \norm{\hat{\vec{\theta}}_{-i}}_{\infty}\right\}} \left(\theta_i - \hat{\theta}_i\right).$

To prove our dispersion guarantees, we begin with the following helpful lemma.

\begin{lemma}\label{lem:FP_discontinuities}
For any agent $i \in [n]$, any set $\sample_{-i} = \left\{\vec{\theta}^{\left(1\right)}_{-i}, \dots, \vec{\theta}^{\left(N\right)}_{-i}\right\}$, and any type $\theta_i \in [0,1]$, the functions $u_{i, M}\left(\theta_i, \cdot, \vec{\theta}_{-i}^{(1)}\right), \dots, u_{i, M}\left(\theta_i, \cdot, \vec{\theta}_{-i}^{(N)}\right)$ are piecewise 1-Lipschitz and their discontinuities fall in the set $\left\{\norm{\vec{\theta}_{-i}^{(j)}}_{\infty}\right\}_{j \in [N]}$.
\end{lemma}

\begin{proof}
Suppose $i = 1$ and choose an arbitrary sample $j \in [N]$. For any value $\theta_1 \in [0,1]$ and bid $\hat{\theta}_1 \in [0,1]$, we know that $u_{1, M}\left(\theta_1, \hat{\theta}_1, \vec{\theta}_{-1}^{(j)}\right) = \textbf{1}_{\left\{\hat{\theta}_1 > \norm{\vec{\theta}_{-1}^{(j)}}_{\infty}\right\}} \left(\theta_1 - \hat{\theta}_1\right)$, as illustrated in Figure~\ref{fig:FP}.
\begin{figure}
  \begin{subfigure}[t]{0.45\textwidth}
    \centering
    \includegraphics{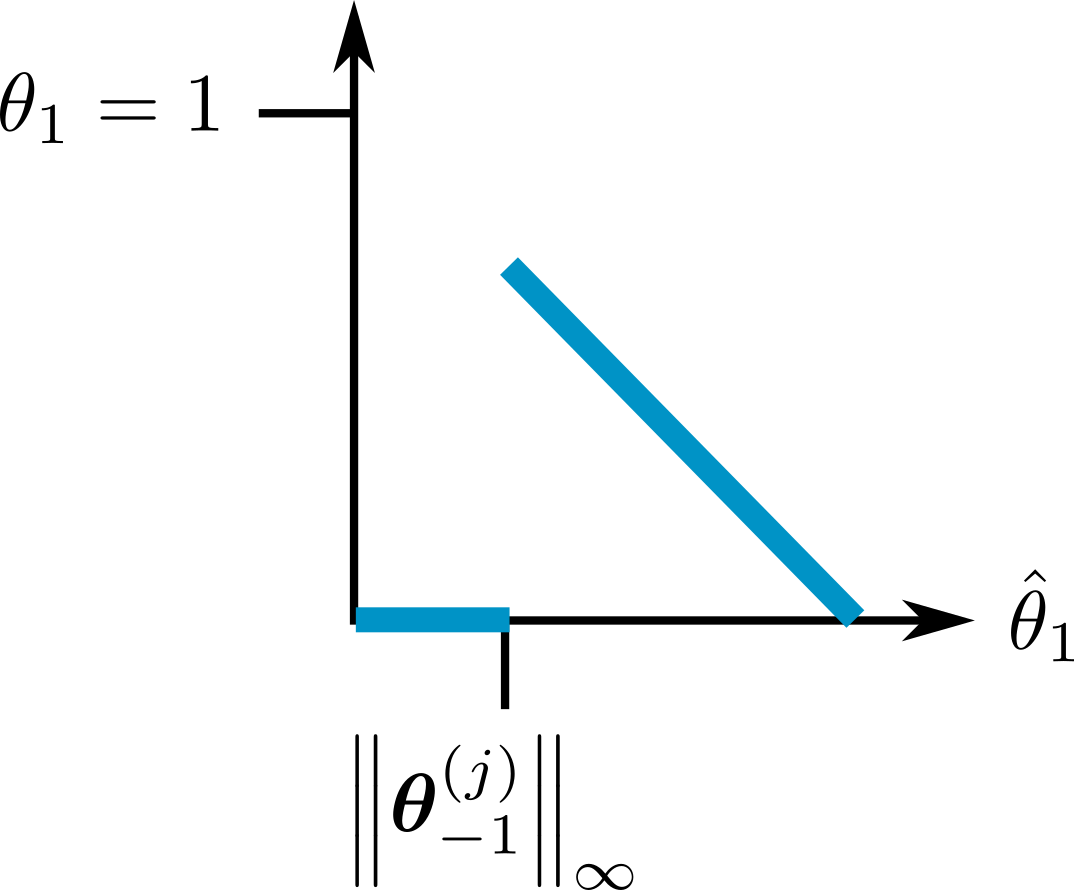}
    \caption{The function $u_{1, M}\left(1, \cdot, \vec{\theta}_{-1}^{(j)}\right)$.}
    \label{fig:FP1}
  \end{subfigure}
  \hfill
  \begin{subfigure}[t]{0.45\textwidth}
    \centering
    \includegraphics{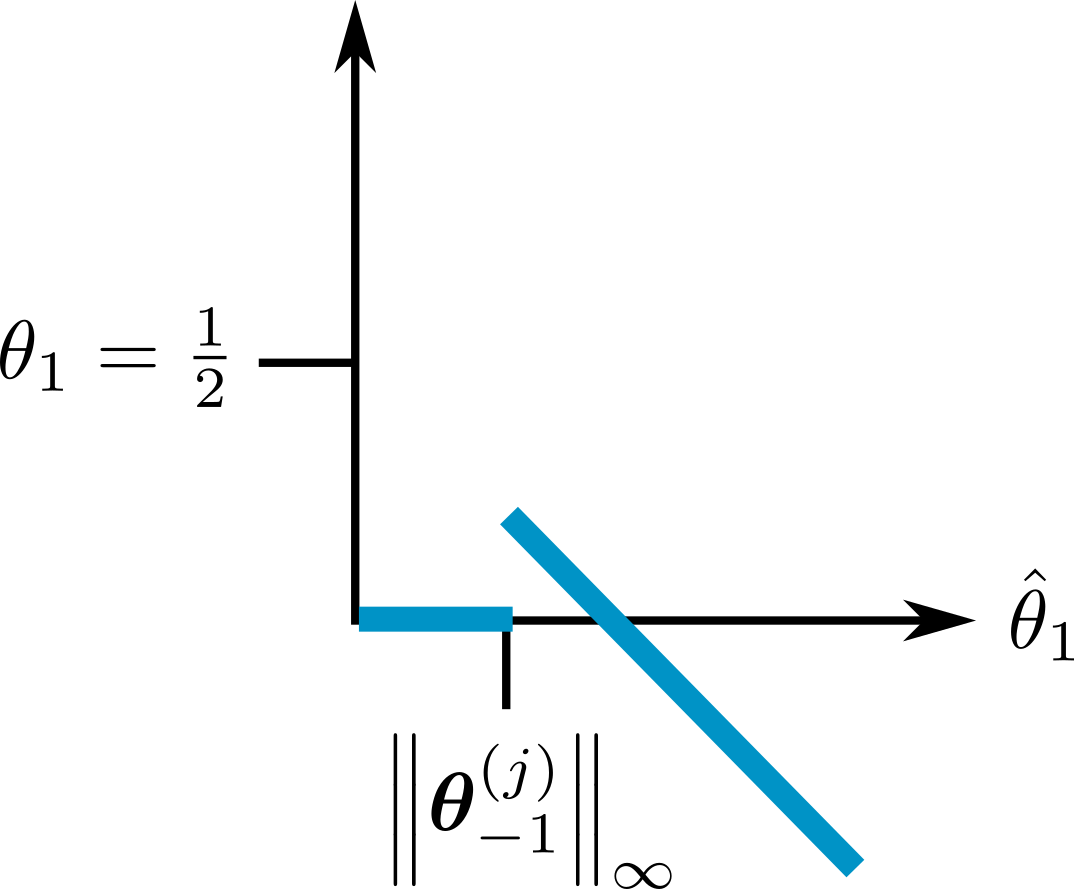}
    \caption{The function $u_{1, M}\left(\frac{1}{2}, \cdot, \vec{\theta}_{-1}^{(j)}\right)$.}
    \label{fig:FP2}
  \end{subfigure}
  \caption{Agent 1's utility function under a first-price single-item auction when the other agents' values and bids are defined by $\vec{\theta}_{-1}^{(j)} \in [0,1]^{n-1}$. The utility is  $u_{1, M}\left(\theta_1, \hat{\theta}_1, \vec{\theta}_{-1}^{(j)}\right) = \textbf{1}_{\left\{\hat{\theta}_1 > \norm{\vec{\theta}_{-1}^{(j)}}_{\infty}\right\}} \left(\theta_1 - \hat{\theta}_1\right)$.}\label{fig:FP}
\end{figure}
Therefore, no matter the value of $\theta_1$, the function $u_{1, M}\left(\theta_1, \cdot, \vec{\theta}_{-1}^{(j)}\right)$ is piecewise 1-Lipschitz with a discontinuity at $\norm{\vec{\theta}_{-1}^{(j)}}_{\infty}$.
\end{proof}

We now use Lemma~\ref{lem:FP_discontinuities} to prove our first dispersion guarantee. The full proof is in Appendix~\ref{app:applications}.

\begin{restatable}{theorem}{FP}\label{thm:FP_disp}
Suppose each agent's type has a $\kappa$-bounded density function.
With probability $1-\delta$ over the draw of the $n$ sets $\sample_{-i} = \left\{\vec{\theta}^{(1)}_{-i}, \dots, \vec{\theta}^{(N)}_{-i}\right\} \sim \dist_{-i}^N$, we have that for all agents $i \in [n]$ and types $\theta_i \in [0,1]$, the functions $u_{i, M}\left(\theta_i, \cdot, \vec{\theta}_{-i}^{(1)}\right), \dots, u_{i, M}\left(\theta_i, \cdot, \vec{\theta}_{-i}^{(N)}\right)$ are piecewise $1$-Lipschitz and $\left(O\left(1/\left(\kappa \sqrt{N}\right)\right),\tilde O\left(n \sqrt{N}\right)\right)$-dispersed.
\end{restatable}

\begin{proof}[Proof sketch]
Based on Lemma~\ref{lem:FP_discontinuities}, in order to prove the theorem, it is enough to show that with probability $1 - \delta/n$, any interval of width $O\left(1/\left(\kappa \sqrt{N}\right)\right)$ contains at most $\tilde O\left(n \sqrt{N}\right)$ points from the set $\left\{\norm{\vec{\theta}_{-i}^{(j)}}_{\infty}\right\}_{j \in [N]}$. We do so by relying on the $\kappa$-bounded assumption, which allows us to analyze the concentration of the discontinuities within any interval.
\end{proof}

\begin{theorem}\label{thm:FP_disp2}
For all agents $i \in [n]$, reported types $\hat{\theta}_i \in [0,1]$, and type profiles $\vec{\theta}_{-i} \in [0,1]^{n-1}$, the function $u_{i, M}\left(\cdot, \hat{\theta}_i, \vec{\theta}_{-i}\right)$ is $1$-Lipschitz.
\end{theorem}

\begin{proof}
Since the reported types $\left(\hat{\theta}_i, \vec{\theta}_{-i}\right)$ are fixed, the allocation is fixed. Thus, $u_{i, M}\left(\cdot, \hat{\theta}_i, \vec{\theta}_{-i}\right)$ is either a constant function if $\hat{\theta}_i \leq \norm{\vec{\theta}_{-i}}_{\infty}$ or a linear function if $\hat{\theta}_i > \norm{\vec{\theta}_{-i}}_{\infty}$.
\end{proof}

Next, we prove the following pseudo-dimension bound. The full proof is in Appendix~\ref{app:applications}.

\begin{restatable}{theorem}{FPpdim}\label{thm:FP_pdim}
For any agent $i \in [n]$, the pseudo-dimension of the class $\cF_{i, M}$ is 2.
\end{restatable}

\begin{proof}[Proof sketch]
First, we prove $\pdim\left(\cF_{i,M}\right) \leq 2$. For a contradiction, suppose there exists a set $\sample_{-i} = \left\{\vec{\theta}_{-i}^{\left(1\right)}, \vec{\theta}_{-i}^{\left(2\right)}, \vec{\theta}_{-i}^{\left(3\right)}\right\}$ that is shattered by $\cF_{i, M}$. Without loss of generality, assume $\norm{\vec{\theta}_{-i}^{\left(1\right)}}_{\infty} < \norm{\vec{\theta}_{-i}^{\left(2\right)}}_{\infty} < \norm{\vec{\theta}_{-i}^{\left(3\right)}}_{\infty}$. Since $\sample_{-i}$ is shatterable, there exist $z^{\left(1\right)}, z^{\left(2\right)}, z^{\left(3\right)} \in \R$ witnessing the shattering.

We split the proof into two cases: $z^{\left(3\right)} > 0$ or $z^{\left(3\right)} \leq 0$. In this sketch, we analyze the former case.
Since $\sample_{-i}$ is shatterable, there is a value $\theta_i \in [0,1]$ and bid $\hat{\theta}_i \in [0,1]$ such that $u_{i, M}\left(\theta_i, \hat{\theta}_i, \vec{\theta}_{-i}^{\left(1\right)}\right) \geq z^{\left(1\right)}$, $
u_{i, M}\left(\theta_i, \hat{\theta}_i, \vec{\theta}_{-i}^{\left(2\right)}\right) < z^{\left(2\right)}$, and $u_{i, M}\left(\theta_i, \hat{\theta}_i, \vec{\theta}_{-i}^{(3)}\right) \geq z^{\left(3\right)}$. Since $u_{i, M}\left(\theta_i, \hat{\theta}_i, \vec{\theta}_{-i}^{(3)}\right) \geq z^{\left(3\right)} > 0$ and $u_{i, M}\left(\theta_i, \hat{\theta}_i, \vec{\theta}_{-i}^{(3)}\right) = 0$ if $\norm{\vec{\theta}_{-i}^{(3)}}_{\infty} \geq \hat{\theta}_i$, it must be that $\hat{\theta}_i > \norm{\vec{\theta}_{-i}^{\left(3\right)}}_{\infty} > \norm{\vec{\theta}_{-i}^{\left(2\right)}}_{\infty} > \norm{\vec{\theta}_{-i}^{\left(1\right)}}_{\infty}$. Therefore, $\theta_i - \hat{\theta}_i \geq z^{\left(1\right)}$ and $\theta_i-\hat{\theta}_i < z^{\left(2\right)}$, so $z^{\left(1\right)} < z^{\left(2\right)}$.
Similarly, there must also be $\theta_i' \in [0,1]$ and $\hat{\theta}_i' \in [0,1]$ such that $u_{i, M}\left(\theta_i', \hat{\theta}_i', \vec{\theta}_{-i}^{(1)}\right) < z^{\left(1\right)}$, $u_{i, M}\left(\theta_i', \hat{\theta}_i', \vec{\theta}_{-i}^{(2)}\right) \geq z^{\left(2\right)}$, and $u_{i, M}\left(\theta_i', \hat{\theta}_i', \vec{\theta}_{-i}^{(3)}\right) \geq z^{\left(3\right)}$. By a similar argument, this means $z^{\left(1\right)} > z^{\left(2\right)}$, which is a contradiction.
We prove that in the case where $z^{(3)} \leq 0$, we also arrive at a contradiction. Therefore, $\pdim\left(\cF_{i, M}\right) \leq 2.$ Furthermore, we exhibit a set of size 2 that is shattered by $\cF_{i, M}$, which means that $\pdim\left(\cF_{i, M}\right) \geq 2.$ Therefore, the theorem statement holds.
\end{proof}

\subsubsection{First-price combinatorial auction}\label{sec:FP_combinatorial}

Under this auction, there are $\ell$ items for sale and each agent's type $\vec{\theta}_i \in [0,1]^{2^{\ell}}$ indicates his value for each bundle $b \subseteq [\ell]$. We denote his value and bid for bundle $b$ as $\theta_i(b)$ and $\hat{\theta}_i(b)$, respectively. The allocation $\left(b^*_1, \dots, b^*_n\right)$ is the solution to the \emph{winner determination problem}:
\[\begin{array}{lll}
\text{maximize} &\sum_{i = 1}^n \hat{\theta}_i\left(b_i\right)&\\
\text{subject to} &b_i \cap b_{i'} = \emptyset &\forall i,i' \in [n], i \not= i'.
\end{array}\]
Each agent $i \in [n]$ pays $\hat{\theta}_i\left(b_i^*\right)$.

We begin with dispersion guarantees. The full proof of the following theorem is in Appendix~\ref{app:applications}.

\begin{restatable}{theorem}{FPcomb}\label{thm:FP_comb_disp}
Suppose that for each pair of agents $i,i' \in [n]$ and each pair of bundles $b, b' \subseteq [\ell]$, the values $\theta_i(b)$ and $\theta_{i'}(b')$ have a $\kappa$-bounded joint density function.
With probability $1-\delta$ over the draw of the $n$ sets $\sample_{-i} = \left\{\vec{\theta}^{\left(1\right)}_{-i}, \dots, \vec{\theta}^{\left(N\right)}_{-i}\right\} \sim \dist_{-i}^N$, we have that for all agents $i \in [n]$ and types $\vec{\theta}_i \in [0,1]^{2^{\ell}}$, the functions $u_{i, M}\left(\vec{\theta}_i, \cdot, \vec{\theta}_{-i}^{(1)}\right), \dots, u_{i, M}\left(\vec{\theta}_i, \cdot, \vec{\theta}_{-i}^{(N)}\right)$ are piecewise $1$-Lipschitz and $\left(O\left(1/\left(\kappa \sqrt{N}\right)\right), \tilde O\left((n+1)^{2\ell} \sqrt{N\ell }\right)\right)$-dispersed.
\end{restatable}

\begin{proof}[Proof sketch]
Consider agent 1. Fix an arbitrary sample $j \in [N]$ and pair of allocations $\left(b_1, \dots, b_n\right)$ and $\left(b_1', \dots, b_n'\right)$. We know that $\left(b_1', \dots, b_n'\right)$ will \emph{not} be the allocation so long as $\hat{\vec{\theta}}_1 \in [0,1]^{2^{\ell}}$ is chosen such that $\hat{\theta}_1(b_1) + \sum_{i = 2}^n \theta^{(j)}_i(b_i) > \hat{\theta}_1(b_1') + \sum_{i = 2}^n \theta^{(j)}_i(b_i').$ This means that across all $\vec{\theta}_1 \in [0,1]^{2^{\ell}}$, there is a fixed set $\cH_j$ of ${(n+1)^\ell \choose 2}$ hyperplanes (one per pair of allocations) such that for any connected component $C$ of $[0,1]^{2^{\ell}} \setminus \cH_j$, the allocation given bids $\left(\hat{\vec{\theta}}_1, \vec{\theta}_{-1}^{(j)}\right)$ is invariant across all $\hat{\vec{\theta}}_1 \in C$. When the allocation is fixed, agent 1's utility is 1-Lipschitz as a function of $\hat{\vec{\theta}}_1$.

  Next, consider the set $\sample_{-1} = \left\{\vec{\theta}^{\left(1\right)}_{-1}, \dots, \vec{\theta}^{\left(N\right)}_{-1}\right\}$ of type profiles and the corresponding sets $\cH_1, \dots, \cH_N$ of hyperplanes. These hyperplanes can be partitioned into ${(n+1)^{\ell} \choose 2}$
  buckets consisting of parallel hyperplanes with offsets independently drawn
  from $\kappa$-bounded distributions. Within each bucket, the
  offsets are $(w',k')$-dispersed with high probability with $w' = O\left(1/\left(\kappa \sqrt{N}\right)\right)$ and $k' = \tilde O\left(\sqrt{N\ell}\right).$ Since the hyperplanes within
  each bucket are parallel and since their offsets are dispersed, for any ball
  $\cB$ of radius $w'$ in $[0,1]^{2^{\ell}}$, at most $k'$ hyperplanes from each
  set intersect $\cB$. The theorem statement holds by a union bound over the $n$ agents and ${(n+1)^{\ell} \choose 2}$
  buckets. \end{proof}

\begin{theorem}
For all agents $i \in [n]$, reported types $\hat{\vec{\theta}}_i \in [0,1]^{2^\ell}$, and type profiles $\vec{\theta}_{-i} \in [0,1]^{(n-1)2^\ell}$, the function $u_{i, M}\left(\cdot, \hat{\vec{\theta}}_i, \vec{\theta}_{-i}\right)$ is $1$-Lipschitz.
\end{theorem}

\begin{proof}
So long as all bids are fixed, the allocation is fixed, so $u_{i, M}\left(\cdot, \hat{\vec{\theta}}_i, \vec{\theta}_{-i}\right)$ is 1-Lipschitz.
\end{proof}

Next, we prove the following pseudo-dimension bound.

\begin{theorem}
For any agent $i \in [n]$, the pseudo-dimension of the class $\cF_{i, M}$ is $O\left(\ell2^{\ell}\log n\right)$.
\end{theorem}

\begin{proof}
As we saw in the proof of Theorem~\ref{thm:FP_comb_disp}, for any type profile $\vec{\theta}_{-i} \in [0,1]^{(n-1)2^\ell}$, there is a set $\cH$ of $(n+1)^{2\ell}$ hyperplanes such that for any connected component $C$ of $[0,1]^{2^{\ell}} \setminus \cH$, the auction's allocation given bids $\left(\hat{\vec{\theta}}_i, \vec{\theta}_{-i}\right)$ is invariant across all $\hat{\vec{\theta}}_i \in C$. So long as the allocation is fixed, $u_{i, M}\left(\cdot, \cdot, \vec{\theta}_{-i}\right)$ is a linear function of $\left(\vec{\theta}_i, \hat{\vec{\theta}}_i\right)$. Therefore, the pseudo-dimension bound follows from Theorem~\ref{thm:delineable} in Appendix~\ref{app:delineable}, which relates the class's pseudo-dimension to the number of hyperplanes splitting the type space into regions where the utility function is linear.
\end{proof}

\subsubsection{Generalized second-price auction}

A \emph{generalized second-price auction} allocates $m$ advertising slots to a set of $n > m$ agents. Each slot $s$ has a probability $\alpha_{s,i}$ of being clicked if agent $i$'s advertisement is in that slot. We assume $\alpha_{s,i}$ is fixed and known by the mechanism designer. The mechanism designer assigns a weight $\omega_i \in (0, 1]$ per agent $i$. Each agent has a value $\theta_i \in [0,1]$ for a click and submits a bid $\hat{\theta}_i \in [0,1]$. The mechanism allocates the first slot to the agent with the highest weighted bid $\omega_i\hat{\theta}_i$, the second slot to the agent with the second highest weighted bid, and so on. Let $\pi(s)$ be the agent allocated slot $s$. If slot $s$ is clicked on, agent $\pi(s)$ pays the lowest amount that would have given him slot $s$, which is $\omega_{\pi(s+1)}\hat{\theta}_{\pi(s+1)}/\omega_{\pi(s)}$. Agent $\pi(s)$'s expected utility is thus $u_{\pi(s), M}\left(\vec{\theta}, \hat{\vec{\theta}}\right) = \alpha_{s,\pi(s)}\left(\theta_{\pi(s)} - \omega_{\pi(s+1)}\hat{\theta}_{\pi(s+1)}/\omega_{\pi(s)}\right).$

For $r \in \Z_{\geq 1}$, let $\mclass_r$ be the set of auctions defined by agent weights from the set $\left\{1/r, 2/r, \dots, 1\right\}$. We begin by proving dispersion guarantees. The full proof is in Appendix~\ref{app:applications}.

\begin{restatable}{theorem}{GSP}
Suppose each agent's type has a $\kappa$-bounded density function.
With probability $1-\delta$ over the draw of the $n$ sets $\sample_{-i} = \left\{\vec{\theta}^{(1)}_{-i}, \dots, \vec{\theta}^{(N)}_{-i}\right\} \sim \dist_{-i}^N$, we have that for all agents $i \in [n]$, types $\theta_i \in [0,1]$, and mechanisms $M \in \cM_r$, the functions $u_{i, M}\left(\theta_i, \cdot, \vec{\theta}_{-i}^{(1)}\right), \dots, u_{i, M}\left(\theta_i, \cdot, \vec{\theta}_{-i}^{(N)}\right)$ are piecewise $0$-Lipschitz and $\left(O\left(1/\left(r \kappa \sqrt{N}\right)\right), \tilde O\left(\sqrt{n^3N}\right)\right)$-dispersed.
\end{restatable}

\begin{proof}[Proof sketch]
Consider agent 1. Choose an arbitrary sample $j \in [N]$ and mechanism $M \in \cM_r$ with weights $\vec{\omega} = \left(\omega_1, \dots, \omega_n\right)$. Let $\omega_{i_1}\theta_{i_1}^{(j)} \leq \cdots \leq \omega_{i_{n-1}}\theta_{i_{n-1}}^{(j)}$ be the weighted types of all agents except agent $1$. Consider the $n$ intervals delineated by these $n-1$ weighted values and suppose we vary $\hat{\theta}_1$ in such a way that $\omega_1\hat{\theta}_1$ remains within one interval. The allocation will be invariant, so agent 1's utility will be a constant function of $\hat{\theta}_1$. Therefore, no matter the value of $\theta_1 \in [0,1]$, the functions $u_{1, M}\left(\theta_1, \cdot, \vec{\theta}_{-1}^{(1)}\right), \dots, u_{1, M}\left(\theta_1, \cdot, \vec{\theta}_{-1}^{(N)}\right)$ are piecewise constant with discontinuities in the set $\cB_{1, \vec{\omega}} = \left\{\omega_{i}\theta_{i}^{(j)}/\omega_1 : i \in \{2, \dots, n\}, j \in [N]\right\}.$
We use the $\kappa$-bounded assumption to show that with probability $1-\delta$, for all $\vec{\omega}$ and agents $i \in [n]$, at most $\tilde{O}\left(n^{3/2} \sqrt{N}\right)$ of the values in each set $\cB_{i,\vec{\omega}}$ fall within any interval of length $O\left(1/\left(r \kappa \sqrt{N}\right)\right)$.
\end{proof}

\begin{theorem}
For all agents $i \in [n]$, all reported types $\hat{\theta}_i \in [0,1]$, all type profiles $\vec{\theta}_{-i} \in [0,1]^{n-1}$, and all generalized second-price auctions $M$, the function $u_{i, M}\left(\cdot, \hat{\theta}_i, \vec{\theta}_{-i}\right)$ is $1$-Lipschitz.
\end{theorem}

\begin{proof}
Since the reported types $\left(\hat{\theta}_i, \vec{\theta}_{-i}\right)$ are fixed, the allocation is fixed. Let $\pi(s)$ be the agent who is allocated slot $s$. The function $u_{i, M}\left(\theta_i, \hat{\theta}_i, \vec{\theta}_{-i}\right)$ is either a constant function of $\theta_i$ if agent $i$ is not allocated a slot or a linear function of $\theta_i$ with a slope of $\alpha_{\pi^{-1}(i), i}$ otherwise.
\end{proof}

We now provide the following pseudo-dimension guarantee. The full proof is in Appendix~\ref{app:applications}.

\begin{restatable}{theorem}{GSPpdim}
For any agent $i \in [n]$ and $r \in \Z_{\geq 1}$, $\pdim\left(\cF_{i, \mclass_r}\right) = O\left(n\log n\right).$
\end{restatable}

\begin{proof}[Proof sketch]
Suppose $i = 1$.
Fix $\vec{\theta}_{-1} = \left(\theta_2, \dots, \theta_n\right) \in [0,1]^{n-1}$.
We denote agent 1's utility when the agents' true types are $\left(\theta_1, \vec{\theta}_{-1}\right)$ and reported types are $\left(\hat{\theta}_1, \vec{\theta}_{-1}\right)$ as the function $u_{\vec{\theta}_{-1}}\left(\theta_1, \hat{\theta}_1, \omega_1, \dots, \omega_n\right)$, which maps $\R^{n+2}$ to $[-1,1]$.
We show that we can split $\R^{n+2}$ into regions where the allocation is fixed as we vary $\left(\theta_1, \hat{\theta}_1, \omega_1, \dots, \omega_n\right)$ over any one region. When the allocation is fixed, $u_{\vec{\theta}_{-1}}$ is a linear function of $\left(\theta_1, \omega_2/\omega_1, \dots, \omega_n/\omega_1\right)$.
We prove that this partition of $\R^{n+2}$ is delineated by a small number of polynomials of degree at most 2 over these $n+2$ variables. This fact allows us to bound the number of regions making up the partition. Since the utility function is simple within each region, we use our bound on the number of regions to prove the theorem.
\end{proof}

\subsubsection{Discriminatory auction}\label{sec:discriminatory}

Under the discriminatory auction, there are $m$ identical units of a single item for sale. For each agent $i \in [n]$, his type $\vec{\theta}_i \in [0,1]^m$ indicates how much he is willing to pay for each additional unit. Thus, $\theta_i[1]$ is the amount he is willing to pay for one unit, $\theta_i[1] + \theta_i[2]$ is the amount
he is willing to pay for two units, and so on. We assume that $\theta_i[1] \geq \theta_i[2] \geq \cdots \geq \theta_i[m]$.
The auctioneer collects $nm$ bids
$\hat{\theta}_i[\mu]$ for $i \in [n]$ and $\mu \in [m]$. If exactly $m_i$ of agent $i$'s bids
are among the $m$ highest of all $nm$ bids,
then agent $i$ is awarded $m_i$ units and pays $\sum_{\mu = 1}^{m_i} \hat{\theta}_i[\mu]$.

We begin with dispersion guarantees. The full proof of the following theorem is in Appendix~\ref{app:applications}.

\begin{restatable}{theorem}{discrim}\label{thm:discriminatory_vary_theta_hat}
Suppose that each agent's value for each marginal unit has a $\kappa$-bounded density function.
With probability $1-\delta$ over the draw of the $n$ sets $\sample_{-i} = \left\{\vec{\theta}^{\left(1\right)}_{-i}, \dots, \vec{\theta}^{\left(N\right)}_{-i}\right\} \sim \dist_{-i}^N$, for all agents $i \in [n]$ and types $\vec{\theta}_i \in [0,1]^m$, the functions $u_{i, M}\left(\vec{\theta}_i, \cdot, \vec{\theta}_{-i}^{(1)}\right), \dots, u_{i, M}\left(\vec{\theta}_i, \cdot, \vec{\theta}_{-i}^{(N)}\right)$ are piecewise $1$-Lipschitz and $\left(O\left(1/\left(\kappa \sqrt{N}\right)\right),\tilde O\left(nm^2 \sqrt{N}\right)\right)$-dispersed.
\end{restatable}

\begin{proof}[Proof sketch]
Consider agent 1 and suppose we sort $\left\{\theta_i^{(j)}[\mu] : i \in \{2, \dots, n\}, j \in [N], \mu \in [m]\right\}$. So long as agent 1's bid falls between these sorted bids, the allocations will be fixed across all $j \in [N]$, and thus each $u_{1, M}\left(\vec{\theta}_1, \cdot, \vec{\theta}_{-1}^{(j)}\right)$ is $1$-Lipschitz.
Across all $\vec{\theta}_1 \in [0,1]^m$, the partition $\cP_j$ splitting $u_{1, M}\left(\vec{\theta}_1, \cdot, \vec{\theta}_{-1}^{(j)}\right)$ into Lipschitz portions is  delineated by the set of hyperplanes \[\left\{\theta_i^{(j)}[\mu] - \hat{\theta}_1[\mu'] = 0 : i \in \{2, \dots, n\}, \mu, \mu' \in [m]\right\}.\] 
We partition these hyperplanes into $m^2(n-1)$
  buckets consisting of parallel hyperplanes with offsets independently drawn
  from $\kappa$-bounded distributions. The remainder of the proof is similar to that of Theorem~\ref{thm:FP_comb_disp}.
\end{proof}

\begin{theorem}
For all agents $i \in [n]$, reported types $\hat{\vec{\theta}}_i \in [0,1]^m$, and type profiles $\vec{\theta}_{-i} \in [0,1]^{(n-1)m}$, the function $u_{i, M}\left(\cdot, \hat{\vec{\theta}}_i, \vec{\theta}_{-i}\right)$ is 1-Lipschitz.
\end{theorem}

\begin{proof}
So long as all bids are fixed, the allocation is fixed, so $u_{i, M}\left(\cdot, \hat{\vec{\theta}}_i, \vec{\theta}_{-i}\right)$ is 1-Lipschitz.
\end{proof}

Next, we prove the following pseudo-dimension bound.

\begin{theorem}\label{thm:discriminatory_pdim}
For any agent $i \in [n]$, the pseudo-dimension of the class $\cF_{i, M}$ is $O\left(m\log (nm)\right)$.
\end{theorem}

\begin{proof}
As we saw in the proof of Theorem~\ref{thm:discriminatory_vary_theta_hat}, for any $\vec{\theta}_{-i} \in [0,1]^{(n-1)m}$, there is a set $\cH$ of $O\left(m^2n\right)$ hyperplanes such that for any connected component $C$ of $[0,1]^{m} \setminus \cH$, the allocation is fixed across all $\hat{\vec{\theta}}_i \in C$. So long as the allocation is fixed, $u_{i, M}\left(\cdot, \cdot, \vec{\theta}_{-i}\right)$ is a linear function of $\left(\vec{\theta}_i, \hat{\vec{\theta}}_i\right)$. Therefore, the pseudo-dimension bound follows directly from Theorem~\ref{thm:delineable} in Appendix~\ref{app:delineable}.
\end{proof}

\subsubsection{Uniform-price auction}

Under this auction, the allocation rule is the same as in the discriminatory auction (Section~\ref{sec:discriminatory}).
However, all $m$ units are sold at a ``market-clearing'' price,
meaning the total amount demanded equals the total amount supplied. See the formal definition in Appendix~\ref{app:uniform}. We obtain the same bounds as we do in Section~\ref{sec:discriminatory}, so we state the theorems in Appendix~\ref{app:uniform}.

\subsubsection{Second-price auction with spiteful agents}
We conclude by studying \emph{spiteful agents}~\citep{Brandt07:Spiteful, Morgan03:Spite, Sharma10:Asymmetric, Tang12:Optimal}, where each bidder's utility not only increases when his surplus increases, but also decreases when the other bidders' surpluses increase. Formally, given a \emph{spite parameter} $\alpha_i \in [0,1]$, agent $i$'s utility under the second-price auction $M$ is $u_{i,M}\left(\vec{\theta}, \hat{\vec{\theta}}\right) = \alpha_i\textbf{1}_{\left\{\hat{\theta}_i > \norm{\hat{\vec{\theta}}_{-i}}_{\infty}\right\}}\left(\theta_i - \norm{\hat{\vec{\theta}}_{-i}}_{\infty}\right) - \left(1 - \alpha_i\right) \sum_{i' \not= i} \textbf{1}_{\left\{\hat{\theta}_{i'} > \norm{\hat{\vec{\theta}}_{-i'}}_{\infty}\right\}}\left(\theta_{i'} - \norm{\hat{\vec{\theta}}_{-i'}}_{\infty}\right).$ The closer $\alpha_i$ is to zero, the more spiteful bidder $i$ is.

We begin with the following lemma, which follows from the form of $u_{1, M}\left(\theta_1, \cdot, \vec{\theta}_{-1}^{(j)}\right)$ (see Figure~\ref{fig:spiteful}). The proof is in Appendix~\ref{app:applications}.
\begin{figure}
    \centering
    \includegraphics{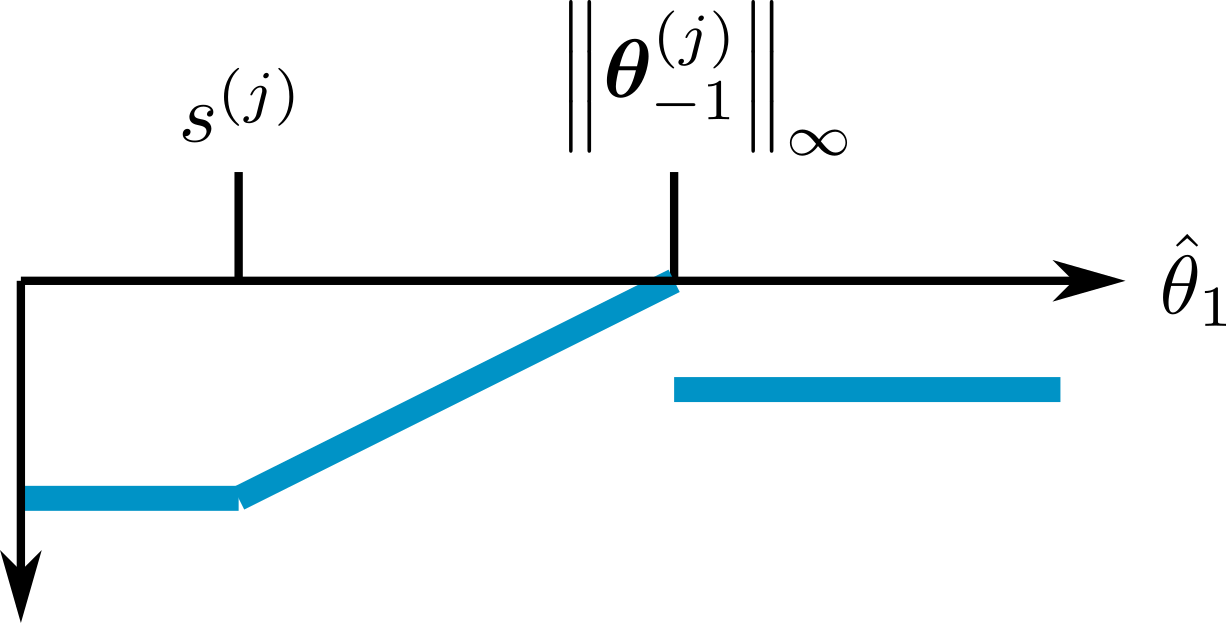}
    \caption{Graph of $u_{1, M}\left(\theta_1, \cdot, \vec{\theta}_{-1}^{(j)}\right)$ for a spiteful agent under a second-price auction with $\alpha_1 = \theta_1 = \frac{1}{2}$, $\norm{\vec{\theta}_{-1}}_{\infty} = \frac{3}{4}$, and $s^{(j)} = \frac{1}{4}$.}\label{fig:spiteful}
\end{figure}

\begin{restatable}{lemma}{SPdisc}\label{lem:SP_discontinuities}
For any agent $i \in [n]$, any $\sample_{-i} = \left\{\vec{\theta}^{\left(1\right)}_{-i}, \dots, \vec{\theta}^{\left(N\right)}_{-i}\right\}$, and any type $\theta_i \in [0,1]$, the functions $u_{i, M}\left(\theta_i, \cdot, \vec{\theta}_{-i}^{(1)}\right), \dots, u_{i, M}\left(\theta_i, \cdot, \vec{\theta}_{-i}^{(N)}\right)$ are piecewise 1-Lipschitz and their discontinuities fall in the set $\left\{\norm{\vec{\theta}_{-i}^{(j)}}_{\infty}\right\}_{j \in [N]}$.
\end{restatable}

Lemma~\ref{lem:SP_discontinuities} implies the following dispersion guarantee.

\begin{restatable}{theorem}{SP}
Suppose each agent's type has a $\kappa$-bounded density function.
With probability $1-\delta$ over the draw of the $n$ sets $\sample_{-i} = \left\{\vec{\theta}^{(1)}_{-i}, \dots, \vec{\theta}^{(N)}_{-i}\right\} \sim \dist_{-i}^N$, we have that for all agents $i \in [n]$ and types $\theta_i \in [0,1]$, the functions $u_{i, M}\left(\theta_i, \cdot, \vec{\theta}_{-i}^{(1)}\right), \dots, u_{i, M}\left(\theta_i, \cdot, \vec{\theta}_{-i}^{(N)}\right)$ are piecewise $1$-Lipschitz and $\left(O\left(1/\left(\kappa \sqrt{N}\right)\right), \tilde O\left(n \sqrt{N}\right)\right)$-dispersed.
\end{restatable}

\begin{proof}
This follows from Lemma~\ref{lem:SP_discontinuities} exactly as Theorem~\ref{thm:FP_disp} follows from Lemma~\ref{lem:FP_discontinuities}.
\end{proof}

\begin{theorem}
For all agents $i \in [n]$, reported types $\hat{\theta}_i \in [0,1]$, and type profiles $\vec{\theta}_{-i} \in [0,1]^{n-1}$, the function $u_{i, M}\left(\cdot, \hat{\theta}_i, \vec{\theta}_{-i}\right)$ is $1$-Lipschitz.
\end{theorem}

\begin{proof}
Since the reported types $\left(\hat{\theta}_i, \vec{\theta}_{-i}\right)$ are fixed, the allocation is fixed. If $\hat{\theta}_i > \norm{\vec{\theta}_{-i}}_{\infty}$, then $u_{i, M}\left(\theta_i, \hat{\theta}_i, \vec{\theta}_{-i}\right)= \alpha_i\left(\theta_i - \norm{\vec{\theta}_{-i}}_{\infty}\right)$, which is an $\alpha_i$-Lipschitz function of $\theta_i$. Otherwise, it is a constant function of $\theta_i$.
\end{proof}

We conclude with the following pseudo-dimension bound.

\begin{theorem}
For any agent $i \in [n]$, the pseudo-dimension of the class $\cF_{i, M}$ is $O(1)$.
\end{theorem}

\begin{proof}
For any vector $\vec{\theta}_{-i} \in [0,1]^{(n-1)}$, let $s$ be the second-highest component. Letting $\cH = \{\hat{\theta}_i = \norm{\vec{\theta}_{-i}}_{\infty}, \hat{\theta}_i = s\}$ we know that for any interval $I$ of $[0,1] \setminus \cH$, the allocation is fixed across all $\hat{\theta}_i \in C$. So long as the allocation is fixed, $u_{i, M}\left(\cdot, \cdot, \vec{\theta}_{-i}\right)$ is a linear function of $\left(\theta_i, \hat{\theta}_i\right)$. Therefore, the pseudo-dimension bound follows directly from Theorem~\ref{thm:delineable} in Appendix~\ref{app:delineable}.
\end{proof}

\section{Conclusion}
In this paper, we provided techniques for estimating how far a mechanism is from \emph{ex-interim} incentive compatible. We introduced an empirical variant of approximate incentive compatibility which intuitively measures the maximum utility an agent can gain by misreporting his type, on average over the samples. We bounded the difference between our empirical incentive compatibility estimate and the true incentive compatibility approximation factor. To do so, we relied on a subtle mixture of tools from learning theory, including dispersion and pseudo-dimension. We thus derived strong
guarantees for many important manipulable mechanisms, including the first-price auction, generalized second-price auction, discriminatory auction, uniform price auction, and second-price auction under spiteful agents.

As in prior research on auction design via machine learning, we assumed that we could receive independent samples from the type distribution. One direction for future work would be to relax this assumption in the spirit of incentive-aware learning~\citep{Epasto18:Incentive}. Another would be to relax the assumption that each agent's type is independent from the others' in the \emph{ex-interim} setting.

\subsection*{Acknowledgments} We thank Zhe Feng, Noah Golowich, Harikrishna Narasimhan, and David Parkes for a valuable discussion on approximate incentive compatibility in the context of deep learning, a motivating factor in our study of this general setting.

Moreover, we thank Fabian Pieroth for pointing out an inaccuracy in the appendix of an earlier arXiv version of this paper, as well as a missing constant in Theorem~\ref{thm:pdim}.

This material is based on work supported by the National Science Foundation under grants IIS-1618714, IIS-1718457, IIS-1617590, IIS-1901403, CCF-1733556, CCF-1422910, CCF-1535967, the ARO under award W911NF-17-1-0082, an Amazon Research Award, a Microsoft Research Faculty Fellowship, and a Google Research Award.

\bibliography{../../dairefs}
\bibliographystyle{plainnat}

\appendix
\section{Helpful lemmas}\label{app:helpful}
\begin{lemma}[\citet{Anthony09:Neural}]\label{lem:poly}
Suppose $f_1,\dots, f_p$ are polynomials of degree at most $d$ in $v \leq p$ variables. Then \[\left|\left\{\begin{pmatrix}
\sign\left(f_1(\vec{x})\right)\\
\vdots\\
\sign\left(f_p(\vec{x})\right)
\end{pmatrix} : \vec{x} \in \R^v\right\}\right| \leq 2\left(\frac{2epd}{v}\right)^v.\]
\end{lemma}

\begin{lemma}[\citet{Shalev14:Understanding}]\label{lem:log_ineq}
Let $\alpha \geq 1$ and $\beta > 0$. If $x < \alpha \log x + \beta$, then $x < 4\alpha \log (2 \alpha) + 2 \beta$.
\end{lemma}

\subsection{Dispersion lemmas}\label{app:dispersion}
We now include a few lemmas about dispersion from work by \citet{Balcan18:Dispersion} as well as a few variations we prove ourselves for our specific applications.

\begin{lemma}[\citet{Balcan18:Dispersion}]\label{lem:dispersion}
  Let $\cB = \{\beta_1, \dots, \beta_r\} \subset \R$ be a collection of
  samples where each $\beta_i$ is drawn from a distribution with a $\kappa$-bounded density function. For any $\delta \geq 0$, the following statements hold
  with probability at least $1-\delta$:
  \begin{enumerate}
    \item If the $\beta_i$ are independent, then every interval of width $w$
    contains at most $k = O(r w \kappa + \sqrt{r \log(1/\delta)})$ samples. In
    particular, for any $\alpha \geq 1/2$ we can take $w = 1/(\kappa
    r^{1-\alpha})$ and $k = O(r^\alpha \sqrt{\log(1/\delta)})$.
    \item If the samples can be partitioned into $P$ buckets $\cB_1, \dots,
    \cB_P$ such that each $\cB_i$ contains independent samples and $|\cB_i| \leq
    M$, then every interval of width $w$ contains at most $k = O(P M w \kappa +
    \sqrt{M \log(P/\delta)}$. In particular, for any $\alpha \geq 1/2$ we can
    take $w = 1/(\kappa M^{1-\alpha})$ and $k = O(P M^\alpha \sqrt{
    \log(P/\delta)})$.
  \end{enumerate}
\end{lemma}

\begin{lemma}[\citet{Balcan18:Dispersion}]\label{lem:constant_mult}
Suppose $X$ is a random variable with $\kappa$-bounded density function and suppose $c \not= 0$ is a constant. Then $cX$ has a $\frac{\kappa}{|c|}$-bounded density function.
\end{lemma}

\begin{lemma}\label{lem:sum_bounded}
Suppose $X$ and $Y$ are random variables taking values in $[0,1]$ and suppose that their joint distribution has a $\kappa$-bounded density function. Then the distribution of $X+Z$ has a $\kappa$-bounded density function.
\end{lemma}

\begin{proof}
Let $Z = X+Y$. We will perform change of variables using the function $g(x,y) = (x, x + y)$. Let $g^{-1}(x,z) = h(x,z) = (x, z - x).$ Then \[J_h(x,z) = \det \begin{pmatrix}
1 & 0\\
-1 & 1
\end{pmatrix} = 1.\] Therefore, $f_{X,Z}(x,z) = f_{X,Y}(x,z-x).$ This means that $f_Z(z) = \int_0^1 f_{X,Y}(x,z-x) \, dx \leq \kappa$, so the theorem statement holds.
\end{proof}

\begin{lemma}\label{lem:difference_bounded}
Suppose $X$ and $Y$ are random variables taking values in $[0,1]$ and suppose that their joint distribution has a $\kappa$-bounded density function. Then the distribution of $X-Z$ has a $\kappa$-bounded density function.
\end{lemma}

\begin{proof}
Let $Z = X-Y$. We will perform change of variables using the function $g(x,y) = (x, x - y)$. Let $g^{-1}(x,z) = h(x,z) = (x, x-z).$ Then \[J_h(x,z) = \det \begin{pmatrix}
1 & 0\\
1 & -1
\end{pmatrix} = -1.\] Therefore, $f_{X,Z}(x,z) = f_{X,Y}(x,x - z).$ This means that $f_Z(z) = \int_0^1 f_{X,Z}(x,z) \, dx = \int_0^1 f_{X,Y}(x,x -z) \, dx \leq \kappa$, so the theorem statement holds.
\end{proof}

\begin{definition}[Hyperplane delineation]
Let $\Psi$ be a set of hyperplanes in $\R^m$ and let $\cP$ be a partition of a set $\R^m$. Let $K_1, \dots, K_q$ be the connected components of $\R^m \setminus \Psi$. Suppose every set in $\cP$ is the union of some collection of sets $K_{i_1}, \dots, K_{i_j}$ together with their limit points. Then we say that the set $\Psi$ \emph{delineates} $\cP$.
\end{definition}

\begin{lemma}[\citet{Balcan18:Dispersion}]\label{lem:hyperplanes}
Let $u_1, \dots, u_N$ be a set of piecewise $L$-Lipschitz functions mapping $\R^m$ to $\R$, drawn i.i.d. from a distribution $\dist$. For each $j \in [N]$, let $\cP_j$ be
the partition of $[0,1]^m$ such that over any $R \in
\cP_j$, $u_j$ is $L$-Lipschitz.
Suppose the hyperplane sets $\Psi_1, \dots, \Psi_N$ delineate the partitions $\cP_1, \dots, \cP_N$. Moreover, suppose the multi-set union of $\Psi_1, \dots, \Psi_N$ can be partitioned into $P$ multi-sets $\cB_1, \dots, \cB_P$ such that for each multi-set $\cB_i$:
\begin{enumerate}
\item The hyperplanes in $\cB_i$ are parallel with probability 1 over the draw of $u_1, \dots, u_N$.
\item The offsets of the hyperplanes in $\cB_i$ are independently drawn from $\kappa$-bounded distributions.
\end{enumerate}
With probability at least $1-\delta$ over the draw of $u_1, \dots, u_N$, the functions are \[\left(\frac{1}{2\kappa\sqrt{\max|\cB_i|}}, O\left(P\sqrt{\max|\cB_i|\ln \frac{P}{\delta}}\right)\right)\text{-dispersed}.\]
\end{lemma}

\section{Delineable utility functions}\label{app:delineable}
To prove some of our pseudo-dimension guarantees, we use the notion of \emph{delineability} introduced by \citet{Balcan18:General}. They used it in the context of profit-maximization, so here we adapt it to our setting.

Given a mechanism $M$ and agent $i \in [n]$, let $\cF_{i, M} = \left\{\left. u_{i, M, \theta_i, \hat{\theta}_i} \right| \theta_i, \hat{\theta}_i \in \Theta_i\right\}$, where $u_{i, M, \theta_i, \hat{\theta}_i}$ is the function introduced in Section~\ref{sec:interim}. Moreover, for a fixed type profile $\vec{\theta}_{-i} \in \times_{j \not= i} \Theta_j$, let $u_{\vec{\theta}_{-i}} : \Theta_i^2 \to [-1,1]$ be a function that maps a pair of types $\theta_i, \hat{\theta}_i \in \Theta_i$ to $u_{i, M, \theta_i, \hat{\theta}_i}\left(\vec{\theta}_{-i}\right)$.

\begin{definition}[$(m,t)$-delineable] Given a mechanism $M$ and agent $i \in [n]$, we say the function class $\cF_{i,M}$ is \emph{$\left(m,t\right)$-delineable} if:
\begin{enumerate}
\item Each agent's type space is a subset of $[0,1]^m$; and
\item For any $\vec{\theta}_{-i} \in [0,1]^{m(n-1)}$, there is a set $\hyp$ of $t$ hyperplanes such that for any connected component $\cC$ of $[0,1]^{2m} \setminus \hyp$, the function $u_{\vec{\theta}_{-i}}\left(\vec{\theta}_i, \hat{\vec{\theta}}_i\right)$ is linear over $\cC.$
\end{enumerate}
\end{definition}

The following theorem is similar to Balcan et al.'s main theorem~\citep{Balcan18:General}, though we have adapted it to our setting. We include the proof for completeness.

\begin{theorem}\label{thm:delineable}
If $\cF_{i,M}$ is $\left(m, t\right)$-delineable, the pseudo dimension of $\cF_{i,M}$ is $O\left(m \log \left(mt\right)\right)$.
\end{theorem}

\begin{proof}
Suppose $\pdim\left(\cF_{i,M}\right) = N$. By definition, there exists a set $\sample = \left\{\vec{\theta}_{-i}^{(1)}, \dots, \vec{\theta}_{-i}^{(N)}\right\}$ that is shattered by $\cF_{i,M}$. Let $z^{(1)}, \dots, z^{(N)} \in \R$ be the points that witness this shattering. Again, by definition, we know that for any $T \subseteq [N]$, there exists a pair of types $\vec{\theta}_i^{(T)}, \hat{\vec{\theta}}_i^{(T)} \in [0,1]^m$ such that $u_{i, M, \vec{\theta}_i^{(T)}, \hat{\vec{\theta}}_i^{(T)}}\left(\vec{\theta}_{-i}^{(j)}\right) \geq z^{(j)}$ if and only if $j \in T$. Let $\Theta^* = \left\{\left(\vec{\theta}_i^{(T)}, \hat{\vec{\theta}}_i^{(T)}\right) : T \subseteq [N]\right\}$. To show that the pseudo-dimension $N$ of $\cF_{i,M}$ is $O(m\log (mt))$, we will show that $|\Theta^*| = 2^N < 4m^2t^{2m}N^{4m}$, which means that $N = O(m \log (mt))$.

To this end, for $\vec{\theta}_{-i}^{(j)} \in \sample$, let $\hyp^{(j)}$ be the set of $t$ hyperplanes such 
that for any connected component $\cC$ of $[0,1]^{2m} \setminus \hyp^{(j)}$, $u_{\vec{\theta}_{-i}^{(j)}}\left(\vec{\theta}_i, \hat{\vec{\theta}}_i\right)$ is linear over $\cC.$ We now consider the \emph{overlay} of all $N$ partitions $[0,1]^{2m} \setminus \hyp^{(1)}, \dots, [0,1]^{2m} \setminus \hyp^{(N)}$. Formally, this overlay is made up of the sets $\cC_1, \dots, \cC_{\tau}$ which are the connected components 
of $[0,1]^{2m} \setminus \left(\bigcup_{j = 1}^N \hyp^{(j)}\right)$. For each set $\cC_{\ell}$ and each $j \in [N]$, $\cC_\ell$ is completely contained in a single connected component of $[0,1]^{2m} \setminus \hyp^{(j)}$, which means that $u_{\vec{\theta}_{-i}^{(j)}}\left(\vec{\theta}_i, \hat{\vec{\theta}}_i\right)$ is linear over $\cC_{\ell}$.
As we know from work by \citet{Buck43:Partition}, since $\left|\hyp^{(j)}\right|\leq t$ for all $j \in [N]$, $\tau < 2m(Nt)^{2m}$.

Now, consider a single connected component $\cC_\ell$ of $[0,1]^{2m} \setminus \left(\bigcup_{j = 1}^N \hyp^{(j)}\right)$. For any sample $\vec{\theta}_{-i}^{(j)} \in \sample$, 
we know that $u_{\vec{\theta}_{-i}^{(j)}}\left(\vec{\theta}_i, \hat{\vec{\theta}}_i\right)$ is linear over $\cC_\ell$. 
Let $\vec{a}_\ell^{(j)} \in \R^{2m}$ and $b_\ell^{(j)} \in \R$ be the weight vector and offset 
such that $u_{\vec{\theta}_{-i}^{(j)}}\left(\vec{\theta}_i, \hat{\vec{\theta}}_i\right) = \vec{a}_\ell^{(j)} \cdot \left(\vec{\theta}_i, \hat{\vec{\theta}}_i\right) + b_\ell^{(j)}$ for all $\left(\vec{\theta}_i, \hat{\vec{\theta}}_i\right) \in \cC_\ell$. We know that there is a hyperplane $\vec{a}_\ell^{(j)} \cdot \left(\vec{\theta}_i, \hat{\vec{\theta}}_i\right) + b_\ell^{(j)} = z^{(j)}$ where on one side of the 
hyperplane, $u_{\vec{\theta}_{-i}^{(j)}}\left(\vec{\theta}_i, \hat{\vec{\theta}}_i\right) \leq z^{(j)}$ and on the other 
side, $u_{\vec{\theta}_{-i}^{(j)}}\left(\vec{\theta}_i, \hat{\vec{\theta}}_i\right) > z^{(j)}$.
Let $\hyp_{\cC_\ell}$ be all $N$ hyperplanes for all $N$ samples $\left(\hyp_{\cC_\ell} = \left\{\vec{a}_\ell^{(j)} \cdot \left(\vec{\theta}_i, \hat{\vec{\theta}}_i\right) + b_\ell^{(j)} = z^{(j)} : j \in [N]\right\}\right).$ Notice that in any 
connected component $\cC$ of $\cC_\ell \setminus \hyp_{\cC_\ell}$, for all $j \in [N]$, $u_{\vec{\theta}_{-i}^{(j)}}\left(\vec{\theta}_i, \hat{\vec{\theta}}_i\right)$ is either greater than $z^{(j)}$ or 
less than $z^{(j)}$ (but not both) for all $\left(\vec{\theta}_i, \hat{\vec{\theta}}_i\right) \in \cC$.
Thus, at most one 
vector $\left(\vec{\theta}_i, \hat{\vec{\theta}}_i\right) \in \Theta^*$ can come from $\cC$. In total, the number of 
connected components of $\cC_\ell \setminus \hyp_{\cC_\ell}$ is smaller than $2mN^{2m}$. 
The same holds for every partition $\cC_\ell$. Thus, the total number of regions where 
for all $j \in [N]$, $u_{\vec{\theta}_{-i}^{(j)}}\left(\vec{\theta}_i, \hat{\vec{\theta}}_i\right)$ is either greater 
than $z^{(j)}$ or less than $z^{(j)}$ (but not both) is smaller than $2mN^{2m}\cdot 2m(Nt)^{2m}$. Therefore, we may bound $|\Theta^*| = 2^N < 2mN^{2m}\cdot 2m(Nt)^{2m}$. By Lemma~\ref{lem:log_ineq}, we have that $N = O(m \log (mt))$.
\end{proof}

\section{Proofs about approximate ex-interim incentive compatibility (Section~\ref{sec:interim})}\label{app:interim}
\AIC*

\begin{proof}
Fix an arbitrary agent $i \in [n]$. By Theorem~\ref{thm:pdim}, we know that with probability at least $1-\delta/n$ over the draw of $\sample_{-i}$, for all mechanisms $M \in \cM$ and all $\theta_i, \hat{\theta}_i \in \Theta_i$, \begin{align*}\E_{\vec{\theta}_{-i} \sim \dist_{-i}}\left[u_{i,M}\left(\theta_i, \hat{\theta}_i, \vec{\theta}_{-i}\right)\right] - \frac{1}{N}\sum_{j = 1}^N u_{i,M}\left(\theta_i, \hat{\theta}_i, \vec{\theta}_{-i}^{(j)}\right)\leq \sqrt{\frac{2 d_i}{N} \ln \frac{eN}{d_i}} + \sqrt{\frac{1}{2N}\ln\frac{2n}{\delta}}\end{align*} and 
\begin{align*}\frac{1}{N}\sum_{j = 1}^N u_{i,M}\left(\theta_i, \theta_i, \vec{\theta}_{-i}^{(j)}\right) - \E_{\vec{\theta}_{-i} \sim \dist_{-i}}\left[u_{i,M}\left(\theta_i, \theta_i, \vec{\theta}_{-i}\right)\right]
\leq \sqrt{\frac{2 d_i}{N} \ln \frac{eN}{d_i}} + \sqrt{\frac{1}{2N}\ln\frac{2n}{\delta}}.\end{align*}

Therefore, \begin{align*}
&\E_{\vec{\theta}_{-i} \sim \dist_{-i}}\left[u_{i,M}\left(\theta_i, \hat{\theta}_i, \vec{\theta}_{-i}\right)\right] - \E_{\vec{\theta}_{-i} \sim \dist_{-i}}\left[u_{i,M}\left(\theta_i, \theta_i, \vec{\theta}_{-i}\right)\right]\\
= \text{ } &\E_{\vec{\theta}_{-i} \sim \dist_{-i}}\left[u_{i,M}\left(\theta_i, \hat{\theta}_i, \vec{\theta}_{-i}\right)\right] - \frac{1}{N}\sum_{j = 1}^N u_{i,M}\left(\theta_i, \hat{\theta}_i, \vec{\theta}_{-i}^{(j)}\right)\\
+ \text{ } & \frac{1}{N}\sum_{j = 1}^N u_{i,M}\left(\theta_i, \hat{\theta}_i, \vec{\theta}_{-i}^{(j)}\right) - u_{i,M}\left(\theta_i, \theta_i, \vec{\theta}_{-i}^{(j)}\right)\\
+ \text{ } & \frac{1}{N}\sum_{j = 1}^N u_{i,M}\left(\theta_i, \theta_i, \vec{\theta}_{-i}^{(j)}\right) - \E_{\vec{\theta}_{-i} \sim \dist_{-i}}\left[u_{i,M}\left(\theta_i, \theta_i, \vec{\theta}_{-i}\right)\right]\\
\leq \text{ }& \max_{\theta_i, \hat{\theta}_i \in \Theta_i}\left\{\frac{1}{N} \sum_{j = 1}^N u_{i,M}\left(\theta_i, \hat{\theta}_i, \vec{\theta}_{-i}^{(j)}\right) - u_{i,M}\left(\theta_i, \theta_i, \vec{\theta}_{-i}^{(j)}\right)\right\} + \epsilon_{\cM}(N, \delta).
\end{align*}

Since the above inequality holds for all $i \in [n]$ with probability $1-\delta/n$, a union bound implies that with probability $1-\delta$, for all $i \in [n]$ and all mechanism $M \in \mclass$, \begin{align*}&\E_{\vec{\theta}_{-i} \sim \dist_{-i}}\left[u_{i,M}\left(\theta_i, \hat{\theta}_i, \vec{\theta}_{-i}\right) - u_{i,M}\left(\theta_i, \theta_i, \vec{\theta}_{-i}\right)\right]\\
 \leq &\max_{\theta_i, \hat{\theta}_i \in \Theta_i}\left\{\frac{1}{N} \sum_{j = 1}^N u_{i,M}\left(\theta_i, \hat{\theta}_i, \vec{\theta}_{-i}^{(j)}\right) - u_{i,M}\left(\theta_i, \theta_i, \vec{\theta}_{-i}^{(j)}\right)\right\} + \epsilon_{\cM}(N, \delta),\end{align*} where $\epsilon_{\cM}(N, \delta) = 2\sqrt{\frac{2 d_i}{N} \ln \frac{eN}{d_i}} + 2\sqrt{\frac{1}{2N}\ln\frac{2n}{\delta}}$ and $d_i = \pdim\left(\cF_{i, \cM}\right)$.
\end{proof}

\greedy*
\begin{proof}
Let $\theta_i^*$ and $\hat{\theta}_i^*$ be the types in $\Theta_i$ that maximize \[\sum_{j = 1}^N u_{i,M}\left(\theta_i, \hat{\theta}_i, \vec{\theta}_{-i}^{(j)}\right) - u_{i,M}\left(\theta_i, \theta_i, \vec{\theta}_{-i}^{(j)}\right).\] By definition of the set $\cG\left(\sample_{-i}, M, \epsilon\right)$, we know there exists a pair $\left(\theta_i', \hat{\theta}_i'\right) \in \cG\left(\sample_{-i}, M, \epsilon\right)$ such that \begin{align*}
&\max_{\theta_i, \hat{\theta}_i \in \Theta_i} \left\{\frac{1}{N} \sum_{j = 1}^N u_{i,M}\left(\theta_i, \hat{\theta}_i, \vec{\theta}_{-i}^{(j)}\right) - u_{i,M}\left(\theta_i, \theta_i, \vec{\theta}_{-i}^{(j)}\right)\right\}\\
= &\frac{1}{N} \sum_{j = 1}^N u_{i,M}\left(\theta_i^*, \hat{\theta}_i^*, \vec{\theta}_{-i}^{(j)}\right) - u_{i,M}\left(\theta_i^*, \theta_i^*, \vec{\theta}_{-i}^{(j)}\right)\\
\leq & \frac{1}{N} \sum_{j = 1}^N u_{i,M}\left(\theta_i', \hat{\theta}_i', \vec{\theta}_{-i}^{(j)}\right) - u_{i,M}\left(\theta_i', \theta_i', \vec{\theta}_{-i}^{(j)}\right) + \epsilon\\
&\leq \max_{\theta_i, \hat{\theta}_i \in \cG\left(\sample_{-i}, M, \epsilon\right)} \left\{\frac{1}{N} \sum_{j = 1}^N u_{i,M}\left(\theta_i, \hat{\theta}_i, \vec{\theta}_{-i}^{(j)}\right) - u_{i,M}\left(\theta_i, \theta_i, \vec{\theta}_{-i}^{(j)}\right)\right\} + \epsilon.\\
\end{align*} Therefore, by this inequality and Theorem~\ref{thm:AIC}, we know that with probability at least $1-\delta$, \begin{align*}&\E_{\vec{\theta}_{-i} \sim \dist_{-i}}\left[u_{i,M}\left(\theta_i, \hat{\theta}_i, \vec{\theta}_{-i}\right) - u_{i,M}\left(\theta_i, \theta_i, \vec{\theta}_{-i}\right)\right]\\
 - &\max_{\left(\theta_i, \hat{\theta}_i\right) \in \cG\left(\sample_{-i}, M, \epsilon\right)}\left\{\frac{1}{N} \sum_{j = 1}^N u_{i,M}\left(\theta_i, \hat{\theta}_i, \vec{\theta}_{-i}^{(j)}\right) - u_{i,M}\left(\theta_i, \theta_i, \vec{\theta}_{-i}^{(j)}\right)\right\}\\
 = & \E_{\vec{\theta}_{-i} \sim \dist_{-i}}\left[u_{i,M}\left(\theta_i, \hat{\theta}_i, \vec{\theta}_{-i}\right) - u_{i,M}\left(\theta_i, \theta_i, \vec{\theta}_{-i}\right)\right]\\
 - &\max_{\theta_i, \hat{\theta}_i \in \Theta_i}\left\{\frac{1}{N} \sum_{j = 1}^N u_{i,M}\left(\theta_i, \hat{\theta}_i, \vec{\theta}_{-i}^{(j)}\right) - u_{i,M}\left(\theta_i, \theta_i, \vec{\theta}_{-i}^{(j)}\right)\right\}\\
 + &\max_{\theta_i, \hat{\theta}_i \in \Theta_i}\left\{\frac{1}{N} \sum_{j = 1}^N u_{i,M}\left(\theta_i, \hat{\theta}_i, \vec{\theta}_{-i}^{(j)}\right) - u_{i,M}\left(\theta_i, \theta_i, \vec{\theta}_{-i}^{(j)}\right)\right\}\\
 - &\max_{\left(\theta_i, \hat{\theta}_i\right) \in \cG\left(\sample_{-i}, M, \epsilon\right)}\left\{\frac{1}{N} \sum_{j = 1}^N u_{i,M}\left(\theta_i, \hat{\theta}_i, \vec{\theta}_{-i}^{(j)}\right) - u_{i,M}\left(\theta_i, \theta_i, \vec{\theta}_{-i}^{(j)}\right)\right\}\\
 &\leq \epsilon + \tilde O \left(\sqrt{\frac{d_i}{N}}\right).\end{align*} The bound on $\left|\cG\left(\sample_{-i}, M, \epsilon\right)\right|$ follows from Lemma~\ref{lem:cover}.
\end{proof}

\begin{lemma}\label{lem:cover}
For any $\cG$ returned by Algorithm~\ref{alg:greedy}, $|\cG| \leq \left(8eN/\left(\epsilon d_i\right)\right)^{2d_i}$, where $d_i = \pdim\left(\cF_{i, \cM}\right)$.
\end{lemma}

\begin{proof}
Let $\cN_1(\epsilon, U)$ be the $\epsilon$-covering number of $U$ and let $\cP_1(\epsilon, U)$ be the $\epsilon$-packing number of $U$. In other words, $\cN_1(\epsilon, U)$ is the size of the smallest set $V' \subseteq \R^N$ such that for all $\vec{v} \in U$, there is a vector $\vec{v}' \in V'$ such that $\norm{\vec{v} - \vec{v}'}_1 \leq \epsilon$, and $\cP_1(\epsilon, U)$ is the size of the largest set $P \subseteq U$ such that for all $\vec{v}, \vec{v}' \in P$, if $\vec{v} \not= \vec{v}'$, then $\norm{\vec{v} - \vec{v}'}_1 \geq \epsilon$. As \citet{Anthony09:Neural} prove, $\cP(\epsilon, U) \leq \cN(\epsilon/2, U)$.

By construction, for every pair of vectors $\vec{v}$ and $\vec{v}'$ in the set $V$ defined in Algorithm~\ref{alg:greedy}, we know that $\norm{\vec{v} - \vec{v}'}_1 \geq \epsilon$. Therefore, $|V| \leq \cP(\epsilon, U) \leq \cN(\epsilon/2, U)$. In the following claim, we prove that $\cN(\epsilon/2, U) \leq \left(8eN/\left(\epsilon d_i\right)\right)^{2d_i}$, which proves the theorem.

\begin{claim}
$\cN(\epsilon/2, U)\leq\left(8eN/\left(\epsilon d_i\right)\right)^{2d_i}$.
\end{claim}

\begin{proof}
Let $U'$ be the set of vectors \[U' = \left\{\frac{1}{N} \begin{pmatrix} u_{i,M}\left(\theta_i, \hat{\theta}_i, \vec{\theta}_{-i}^{(1)}\right)\\
\vdots\\
u_{i,M}\left(\theta_i, \hat{\theta}_i, \vec{\theta}_{-i}^{(N)}\right)\end{pmatrix}: \theta_i, \hat{\theta}_i \in \Theta_i\right\}.\] Note that $U \subseteq \left\{ \vec{v} - \vec{v}' : \vec{v}, \vec{v}' \in U'\right\}$. We claim that $\cN(\epsilon/2, U) \leq \cN(\epsilon/4, U')^2\leq \left(8eN/\left(\epsilon d_i\right)\right)^{2d_i}$, where the second inequality follows from a theorem by \citet{Anthony09:Neural}. To see why the first inequality holds, suppose $C'$ is an $\frac{\epsilon}{4}$-cover of $U'$ with minimal size. Let $C = \left\{\vec{c} - \vec{c}' : \vec{c}, \vec{c}' \in C'\right\}$. Then for all $\vec{v} \in U$, we know that $\vec{v} = \vec{v}' - \vec{v}''$ for some $\vec{v}', \vec{v}'' \in U'$. Moreover, we know that there are vectors $\vec{c}', \vec{c}'' \in C'$ such that $\norm{\vec{v}' - \vec{c}'}_1 \leq \epsilon/4$ and $\norm{\vec{v}'' - \vec{c}''}_1 \leq \epsilon/4$. Therefore, $\norm{\vec{v} - \left(\vec{c}' - \vec{c}''\right)}_1 = \norm{\left(\vec{v}' - \vec{v}''\right) - \left(\vec{c}' - \vec{c}''\right)}_1 \leq \epsilon/2$. Since $\vec{c}' - \vec{c}'' \in C$, we have that $C$ is a cover of $U$. Therefore, $\cN(\epsilon/2, U) \leq |C| \leq |C'|^2 = \cN(\epsilon/4, U')^2\leq \left(8eN/\left(\epsilon d_i\right)\right)^{2d_i}$.
\end{proof}
\end{proof}

\begin{lemma}\label{lem:grid_dispersion}
Let $\cM$ be a class of mechanisms. For each $i \in [n]$, let $\sample_{-i} =  \left\{\vec{\theta}^{(1)}_{-i}, \dots, \vec{\theta}^{(N)}_{-i}\right\}$ be a set of type profiles for all agents except agent $i$ and let $L_i, k_i, w_i \in \R$ be real values such that for each mechanism $M \in \cM$, the following conditions hold:
\begin{enumerate}
\item For any $\vec{\theta}_i \in [0,1]^m$, the functions $u_{i, M}\left(\vec{\theta}_i, \cdot, \vec{\theta}_{-i}^{(1)}\right), \dots, u_{i, M}\left(\vec{\theta}_i, \cdot, \vec{\theta}_{-i}^{(N)}\right)$ are piecewise $L_i$-Lipschitz and $(w_i,k_i)$-dispersed.
\item For any $\hat{\vec{\theta}}_i \in [0,1]^m$, the functions $u_{i, M}\left(\cdot, \hat{\vec{\theta}}_i, \vec{\theta}_{-i}^{(1)}\right), \dots, u_{i, M}\left(\cdot, \hat{\vec{\theta}}_i, \vec{\theta}_{-i}^{(N)}\right)$ are piecewise $L_i$-Lipschitz and $(w_i,k_i)$-dispersed.
\end{enumerate}
Then for all $M \in \cM$ and $i \in [n]$, \begin{align}\max_{\vec{\theta}_i, \hat{\vec{\theta}}_i \in [0,1]^m}&\left\{\frac{1}{N} \sum_{j = 1}^N u_{i,M}\left(\vec{\theta}_i, \hat{\vec{\theta}}_i, \vec{\theta}_{-i}^{(j)}\right) - u_{i,M}\left(\vec{\theta}_i, \vec{\theta}_i, \vec{\theta}_{-i}^{(j)}\right)\right\}\nonumber\\
\leq \max_{\vec{\theta}_i, \hat{\vec{\theta}}_i \in \cG_{w_i}}&\left\{\frac{1}{N} \sum_{j = 1}^N u_{i,M}\left(\vec{\theta}_i, \hat{\vec{\theta}}_i, \vec{\theta}_{-i}^{(j)}\right) - u_{i,M}\left(\vec{\theta}_i, \vec{\theta}_i, \vec{\theta}_{-i}^{(j)}\right)\right\} + 4L_iw_i + \frac{8k_i}{N}\label{eq:grid_dispersion}.\end{align}
\end{lemma}

\begin{proof}
By definition of dispersion, we know that the following conditions hold:

\paragraph{Condition 1.} For all mechanisms $M \in \cM$, all agents $i \in [n]$, all types $\vec{\theta}_i \in [0,1]^m$, and all reported types $\hat{\vec{\theta}}_i, \hat{\vec{\theta}}_i' \in [0,1]^m$, if $\norm{\hat{\vec{\theta}}_i - \hat{\vec{\theta}}_i'}_2 \leq w_i$, then \begin{equation}\left|\frac{1}{N}\sum_{j = 1}^N u_{i, M}\left(\vec{\theta}_i, \hat{\vec{\theta}}_i, \vec{\theta}_{-i}^{(j)}\right) - u_{i, M}\left(\vec{\theta}_i, \hat{\vec{\theta}}_i', \vec{\theta}_{-i}^{(j)}\right)\right| \leq L_iw_i + \frac{2k_i}{N}.\label{eq:change_b}\end{equation}

\paragraph{Condition 2.} For all mechanisms $M \in \cM$, all agents $i \in [n]$, all reported types $\hat{\vec{\theta}}_i \in [0,1]^m$, and all types $\vec{\theta}_i, \vec{\theta}_i' \in [0,1]^m$, if $\norm{\vec{\theta}_i - \vec{\theta}_i'}_2 \leq w_i$, then \begin{equation}\left|\frac{1}{N}\sum_{j = 1}^N u_{i, M}\left(\vec{\theta}_i, \hat{\vec{\theta}}_i, \vec{\theta}_{-i}^{(j)}\right) - u_{i, M}\left(\vec{\theta}_i', \hat{\vec{\theta}}_i, \vec{\theta}_{-i}^{(j)}\right)\right| \leq L_iw_i + \frac{2k_i}{N}.\label{eq:change_v}\end{equation}

We claim that the inequality from the lemma statement (Equation~\eqref{eq:grid_dispersion}) holds so long as Conditions 1 and 2 hold.

\begin{claim}\label{claim:conditions}If Conditions 1 and 2 hold, then for all $M \in \cM$ and $i \in [n]$, \begin{align*}\max_{\vec{\theta}_i, \hat{\vec{\theta}}_i \in [0,1]^m}&\left\{\frac{1}{N} \sum_{j = 1}^N u_{i,M}\left(\vec{\theta}_i, \hat{\vec{\theta}}_i, \vec{\theta}_{-i}^{(j)}\right) - u_{i,M}\left(\vec{\theta}_i, \vec{\theta}_i, \vec{\theta}_{-i}^{(j)}\right)\right\}\\
\leq \max_{\vec{\theta}_i, \hat{\vec{\theta}}_i \in \cG_{w_i}}&\left\{\frac{1}{N} \sum_{j = 1}^N u_{i,M}\left(\vec{\theta}_i, \hat{\vec{\theta}}_i, \vec{\theta}_{-i}^{(j)}\right) - u_{i,M}\left(\vec{\theta}_i, \vec{\theta}_i, \vec{\theta}_{-i}^{(j)}\right)\right\} + 4L_iw_i + \frac{8k_i}{N}.\end{align*}
\end{claim}

\begin{proof}[Proof of Claim~\ref{claim:conditions}]
Fix an agent $i \in [n]$ and let $\vec{\theta}_i$ and $\hat{\vec{\theta}}_i$ be two fixed, arbitrary vectors in $[0,1]^m$.
By definition of $\cG_w$, there must be a point $\vec{p} \in \cG_w$ such that $\norm{\vec{\theta}_i - \vec{p}}_1 \leq w_i$. By Equation~\eqref{eq:change_b}, \[\left|\frac{1}{N}\sum_{j = 1}^N u_{i, M}\left(\vec{\theta}_i, \vec{\theta}_i, \vec{\theta}_{-i}^{(j)}\right) - u_{i, M}\left(\vec{\theta}_i, \vec{p}, \vec{\theta}_{-i}^{(j)}\right)\right| \leq L_iw_i + \frac{2k_i}{N}\] and by Equation~\eqref{eq:change_v}, \[\left|\frac{1}{N}\sum_{j = 1}^N u_{i, M}\left(\vec{\theta}_i, \vec{p}, \vec{\theta}_{-i}^{(j)}\right) - u_{i, M}\left(\vec{p},\vec{p}, \vec{\theta}_{-i}^{(j)}\right)\right| \leq L_iw_i + \frac{2k_i}{N}.\]

Similarly, there must be a point $\hat{\vec{p}} \in \cG_w$ such that $\norm{\hat{\vec{p}} - \hat{\vec{\theta}}_i}_1 \leq w_i$. By Equation~\eqref{eq:change_b}, \[\left|\frac{1}{N}\sum_{j = 1}^N u_{i, M}\left(\vec{\theta}_i, \hat{\vec{\theta}}_i, \vec{\theta}_{-i}^{(j)}\right) - u_{i, M}\left(\vec{\theta}_i, \hat{\vec{p}}, \vec{\theta}_{-i}^{(j)}\right)\right| \leq L_iw_i + \frac{2k_i}{N}\] and by Equation~\eqref{eq:change_v}, \[\left|\frac{1}{N}\sum_{j = 1}^N u_{i, M}\left(\vec{\theta}_i, \hat{\vec{p}}, \vec{\theta}_{-i}^{(j)}\right) - u_{i, M}\left(\vec{p}, \hat{\vec{p}}, \vec{\theta}_{-i}^{(j)}\right)\right| \leq L_iw_i + \frac{2k_i}{N}.\]

Therefore, \begin{align*}&\left|\frac{1}{N}\sum_{j = 1}^N u_{i, M}\left(\vec{\theta}_i, \vec{\theta}_i, \vec{\theta}_{-i}^{(j)}\right) - u_{i, M}\left(\vec{\theta}_i, \hat{\vec{\theta}}_i, \vec{\theta}_{-i}^{(j)}\right)\right|\\
\leq \text{ }&\left|\frac{1}{N}\sum_{j = 1}^Nu_{i, M}\left(\vec{\theta}_i, \vec{\theta}_i, \vec{\theta}_{-i}^{(j)}\right) - u_{i, M}\left(\vec{\theta}_i, \vec{p}, \vec{\theta}_{-i}^{(j)}\right)\right|\\
+\text{ }& \left|\frac{1}{N}\sum_{j = 1}^Nu_{i, M}\left(\vec{\theta}_i, \vec{p}, \vec{\theta}_{-i}^{(j)}\right) -u_{i, M}\left(\vec{p}, \vec{p}, \vec{\theta}_{-i}^{(j)}\right)\right|\\
+\text{ }&\left|\frac{1}{N}\sum_{j = 1}^Nu_{i, M}\left(\vec{p}, \vec{p}, \vec{\theta}_{-i}^{(j)}\right) - u_{i, M}\left(\vec{p}, \hat{\vec{p}}, \vec{\theta}_{-i}^{(j)}\right)\right|\\
+ \text{ }&\left|\frac{1}{N}\sum_{j = 1}^Nu_{i, M}\left(\vec{p}, \hat{\vec{p}}, \vec{\theta}_{-i}^{(j)}\right) - u_{i, M}\left(\vec{\theta}_i, \hat{\vec{p}}, \vec{\theta}_{-i}^{(j)}\right)\right|\\
+\text{ }&\left|\frac{1}{N}\sum_{j = 1}^Nu_{i, M}\left(\vec{\theta}_i, \hat{\vec{p}}, \vec{\theta}_{-i}^{(j)}\right) - u_{i, M}\left(\vec{\theta}_i, \hat{\vec{\theta}}_i, \vec{\theta}_{-i}^{(j)}\right)\right|\\
\leq \text{ }&\left|\frac{1}{N}\sum_{j = 1}^Nu_{i, M}\left(\vec{p}, \vec{p}, \vec{\theta}_{-i}^{(j)}\right) - u_{i, M}\left(\vec{p}, \hat{\vec{p}}, \vec{\theta}_{-i}^{(j)}\right)\right| + 4L_iw_i + \frac{8k_i}{N}.\end{align*} Since $\vec{p} , \hat{\vec{p}} \in \cG_{w_i}$, the claim holds.
\end{proof}

The lemma thus follows from Claim~\ref{claim:conditions}.
\end{proof}

\begin{theorem}\label{thm:main}
Let $\cM$ be a mechanism class. Given $n$ sets of samples $\sample_{-i} =  \left\{\vec{\theta}^{(1)}_{-i}, \dots, \vec{\theta}^{(N)}_{-i}\right\} \sim \dist_{-i}^N$, for each $i \in [n]$ let $L_i, k_i, w_i \in \R$ be defined such that for each $M \in \cM$, the following conditions hold:
\begin{enumerate}
\item For any $\vec{\theta}_i \in [0,1]^m$, the functions $u_{i, M}\left(\vec{\theta}_i, \cdot, \vec{\theta}_{-i}^{(1)}\right), \dots, u_{i, M}\left(\vec{\theta}_i, \cdot, \vec{\theta}_{-i}^{(N)}\right)$ are piecewise $L_i$-Lipschitz and $(w_i,k_i)$-dispersed.
\item For any $\hat{\vec{\theta}}_i \in [0,1]^m$, the functions $u_{i, M}\left(\cdot, \hat{\vec{\theta}}_i, \vec{\theta}_{-i}^{(1)}\right), \dots, u_{i, M}\left(\cdot, \hat{\vec{\theta}}_i, \vec{\theta}_{-i}^{(N)}\right)$ are piecewise $L_i$-Lipschitz and $(w_i,k_i)$-dispersed.
\end{enumerate} Then with probability $1-\delta$, for every $M \in \cM$, every agent $i \in [n]$, and every pair $\theta_i, \hat{\theta}_i \in \Theta_i$, \begin{align*}&\E_{\vec{\theta}_{-i} \sim \dist_{-i}}\left[u_{i,M}\left(\vec{\theta}_i, \hat{\vec{\theta}}_i, \vec{\theta}_{-i}\right) - u_{i,M}\left(\vec{\theta}_i, \vec{\theta}_i, \vec{\theta}_{-i}\right)\right]\nonumber\\
 \leq &\max_{\vec{\theta}_i, \hat{\vec{\theta}}_i \in \cG_{w_i}}\left\{\frac{1}{N} \sum_{j = 1}^N u_{i,M}\left(\vec{\theta}_i, \hat{\vec{\theta}}_i, \vec{\theta}_{-i}^{(j)}\right) - u_{i,M}\left(\vec{\theta}_i, \vec{\theta}_i, \vec{\theta}_{-i}^{(j)}\right)\right\} + 4L_iw_i + \frac{8k_i}{N} + \epsilon_{\mclass}(N, \delta)\end{align*} where $\epsilon_{\mclass}(N, \delta) = 2\sqrt{\frac{2 d_i}{N} \ln \frac{eN}{d_i}} + 2\sqrt{\frac{1}{2N}\ln\frac{2n}{\delta}}$ and $d_i = \pdim\left(\cF_{i, \cM}\right)$.
\end{theorem}

\begin{proof}
From Theorem~\ref{thm:AIC}, we know that with probability $1-\delta$, for every mechanism $M \in \cM$ and every agent $i \in [n]$, \begin{align*}&\max_{\vec{\theta}_i, \hat{\vec{\theta}}_i \in [0,1]^m}\left\{\E_{\vec{\theta}_{-i} \sim \dist_{-i}}\left[u_{i,M}\left(\vec{\theta}_i, \hat{\vec{\theta}}_i, \vec{\theta}_{-i}\right) - u_{i,M}\left(\vec{\theta}_i, \vec{\theta}_i, \vec{\theta}_{-i}\right)\right]\right\}\\
 \leq &\max_{\vec{\theta}_i, \hat{\vec{\theta}}_i \in [0,1]^m}\left\{\frac{1}{N} \sum_{j = 1}^N u_{i,M}\left(\vec{\theta}_i, \hat{\vec{\theta}}_i, \vec{\theta}_{-i}^{(j)}\right) - u_{i,M}\left(\vec{\theta}_i, \vec{\theta}_i, \vec{\theta}_{-i}^{(j)}\right)\right\} + \epsilon_{\mclass}(N, \delta).\end{align*}
Moreover, from Lemma~\ref{lem:grid_dispersion}, we know that the following inequality holds:
\begin{align*}\max_{\vec{\theta}_i, \hat{\vec{\theta}}_i \in [0,1]^m}&\left\{\frac{1}{N} \sum_{j = 1}^N u_{i,M}\left(\vec{\theta}_i, \hat{\vec{\theta}}_i, \vec{\theta}_{-i}^{(j)}\right) - u_{i,M}\left(\vec{\theta}_i, \vec{\theta}_i, \vec{\theta}_{-i}^{(j)}\right)\right\}\nonumber\\
\leq \max_{\vec{\theta}_i, \hat{\vec{\theta}}_i \in \cG_{w_i}}&\left\{\frac{1}{N} \sum_{j = 1}^N u_{i,M}\left(\vec{\theta}_i, \hat{\vec{\theta}}_i, \vec{\theta}_{-i}^{(j)}\right) - u_{i,M}\left(\vec{\theta}_i, \vec{\theta}_i, \vec{\theta}_{-i}^{(j)}\right)\right\} + 4L_iw_i + \frac{8k_i}{N}.\end{align*}

Therefore, with probability $1-\delta$, for every mechanism $M \in \cM$ and every agent $i \in [n]$, \begin{align*}&\max_{\vec{\theta}_i, \hat{\vec{\theta}}_i \in [0,1]^m}\left\{\E_{\vec{\theta}_{-i} \sim \dist_{-i}}\left[u_{i,M}\left(\vec{\theta}_i, \hat{\vec{\theta}}_i, \vec{\theta}_{-i}\right) - u_{i,M}\left(\vec{\theta}_i, \vec{\theta}_i, \vec{\theta}_{-i}\right)\right]\right\}\nonumber\\
 = &\max_{\vec{\theta}_i, \hat{\vec{\theta}}_i \in [0,1]^m}\left\{\E_{\vec{\theta}_{-i} \sim \dist_{-i}}\left[u_{i,M}\left(\vec{\theta}_i, \hat{\vec{\theta}}_i, \vec{\theta}_{-i}\right) - u_{i,M}\left(\vec{\theta}_i, \vec{\theta}_i, \vec{\theta}_{-i}\right)\right]\right\}\\
- &\max_{\vec{\theta}_i, \hat{\vec{\theta}}_i \in [0,1]^m}\left\{\frac{1}{N} \sum_{j = 1}^N u_{i,M}\left(\vec{\theta}_i, \hat{\vec{\theta}}_i, \vec{\theta}_{-i}^{(j)}\right) - u_{i,M}\left(\vec{\theta}_i, \vec{\theta}_i, \vec{\theta}_{-i}^{(j)}\right)\right\}\\
+ &\max_{\vec{\theta}_i, \hat{\vec{\theta}}_i \in [0,1]^m}\left\{\frac{1}{N} \sum_{j = 1}^N u_{i,M}\left(\vec{\theta}_i, \hat{\vec{\theta}}_i, \vec{\theta}_{-i}^{(j)}\right) - u_{i,M}\left(\vec{\theta}_i, \vec{\theta}_i, \vec{\theta}_{-i}^{(j)}\right)\right\}\\
- &\max_{\vec{\theta}_i, \hat{\vec{\theta}}_i \in \cG_{w_i}}\left\{\frac{1}{N} \sum_{j = 1}^N u_{i,M}\left(\vec{\theta}_i, \hat{\vec{\theta}}_i, \vec{\theta}_{-i}^{(j)}\right) - u_{i,M}\left(\vec{\theta}_i, \vec{\theta}_i, \vec{\theta}_{-i}^{(j)}\right)\right\}\\
+ &\max_{\vec{\theta}_i, \hat{\vec{\theta}}_i \in \cG_{w_i}}\left\{\frac{1}{N} \sum_{j = 1}^N u_{i,M}\left(\vec{\theta}_i, \hat{\vec{\theta}}_i, \vec{\theta}_{-i}^{(j)}\right) - u_{i,M}\left(\vec{\theta}_i, \vec{\theta}_i, \vec{\theta}_{-i}^{(j)}\right)\right\}\\
\leq&\max_{\vec{\theta}_i, \hat{\vec{\theta}}_i \in \cG_{w_i}}\left\{\frac{1}{N} \sum_{j = 1}^N u_{i,M}\left(\vec{\theta}_i, \hat{\vec{\theta}}_i, \vec{\theta}_{-i}^{(j)}\right) - u_{i,M}\left(\vec{\theta}_i, \vec{\theta}_i, \vec{\theta}_{-i}^{(j)}\right)\right\} + 4L_iw_i + \frac{8k_i}{N} + \epsilon_{\mclass}(N, \delta).\end{align*}
\end{proof}

\dispersion*

\begin{proof}
By definition of dispersion, we know that with probability $1-\delta$ over the draw of the $n$ sets $\sample_{-i} = \left\{\vec{\theta}^{(1)}_{-i}, \dots, \vec{\theta}^{(N)}_{-i}\right\} \sim \dist_{-i}^N$, the following conditions hold:

\paragraph{Condition 1.} For all mechanisms $M \in \cM$, all agents $i \in [n]$, all types $\vec{\theta}_i \in [0,1]^m$, and all reported types $\hat{\vec{\theta}}_i, \hat{\vec{\theta}}_i' \in [0,1]^m$, if $\norm{\hat{\vec{\theta}}_i - \hat{\vec{\theta}}_i'}_1 \leq w_i$, then \begin{equation}\left|\frac{1}{N}\sum_{j = 1}^N u_{i, M}\left(\vec{\theta}_i, \hat{\vec{\theta}}_i, \vec{\theta}_{-i}^{(j)}\right) - u_{i, M}\left(\vec{\theta}_i, \hat{\vec{\theta}}_i', \vec{\theta}_{-i}^{(j)}\right)\right| \leq L_iw_i + \frac{2k_i}{N}.\label{eq:change_b_prob}\end{equation}

\paragraph{Condition 2.} For all mechanisms $M \in \cM$, all agents $i \in [n]$, all reported types $\hat{\vec{\theta}}_i \in [0,1]^m$, and all types $\vec{\theta}_i, \vec{\theta}_i' \in [0,1]^m$, if $\norm{\vec{\theta}_i - \vec{\theta}_i'}_1 \leq w_i$, then \begin{equation}\left|\frac{1}{N}\sum_{j = 1}^N u_{i, M}\left(\vec{\theta}_i, \hat{\vec{\theta}}_i, \vec{\theta}_{-i}^{(j)}\right) - u_{i, M}\left(\vec{\theta}_i', \hat{\vec{\theta}}_i, \vec{\theta}_{-i}^{(j)}\right)\right| \leq L_iw_i + \frac{2k_i}{N}.\label{eq:change_v_prob}\end{equation}

From Claim~\ref{claim:conditions}, we know that Equation~\eqref{eq:grid_dispersion_prob} holds so long as Conditions 1 and 2 hold. Since Conditions 1 and 2 hold with probability $1-\delta$, the lemma statement holds.
\end{proof}

\main*

\begin{proof}
Theorem~\ref{thm:AIC} guarantees that with probability at most $\delta$, the extent to which any agent $i$ can improve his utility by misreporting his type, averaged over all profiles in $\sample_{-i}$, does not approximate the true incentive compatibility approximation factor, as summarized below:

\paragraph{Bad event 1.} For some mechanism $M \in \cM$ and agent $i \in [n]$,\begin{align*}&\max_{\vec{\theta}_i, \hat{\vec{\theta}}_i \in \Theta_i}\left\{\E_{\vec{\theta}_{-i} \sim \dist_{-i}}\left[u_{i,M}\left(\vec{\theta}_i, \hat{\vec{\theta}}_i, \vec{\theta}_{-i}\right) - u_{i,M}\left(\vec{\theta}_i,\vec{\theta}_i, \vec{\theta}_{-i}\right)\right]\right\}\nonumber\\
 > &\max_{\vec{\theta}_i, \hat{\vec{\theta}}_i \in [0,1]^m}\left\{\frac{1}{N} \sum_{j = 1}^N u_{i,M}\left(\vec{\theta}_i,\hat{\vec{\theta}}_i, \vec{\theta}_{-i}^{(j)}\right)
 - u_{i,M}\left(\vec{\theta}_i, \vec{\theta}_i, \vec{\theta}_{-i}^{(j)}\right)\right\} + \epsilon_{\cM}(N, \delta),\end{align*} where $\epsilon_{\cM}(N, \delta) = 2\sqrt{\frac{2 d_i}{N} \ln \frac{eN}{d_i}} + 2\sqrt{\frac{1}{2N}\ln\frac{2n}{\delta}}$ and $d_i = \pdim\left(\cF_{i, \cM}\right)$.
 
By Lemma~\ref{lem:grid_dispersion_prob}, we also know that with probability at most $\delta$, we cannot approximate Equation~\eqref{eq:full_type} by discretizing the agent's type space, as summarized by the following bad event:
 
 \paragraph{Bad event 2.} For some mechanism $M \in \cM$ and agent $i \in [n]$, \begin{align*}\max_{\vec{\theta}_i, \hat{\vec{\theta}}_i \in [0,1]^m}&\left\{\frac{1}{N} \sum_{j = 1}^N u_{i,M}\left(\vec{\theta}_i,\hat{\vec{\theta}}_i, \vec{\theta}_{-i}^{(j)}\right)
 - u_{i,M}\left(\vec{\theta}_i, \vec{\theta}_i, \vec{\theta}_{-i}^{(j)}\right)\right\}\nonumber\\
> \max_{\vec{\theta}_i, \hat{\vec{\theta}}_i \in \cG_{w_i}}&\left\{\frac{1}{N} \sum_{j = 1}^N u_{i,M}\left(\vec{\theta}_i,\hat{\vec{\theta}}_i, \vec{\theta}_{-i}^{(j)}\right)
 - u_{i,M}\left(\vec{\theta}_i, \vec{\theta}_i, \vec{\theta}_{-i}^{(j)}\right)\right\} + 4L_iw_i + \frac{8k_i}{N}.\end{align*}

By a union bound, the probability that either Bad Event 1 or Bad Event 2 occurs is at most $2\delta$. Therefore, the probability that neither occurs is at least $1-2\delta$. If neither occurs, then for every mechanism $M \in \cM$ and every agent $i \in [n]$, \begin{align*}&\max_{\vec{\theta}_i, \hat{\vec{\theta}}_i \in [0,1]^m}\left\{\E_{\vec{\theta}_{-i} \sim \dist_{-i}}\left[u_{i,M}\left(\vec{\theta}_i, \hat{\vec{\theta}}_i, \vec{\theta}_{-i}\right) - u_{i,M}\left(\vec{\theta}_i, \vec{\theta}_i, \vec{\theta}_{-i}\right)\right]\right\}\\
 \leq &\max_{\vec{\theta}_i, \hat{\vec{\theta}}_i \in \cG_{w_i}}\left\{\frac{1}{N} \sum_{j = 1}^N u_{i,M}\left(\vec{\theta}_i, \hat{\vec{\theta}}_i, \vec{\theta}_{-i}^{(j)}\right) - u_{i,M}\left(\vec{\theta}_i, \vec{\theta}_i, \vec{\theta}_{-i}^{(j)}\right)\right\} + \epsilon,\end{align*} where $\epsilon = 4L_iw_i + \frac{4k_i}{N} + 2\sqrt{\frac{2 d_i}{N} \ln \frac{eN}{d_i}} + 2\sqrt{\frac{1}{2N}\ln\frac{2n}{\delta}}$ and $d_i = \pdim\left(\cF_{i, \cM}\right)$.
\end{proof}

\section{Proofs about dispersion and pseudo-dimension (Section~\ref{sec:applications})}\label{app:applications}
\FP*

\begin{proof}
Consider an arbitrary bidder, and without loss of generality, suppose that bidder is bidder 1. Next, choose an arbitrary sample $\vec{\theta}_{-1}^{(j)} = \left(\theta_2^{(j)}, \dots, \theta_n^{(j)}\right)$. For any value $\theta_1 \in [0,1]$ and bid $\hat{\theta}_1 \in [0,1]$, we know that $u_{1, M}\left(\theta_1, \hat{\theta}_1, \vec{\theta}_{-1}^{(j)}\right) = \textbf{1}_{\left\{\hat{\theta}_1 > \norm{\vec{\theta}_{-1}^{(j)}}_{\infty}\right\}} \left(\theta_1 - \hat{\theta}_1\right)$, as illustrated in Figure~\ref{fig:FP}. Therefore, if $\hat{\theta}_1 \leq \norm{\vec{\theta}_{-1}^{(j)}}_{\infty}$, then $u_{1, M}\left(\theta_1, \hat{\theta}_1, \vec{\theta}_{-1}^{(j)}\right)$ is a constant function of $\hat{\theta}_1$, whereas if $\hat{\theta}_1 > \norm{\vec{\theta}_{-1}^{(j)}}_{\infty}$, then $u_{1, M}\left(\theta_1, \hat{\theta}_1, \vec{\theta}_{-1}^{(j)}\right)$ is a linear function of $\hat{\theta}_1$ with a slope of $-1$. Therefore, for all $\theta_i \in [0,1]$, the function $u_{1, M}\left(\theta_1, \cdot, \vec{\theta}_{-1}^{(j)}\right)$ is piecewise $1$-Lipschitz with a discontinuity at $\norm{\vec{\theta}_{-1}^{(j)}}_{\infty}$.

We now prove that with probability $1 - \frac{\delta}{n}$, for all $\theta_1 \in [0,1]$, the 
functions \[u_{1, M}\left(\theta_1, \cdot, \vec{\theta}_{-1}^{(1)}\right), \dots, u_{1, M}\left(\theta_1, \cdot, \vec{\theta}_{-1}^{(N)}\right)\] are $(w,k)$-dispersed. Since for any $\theta_1 \in [0,1]$ the 
function $u_{1, M}\left(\theta_1, \hat{\theta}_1, \vec{\theta}_{-1}^{(j)}\right)$ will only have a 
discontinuity at a point in the set $\left\{\theta_2^{(j)}, \dots, \theta_n^{(j)}\right\}$, it is 
enough to prove that with probability $1-\frac{\delta}{n}$, at most $k$ points in the set $\cB = \bigcup_{j = 1}^N \left\{\theta_2^{(j)}, \dots, \theta_n^{(j)}\right\}$ fall within any interval of 
width $w$. The theorem statement then holds by a union bound over all $n$ bidders.

\begin{claim}\label{claim:fp_dispersed}
With probability $1-\frac{\delta}{n}$, at most $k$ points in the set $\bigcup_{j = 1}^N \left\{\theta_2^{(j)}, \dots, \theta_n^{(j)}\right\}$ fall within any interval of 
width $w$.
\end{claim}

\begin{proof}[Proof of Claim~\ref{claim:fp_dispersed}]
For $i \in \{2, \dots, n\}$, let $\cB_i = \left\{\theta_i^{(j)}\right\}_{j \in [N]}$. The claim follows from Lemma~\ref{lem:dispersion} in Appendix~\ref{app:dispersion}.
\end{proof}
\end{proof}

\FPpdim*

\begin{proof}
First, we prove that $\pdim\left(\cF_{i, M}\right) \leq 2.$

\begin{claim}\label{claim:fp_pdim_ub}
The pseudo-dimension of $\cF_{i, M}$ is at most 2.
\end{claim}

\begin{proof}[Proof of Claim~\ref{claim:fp_pdim_ub}]
For a contradiction, suppose there exists a set $\sample_{-i} = \left\{\vec{\theta}_{-i}^{(1)}, \vec{\theta}_{-i}^{(2)}, \vec{\theta}_{-i}^{(3)}\right\}$ that is shattered by $\cF_{i, M}$. Without loss of generality, assume that $\norm{\vec{\theta}_{-i}^{(1)}}_{\infty} < \norm{\vec{\theta}_{-i}^{(2)}}_{\infty} < \norm{\vec{\theta}_{-i}^{(3)}}_{\infty}$. Since $\sample_{-i}$ is shatterable, there exists three values $z^{(1)}, z^{(2)}, z^{(3)} \in \R$ that witness the shattering of $\sample_{-i}$ by $\cF_{i, M}$. We split the proof into two cases: one where $z^{(3)} > 0$ and one where $z^{(3)} \leq 0$.

\paragraph{Case 1: $z^{(3)} > 0$.} Since $\sample_{-i}$ is shatterable, there is a value $\theta_i \in [0,1]$ and bid $\hat{\theta}_i \in [0,1]$ such that $u_{i, M}\left(\theta_i, \hat{\theta}_i, \vec{\theta}_{-i}^{(1)}\right) \geq z^{(1)}$, $u_{i, M}\left(\theta_i, \hat{\theta}_i, \vec{\theta}_{-i}^{(2)}\right) < z^{(2)}$, and $u_{i, M}\left(\theta_i, \hat{\theta}_i, \vec{\theta}_{-i}^{(3)}\right) \geq z^{(3)}.$ Recall that for any $\vec{\theta}_{-i} \in [0,1]^{n-1}$, $u_{i, M}\left(\theta_i, \hat{\theta}_i, \vec{\theta}_{-i}\right) = \textbf{1}_{\left\{\hat{\theta}_i > \norm{\vec{\theta}_{-i}}_{\infty}\right\}} \left(\theta_i - \hat{\theta}_i\right).$ Since $u_{i, M}\left(\theta_i, \hat{\theta}_i, \vec{\theta}_{-i}^{(3)}\right) \geq z^{(3)} > 0$, it must be that $\hat{\theta}_i > \norm{\vec{\theta}_{-i}^{(3)}}_{\infty} > \norm{\vec{\theta}_{-i}^{(2)}}_{\infty} > \norm{\vec{\theta}_{-i}^{(1)}}_{\infty}$. Therefore, $\theta_i - \hat{\theta}_i \geq z^{(1)}$ and $\theta_i-\hat{\theta}_i < z^{(2)}$, which means that $z^{(1)} < z^{(2)}$.

Next, there must also be a value $\theta_i' \in [0,1]$ and bid $\hat{\theta}_i' \in [0,1]$ such that $u_{i, M}\left(\theta_i', \hat{\theta}_i', \vec{\theta}_{-i}^{(1)}\right) < z^{(1)},$ $u_{i, M}\left(\theta_i', \hat{\theta}_i', \vec{\theta}_{-i}^{(2)}\right) \geq z^{(2)}$, and $u_{i, M}\left(\theta_i', \hat{\theta}_i', \vec{\theta}_{-i}^{(3)}\right) \geq z^{(3)}$. Again, since $u_{i, M}\left(\theta_i', \hat{\theta}_i', \vec{\theta}_{-i}^{(3)}\right) \geq z^{(3)} > 0,$ it must be that $\hat{\theta}_i' > \norm{\vec{\theta}_{-i}^{(3)}}_{\infty} > \norm{\vec{\theta}_{-i}^{(2)}}_{\infty} > \norm{\vec{\theta}_{-i}^{(1)}}_{\infty}$. Therefore, $\theta_i' - \hat{\theta}_i' < z^{(1)}$ and $\theta_i'-\hat{\theta}_i' \geq z^{(2)}$, which means that $z^{(2)} > z^{(1)}$, which is a contradiction.

\paragraph{Case 2: $z^{(3)} \leq 0$.} Since $\sample_{-i}$ is shatterable, there is a value $\theta_i \in [0,1]$ and bid $\hat{\theta}_i \in [0,1]$ such that $u_{i, M}\left(\theta_i, \hat{\theta}_i, \vec{\theta}_{-i}^{(1)}\right) \geq z^{(1)}$, $u_{i, M}\left(\theta_i, \hat{\theta}_i, \vec{\theta}_{-i}^{(2)}\right) < z^{(2)}$, and $u_{i, M}\left(\theta_i, \hat{\theta}_i, \vec{\theta}_{-i}^{(3)}\right) < z^{(3)}.$ Since $u_{i, M}\left(\theta_i, \hat{\theta}_i, \vec{\theta}_{-i}^{(3)}\right) < z^{(3)} \leq 0$, it must be that $\hat{\theta}_i > \norm{\vec{\theta}_{-i}^{(3)}}_{\infty} > \norm{\vec{\theta}_{-i}^{(2)}}_{\infty} > \norm{\vec{\theta}_{-i}^{(1)}}_{\infty}$. Therefore, $\theta_i - \hat{\theta}_i \geq z^{(1)}$ and $\theta_i-\hat{\theta}_i < z^{(2)}$, which means that $z^{(1)} < z^{(2)}$.

Next, there must also be a value $\theta_i' \in [0,1]$ and bid $\hat{\theta}_i' \in [0,1]$ such that $u_{i, M}\left(\theta_i', \hat{\theta}_i', \vec{\theta}_{-i}^{(1)}\right) < z^{(1)},$ $u_{i, M}\left(\theta_i', \hat{\theta}_i', \vec{\theta}_{-i}^{(2)}\right) \geq z^{(2)}$, and $u_{i, M}\left(\theta_i', \hat{\theta}_i', \vec{\theta}_{-i}^{(3)}\right) < z^{(3)}$. Again, since $u_{i, M}\left(\theta_i', \hat{\theta}_i', \vec{\theta}_{-i}^{(3)}\right) < z^{(3)} \leq 0,$ it must be that $\hat{\theta}_i' > \norm{\vec{\theta}_{-i}^{(3)}}_{\infty} > \norm{\vec{\theta}_{-i}^{(2)}}_{\infty} > \norm{\vec{\theta}_{-i}^{(1)}}_{\infty}$. Therefore, $\theta_i' - \hat{\theta}_i' < z^{(1)}$ and $\theta_i'-\hat{\theta}_i' \geq z^{(2)}$, which means that $z^{(2)} > z^{(1)}$, which is a contradiction.

Since we arrive at a contradiction in both cases, we conclude that $\pdim\left(\cF_{i, M}\right) \leq 2.$
\end{proof}

We now prove that $\pdim\left(\cF_{i, M}\right) \geq 2.$

\begin{claim}\label{claim:fp_pdim_lb}
The pseudo-dimension of $\cF_{i, M}$ is at least 2.
\end{claim}

\begin{proof}[Proof of Claim~\ref{claim:fp_pdim_lb}]
To prove this claim, we exhibit a set of size 2 that is shattered by $\cF_{i, M}$. Let $\vec{\theta}_{-i}^{(1)}$ and $\vec{\theta}_{-i}^{(2)}$ be two sets of values such that $\norm{\vec{\theta}_{-i}^{(1)}}_{\infty} = 1/3$ and $\norm{\vec{\theta}_{-i}^{(2)}}_{\infty} = 2/3$. Table~\ref{tab:shatter} illustrates that $z^{(1)} = 3/16$ and $z^{(2)} = 1/16$ witness the shattering of $\left\{\vec{\theta}_{-i}^{(1)}, \vec{\theta}_{-i}^{(2)}\right\}$ by $\cF_{i, M}$.
\begin{table}
\centering
\begin{tabular}{@{}llll@{}}
\toprule
$\theta_i$ & $\hat{\theta}_i$ & $u_{i, M}\left(\theta_i, \hat{\theta}_i, \vec{\theta}_{-i}^{(1)}\right)$ & $u_{i, M}\left(\theta_i, \hat{\theta}_i, \vec{\theta}_{-i}^{(2)}\right)$ \\ \midrule
1   & 1   & 0                                       & 0                                       \\\cmidrule{2-4}
&7/8   & 1/8                                     & 1/8                                     \\\cmidrule{2-4}
&3/4   & 1/4                                     & 1/4                                     \\\cmidrule{2-4}
&1/2   & 1/2                                     & 0                                       \\ \bottomrule
\end{tabular}
\caption{Shattering of the example from Claim~\ref{claim:fp_pdim_lb}.}\label{tab:shatter}
\end{table}
\end{proof}

Together, Claims~\ref{claim:fp_pdim_ub} and \ref{claim:fp_pdim_lb} prove that $\pdim\left(\cF_{i, M}\right) = 2.$
\end{proof}

\FPcomb*

\begin{proof}
Fix an arbitrary agent, and without loss of generality, suppose that agent is agent 1. Next, fix an arbitrary sample $\vec{\theta}^{(j)}_{-1}$ and pair of allocations $\left(b_1, \dots, b_n\right)$ and $\left(b_1', \dots, b_n'\right)$. We know that $\left(b_1', \dots, b_n'\right)$ will \emph{not} be the allocation of the first-price combinatorial auction so long as agent 1's reported type $\hat{\vec{\theta}}_1 \in [0,1]^{2^{\ell}}$ is chosen such that \[\hat{\theta}_1(b_1) + \sum_{i = 2}^n \theta^{(j)}_i(b_i) > \hat{\theta}_1(b_1') + \sum_{i = 2}^n \theta^{(j)}_i(b_i').\] This means that across all $\vec{\theta}_1 \in [0,1]^{2^{\ell}}$, there is fixed a set $\Psi_j$ of ${(n+1)^\ell \choose 2}$ hyperplanes (one per pair of allocations) such that for any connected component $C$ of $[0,1]^{2^{\ell}} \setminus \Psi_j$, the allocation of the first-price combinatorial auction is invariant across all $\hat{\vec{\theta}}_1 \in C$. Namely, \[\Psi_j = \left\{\hat{\theta}_1(b_1) - \hat{\theta}_1(b_1') = \sum_{i = 2}^n \theta^{(j)}_i(b_i') - \theta^{(j)}_i(b_i) : \left(b_1, \dots, b_n\right) \text{ and }\left(b_1', \dots, b_n'\right) \text{ are allocations}\right\}.\] So long as the allocation is fixed, agent 1's utility is 1-Lipschitz in $\hat{\vec{\theta}}_1$.

For each pair of allocations $\vec{b} = \left(b_1, \dots, b_n\right)$ and $\vec{b}' = \left(b_1', \dots, b_n'\right)$, let \[\cB_{\vec{b}, \vec{b}'} = \left\{\hat{\theta}_1(b_1) - \hat{\theta}_1(b_1') = \sum_{i = 2}^n \theta^{(j)}_i(b_i') - \theta^{(j)}_i(b_i) : j \in [N]\right\}.\] Within each set, the hyperplanes are parallel with probability 1. By Lemmas~\ref{lem:sum_bounded} and \ref{lem:difference_bounded} in Appendix~\ref{app:dispersion}, their offsets are independently drawn from distributions with $\kappa$-bounded density functions. Therefore, by Lemma~\ref{lem:hyperplanes} in Appendix~\ref{app:dispersion}, with probability $1 - \frac{\delta}{n}$, for all $\vec{\theta}_1 \in [0,1]^{2^{\ell}}$, the functions $u_{1, M}\left(\vec{\theta}_1, \cdot, \vec{\theta}_{-1}^{(1)}\right), \dots, u_{1, M}\left(\vec{\theta}_1, \cdot, \vec{\theta}_{-1}^{(N)}\right)$ are $(w,k)$-dispersed with $w = O\left(\frac{1}{\kappa \sqrt{N}}\right)$ and $k = O\left((n+1)^{2\ell} \sqrt{N \log\frac{n(n+1)^{2\ell}}{\delta}}\right).$ The theorem statement holds by a union bound over all $n$ agents.
\end{proof}

\GSP*

\begin{proof}
Let $w = O\left(\frac{1}{r \kappa \sqrt{N}}\right)$ and $k = O\left(n^{3/2} \sqrt{N \log\frac{nr}{\delta}}\right)$. Consider an arbitrary agent, and without loss of generality, suppose that agent is agent 1. Next, choose an arbitrary sample $\vec{\theta}_{-1}^{(j)} = \left(\theta_2^{(j)}, \dots, \theta_n^{(j)}\right)$ and mechanism $M \in \cM_r$ with agent weights $\vec{\omega} = \left(\omega_1, \dots, \omega_n\right) \in \left\{\frac{1}{r}, \frac{2}{r}, \dots, 1\right\}^n$. Let $\omega_{i_1}\theta_{i_1}^{(j)} \leq \cdots \leq \omega_{i_{n-1}}\theta_{i_{n-1}}^{(j)}$ be the weighted types of all agents except agent $1$. Consider the $n$ intervals delineated by these $n-1$ weighted values: $I_1 = \left[0, \omega_{i_1}\theta_{i_1}^{(j)}\right), \dots, I_n = \left[\omega_{i_{n-1}}\theta_{i_{n-1}}^{(j)},1\right]$ with \[I_\tau = \left[\omega_{i_{\tau-1}}\theta_{i_{\tau-1}}^{(j)}, \omega_{i_{\tau}}\theta_{i_{\tau}}^{(j)}\right).\] Suppose we vary $\hat{\theta}_1$ in such a way that $\omega_1\hat{\theta}_1$ remains within a single interval $I_{\tau}$. Said another way, consider varying $\hat{\theta}_1$ over the interval \[\left[\frac{\omega_{i_{\tau-1}}\theta_{i_{\tau-1}}^{(j)}}{\omega_1}, \frac{\omega_{i_{\tau}}\theta_{i_{\tau}}^{(j)}}{\omega_1}\right).\] The allocation will be constant, which means agent 1's utility will be a constant function of $\hat{\theta}_1$. Therefore, no matter the value of $\theta_1$, $u_{1, M}\left(\theta_1, \cdot, \vec{\theta}_{-1}^{(1)}\right), \dots, u_{1, M}\left(\theta_1, \cdot, \vec{\theta}_{-1}^{(N)}\right)$ are piecewise constant with discontinuities in the set \[\cB_{1, \vec{\omega}} = \left\{\frac{\omega_{i'}\theta_{i'}^{(j)}}{\omega_1} : i' \in \{2, \dots, n\}, j \in [N]\right\}.\]
This means that in order to prove the theorem, it is enough to show that with probability $1-\delta$, for all $\vec{\omega} \in \left\{\frac{1}{r}, \frac{2}{r}, \dots, 1\right\}^n$ and all agents $i \in [n]$, at most $k$ of the values in each set \[\cB_{i,\vec{\omega}} = \left\{\frac{\omega_{i'}\theta_{i'}^{(j)}}{\omega_i} : i'\in [n]\setminus \{i\}, j \in [N]\right\}\] fall within any interval of length $w$. This follows from the following claim and a union bound.

\begin{claim}\label{claim:omega}
For a fixed $\vec{\omega} = (\omega_1, \dots, \omega_n) \in \left\{\frac{1}{r}, \frac{2}{r}, \dots, 1\right\}^n$ and agent $i \in [n]$, with probability $1 - \frac{\delta}{nr^n}$, at most $k$ points from $\cB_{i,\vec{\omega}}$ fall within any interval of width $w$.
\end{claim}
\begin{proof}[Proof of Claim~\ref{claim:omega}]
For each $i' \in [n]\setminus \{i\}$, let \[\cB_{i, i', \vec{\omega}} = \left\{\frac{\omega_{i'}\theta_{i'}^{(j)}}{\omega_i} : j \in [N]\right\}.\] Since $\frac{\omega_{i'}}{\omega_{i}} \geq \frac{1}{r}$, Lemma~\ref{lem:constant_mult} in Appendix~\ref{app:dispersion} implies that the points in $\cB_{i, i', \vec{\omega}}$ are independent draws from a $r\kappa$-bounded distribution. The claim then follows from Lemma~\ref{lem:dispersion} in Appendix~\ref{app:dispersion}.
\end{proof}
\end{proof}

\GSPpdim*

\begin{proof}
Fix an arbitrary agent, and without loss of generality, suppose that agent is agent 1.
Fix a vector $\vec{\theta}_{-1} = \left(\theta_2, \dots, \theta_n\right) \in [0,1]^{n-1}$.
We denote the utility of agent 1 when the agents' types are $\left(\theta_1, \vec{\theta}_{-1}\right)$ and reported types are $\left(\hat{\theta}_1, \vec{\theta}_{-1}\right)$ as $u_{\vec{\theta}_{-1}}\left(\theta_1, \hat{\theta}_1, \omega_1, \dots, \omega_n\right)$, which maps $\R^{n+2}$ to $[-1,1]$.

We now show that we can split $\R^{n+2}$ into regions $R$ where the allocation of the generalized second-price auction is fixed across all $\left(\theta_1, \hat{\theta}_1, \omega_1, \dots, \omega_n\right) \in R$, given the fixed set of reported types $\vec{\theta}_{-1}= \left(\theta_2, \dots, \theta_n\right)$. Let $\bar{\pi}$ be a permutation of the $n$ agents, and let $R_{\bar{\pi}} \subseteq [0,1]^{n+2}$ be the set of all vectors $\left(\theta_1, \hat{\theta}_1, \omega_1, \dots, \omega_n\right)$ where the ordering of the elements $\left\{\omega_1\hat{\theta}_1, \omega_2\theta_2, \dots, \omega_n\theta_n\right\}$ matches $\bar{\pi}$. Specifically, for all $\left(\theta_1, \hat{\theta}_1, \omega_1, \dots, \omega_n\right) \in R_{\bar{\pi}}$, if $i' = \argmax \left\{\omega_1\hat{\theta}_1, \omega_2\theta_2, \dots, \omega_n\theta_n\right\}$, then $\bar{\pi}(1) = i'$ and if $i'' = \argmin \left\{\omega_1\hat{\theta}_1, \omega_2\theta_2, \dots, \omega_n\theta_n\right\}$, then $\bar{\pi}(n) = i''$. By definition of the generalized second-price auction, this means that the agent who receives the first slot is agent $\bar{\pi}(1)$ and the agent who receives the $m^{th}$ is agent $\bar{\pi}(m)$. Therefore, for all $\left(\theta_1, \hat{\theta}_1, \omega_1, \dots, \omega_n\right) \in R_{\bar{\pi}}$, the allocation of the generalized second-price auction given agent weights $\left(\omega_1, \dots, \omega_n\right)$ and reported types $\left(\hat{\theta}_1, \vec{\theta}_{-1}\right)$ is invariant. Next, consider the transformation $\phi : \left(\theta_1, \hat{\theta}_1, \omega_1, \dots, \omega_n\right) \mapsto \left(\theta_1, \frac{\omega_2}{\omega_1}, \dots, \frac{\omega_n}{\omega_1}\right)$. Since the vector $\vec{\theta}_{-1}$ is fixed, the function $u_{\vec{\theta}_{-1}}\left(\theta_1, \hat{\theta}_1, \omega_1, \dots, \omega_n\right)$ is a fixed linear function of $\phi\left(\theta_1, \hat{\theta}_1, \omega_1, \dots, \omega_n\right)$ across all $\left(\theta_1, \hat{\theta}_1, \omega_1, \dots, \omega_n\right) \in R_{\bar{\pi}}$.

Next, let $\sample_{-1} = \left\{\vec{\theta}_{-1}^{(1)}, \dots, \vec{\theta}_{-1}^{(N)}\right\}$ be a set of vectors that is shatterable by the set of functions $\cF_{1, \cM_r}$, where $\vec{\theta}_{-1}^{(j)} = \left(\theta_2^{(j)}, \dots, \theta_n^{(j)}\right)$. This means there is a set of values $z^{(1)}, \dots, z^{(N)} \in \R$ that witnesses the shattering of $\sample_{-1}$ by $\cF_{1, \cM_r}$. Therefore, for every $T \subseteq [N]$, there exists a pair of values $\theta_T, \hat{\theta}_T$ and a mechanism $M_T \in \cM_r$ such that $u_{1, M_T, \theta_T, \hat{\theta}_T}\left(\vec{\theta}_{-1}^{(j)}\right) \leq z^{(j)}$ if and only if $j \in T$. Let $\left(\omega_{1, T}, \dots, \omega_{n, T}\right)$ be the set of agent weights corresponding to each such mechanism $M_T$. We define the set $\cR$ to be the set of $2^N$ vectors $\cR = \left\{\left(\theta_T, \hat{\theta}_T, \omega_{1, T}, \dots, \omega_{n, T}\right) : T \subseteq [N]\right\} \subset \R^{n+2}$.

For each $j \in [N]$, let $\cH^{(j)}$ be the set of ${n \choose 2} + n-1$ hypersurfaces \[\cH^{(j)} = \left\{\omega_1\hat{\theta}_1 - \omega_i\theta_i^{(j)} = 0 : i \in \{2, \dots, n\}\right\} \cup \left\{\omega_i\theta_i - \omega_{i'}\theta_{i'}^{(j)} = 0 : i, i' \in \{2, \dots, n\}\right\}.\] As we saw, as we range $\left(\theta_1, \hat{\theta}_1, \omega_1, \dots, \omega_n\right)$ over a single connected component of $\R^{n+2} \setminus \cH^{(j)}$, the allocation of the generalized second-price auction with agent weights $\left(\omega_1, \dots, \omega_n\right)$ and reported types $\left(\hat{\theta}_1, \vec{\theta}_{-1}^{(j)}\right)$ is invariant. If we define $\cH = \bigcup_{j =1}^N \cH^{(j)}$, 
then as we range $\left(\theta_1, \hat{\theta}_1, \omega_1, \dots, \omega_n\right)$ over a single connected component of $\R^{n+2} \setminus \cH$, the allocations are fixed across all $j \in [N]$. By Lemma~\ref{lem:poly} in Appendix~\ref{app:helpful}, the number of connected components is at most \[2\left(\frac{4eN\left({n \choose 2} + n-1\right)}{n+2}\right)^{n+2} = 2\left(2eN(n-1)\right)^{n+2}.\]

Next, fix a connected component $C$ of $\R^{n+2} \setminus \cH$. We know that for each $j \in [N]$, there is a constant $\alpha_j \in [0,1]$ and a bidder $i_j \not=1$ such that $u_{\vec{\theta}_{-1}^{(j)}}\left(\theta_1, \hat{\theta}_1, \omega_1, \dots, \omega_n\right) = \alpha_j\left(\theta_1 - \frac{\omega_{i_j}\theta_{i_j}}{\omega_1}\right)$ (where $\alpha_j = 0$ if agent $i$ does not receive any slot). Consider the $N$ hypersurfaces $\cH_C = \left\{\alpha_j\left(\theta_1 - \frac{\omega_{i_j}t^{(j)}_{i_j}}{\omega_1}\right) = z^{(j)} : j \in [N]\right\}$. We know that at most one vector in $\cR$ can come from each connected component of $C \setminus \cH_C$, of which there are at most $2\left(\frac{4eN}{n}\right)^n$. In total, this means that $2^N = |\cR| \leq 4\left(2eN(n-1)\right)^{n+2}\left(\frac{4eN}{n}\right)^n$, so by Lemma~\ref{lem:log_ineq} in Appendix~\ref{app:helpful}, $N = O(n \log n)$.
\end{proof}

\discrim*

\begin{proof}
Fix an arbitrary agent, and without loss of generality, suppose that agent is agent 1.
Let $\theta_1' \leq \theta_2' \leq \cdots \leq \theta_{Nm(n-1)}'$ be the sorted values of $\left\{\theta_i^{(j)}[\mu] : i \in \{2, \dots, n\}, j \in [N], \mu \in [m]\right\}$, and define $\theta_0' = 0$ and $\theta_{Nm(n-1) + 1}' = 1$. So long as agent 1's bids fall between these sorted bids, the allocation will be fixed across all samples. More formally, for each $\mu \in [m]$, so long as $\theta_i' < \hat{\theta}_1[\mu] < \theta_{i+1}'$ for some $i \in \{0, \dots, Nm(n-1)\}$, the resulting allocation will be fixed across all $N$ samples.
 Now, for a given sample $\vec{\theta}_{-1}^{(j)}$, consider a region $R \subseteq [0,1]^m$ where the allocation of the discriminatory auction given bids $\left(\hat{\vec{\theta}}_1, \vec{\theta}_{-1}^{(j)}\right)$ is invariant across all $\hat{\vec{\theta}}_1 \in R$. In particular, let $m_R$ be the number of units allocated to agent 1. For a fixed $\vec{\theta}_1 \in [0,1]^m$, agent 1's utility will be linear in $\hat{\vec{\theta}}_1$, because across all $\hat{\vec{\theta}}_1 \in R$, $u_{i, M}\left(\vec{\theta}_1, \hat{\vec{\theta}}_1, \vec{\theta}_{-1}^{(j)}\right) = \sum_{\mu = 1}^{m_R} \theta_1[\mu] - \hat{\theta}_1[\mu]$. Therefore, the functions $u_{1, M}\left(\vec{\theta}_1, \cdot, \vec{\theta}_{-1}^{(1)}\right), \dots, u_{1, M}\left(\vec{\theta}_1, \cdot, \vec{\theta}_{-1}^{(N)}\right)$ are piecewise $1$-Lipschitz.

 Note that across all $\vec{\theta}_1 \in [0,1]^m$, the partition $\cP_j$ splitting $u_{1, M}\left(\vec{\theta}_1, \cdot, \vec{\theta}_{-1}^{(j)}\right)$ into Lipschitz portions is invariant. In particular, no matter the value of $\vec{\theta}_1,$ the partition $\cP_j$ is delineated by the set of $m^2(n-1)$ hyperplanes $\cH_j = \left\{\theta_i^{(j)}[\mu] - \hat{\theta}_1[\mu'] = 0 : i \in \{2, \dots, n\}, \mu, \mu' \in [m]\right\}$. For each $i \in \{2, \dots, n\}$ and pair $\mu, \mu' \in [m]$, let $\cB_{i, \mu, \mu'}$ be the set of hyperplanes \[\cB_{i, \mu, \mu'} = \left\{\theta_i^{(j)}[\mu] - \hat{\theta}_1[\mu'] = 0 : j \in [N]\right\}.\] Within each set, the hyperplanes are parallel with probability 1 and their offsets are independently drawn from $\kappa$-bounded distributions. Therefore, by Lemma~\ref{lem:hyperplanes} in Appendix~\ref{app:dispersion}, with probability $1- \frac{\delta}{n}$, for all $\vec{\theta}_1 \in [0,1]^m$, the functions $u_{1, M}\left(\vec{\theta}_1, \cdot, \vec{\theta}_{-1}^{(1)}\right), \dots, u_{1, M}\left(\vec{\theta}_1, \cdot, \vec{\theta}_{-1}^{(N)}\right)$ are $\left(O\left(\frac{1}{\kappa \sqrt{N}}\right),\tilde{O}\left(nm^2 \sqrt{N}\right)\right)$-dispersed. The theorem holds by a union bound over the agents.
\end{proof}

\SPdisc*

\begin{proof}
Consider an arbitrary bidder, and without loss of generality, suppose that bidder is bidder 1. Next, choose an arbitrary sample $\vec{\theta}_{-1}^{(j)} = \left(\theta_2^{(j)}, \dots, \theta_n^{(j)}\right)$. Letting $j^* = \argmax\left\{\theta_2^{(j)}, \dots, \theta_n^{(j)}\right\}$ and letting $s^{(j)}$ be the second-largest component of $\vec{\theta}_{-1}^{(j)}$, we know that either agent 1 will win and pay $\norm{\vec{\theta}_{-1}^{(j)}}_{\infty}$ or agent $j^*$ will win and pay the maximum of $\hat{\theta}_1$ and $s^{(j)}$. Therefore, for any value $\theta_1 \in [0,1]$ and bid $\hat{\theta}_1 \in [0,1]$, we know that 
\begin{equation}u_{1, M}\left(\theta_1, \hat{\theta}_1, \vec{\theta}_{-1}^{(j)}\right) = \begin{cases}
\left(\alpha_1 - 1\right)\left(\norm{\vec{\theta}_{-1}^{(j)}}_{\infty} - s^{(j)}\right)	&\text{if } \hat{\theta}_1 \leq s^{(j)}\\
\left(\alpha_1 - 1\right)\left(\norm{\vec{\theta}_{-1}^{(j)}}_{\infty} - \hat{\theta}_1\right)	&\text{if } s^{(j)} < \hat{\theta}_1 \leq \norm{\vec{\theta}_{-1}^{(j)}}_{\infty}\\
\alpha_1\left(\theta_1 - \norm{\vec{\theta}_{-1}^{(j)}}_{\infty}\right)	&\text{if } \hat{\theta}_1 > \norm{\vec{\theta}_{-1}^{(j)}}_{\infty}.
\end{cases}\label{eq:spite_utility}\end{equation}
This means that if $\hat{\theta}_1 > \norm{\vec{\theta}_{-1}^{(j)}}_{\infty}$, then $u_{1, M}\left(\theta_1, \hat{\theta}_1, \vec{\theta}_{-1}^{(j)}\right)$ is a constant function of $\hat{\theta}_1$, whereas if $\hat{\theta}_1 \leq \norm{\vec{\theta}_{-1}^{(j)}}_{\infty}$, then $u_{1, M}\left(\theta_1, \hat{\theta}_1, \vec{\theta}_{-1}^{(j)}\right)$ is a continuous function of $\hat{\theta}_1$ with a slope of either 0 (if $\hat{\theta}_1 < s^{(j)}$) or $\alpha_1 - 1$ (if $\hat{\theta}_1 \geq s^{(j)}$). (See Figure~\ref{fig:spiteful} in Appendix~\ref{app:applications}.) Since $\left|\alpha_1 - 1\right| \leq 1$, this means that for all $\theta_i \in [0,1]$, the function $u_{1, M}\left(\theta_1, \cdot, \vec{\theta}_{-1}^{(j)}\right)$ is piecewise $1$-Lipschitz with a discontinuity at $\norm{\vec{\theta}_{-1}^{(j)}}_{\infty}$.
\end{proof}

\subsection{Uniform-price auctions}\label{app:uniform}
Under the uniform-price auction, the allocation rule is the same as in the discriminatory auction (Section~\ref{sec:discriminatory}).
However, all $m$ units are sold at a ``market-clearing'' price,
meaning that the total amount demanded is equal to the total amount supplied.
We adopt the convention~\citep{Krishna02:Auction} that the
market-clearing price equals the highest losing bid.
Let $\vec{c}_{-i} \in \R^m$ be the vector $\hat{\vec{\theta}}_{-i}$ of competing bids facing agent $i$, sorted in decreasing order and limited to the top $m$ bids. This means that $c_{-i}[1] = \norm{\hat{\vec{\theta}}_{-i}}_{\infty}$
is the highest of the other bids, $c_{-i}[2]$ is
the second-highest, and so on. Agent $i$ will win
exactly one unit if and only if $\hat{\theta}_i[1] > c_{-i}[m]$ and $\hat{\theta}_i[2] < c_{-i}[m-1]$; that is, his bid must be higher than the lowest competing bid but not the second-lowest. Similarly, agent $i$ will win
exactly two units if and only if $\hat{\theta}_i[2] > c_{-i}[m-1]$ and $\hat{\theta}_i[3] < c_{-i}[m-2]$. More generally, agent $i$ will win exactly $m_i > 0$ units if and
only if $\hat{\theta}_i\left[m_i\right] > c_{-i}\left[m - m_i + 1\right]$ and $\hat{\theta}_i\left[m_i + 1\right] < c_{-i}\left[m - m_i\right]$.
The market-clearing price, which equals the highest losing bid, is
$p = \max\left\{\hat{\theta}_i\left[m_i + 1\right], c_{-i}\left[m - m_i + 1\right]\right\}$. In a uniform-price auction, agent $i$ pays $m_ip$.

\begin{theorem}
Suppose that each agent's value for each marginal unit has a $\kappa$-bounded density function.
With probability $1-\delta$ over the draw of the $n$ sets $\sample_{-i} = \left\{\vec{\theta}^{(1)}_{-i}, \dots, \vec{\theta}^{(N)}_{-i}\right\} \sim \dist_{-i}^N$, we have that for all agents $i \in [n]$ and types $\vec{\theta}_i \in [0,1]^m$, the functions \[u_{i, M}\left(\vec{\theta}_i, \cdot, \vec{\theta}_{-i}^{(1)}\right), \dots, u_{i, M}\left(\vec{\theta}_i, \cdot, \vec{\theta}_{-i}^{(N)}\right)\] are piecewise $1$-Lipschitz and $\left(O\left(\frac{1}{\kappa \sqrt{N}}\right),O\left(nm^2 \sqrt{N \log\frac{n}{\delta}}\right)\right)$-dispersed.
\end{theorem}

\begin{proof}
This theorem follows by the exact same reasoning as Theorem~\ref{thm:discriminatory_vary_theta_hat}.
\end{proof}

\begin{theorem}
For all agents $i \in [n]$, all $\hat{\vec{\theta}}_i \in [0,1]^m$, and all $\vec{\theta}_{-i} \in [0,1]^{(n-1)m}$, the function $u_{i, M}\left(\cdot, \hat{\vec{\theta}}_i, \vec{\theta}_{-i}\right)$ is 1-Lipschitz.
\end{theorem}

\begin{proof}
So long as all bids are fixed, the allocation is fixed, so $u_{i, M}\left(\cdot, \hat{\vec{\theta}}_i, \vec{\theta}_{-i}\right)$ is 1-Lipschitz function as a function of the true type $\vec{\theta}_i \in [0,1]^m$.
\end{proof}

Next, we prove the following pseudo-dimension bound.

\begin{theorem}
For any agent $i \in [n]$, the pseudo-dimension of the class $\cF_{i, M}$ is $O\left(m\log (nm)\right)$.
\end{theorem}

\begin{proof}
This theorem follows by the exact same reasoning as Theorem~\ref{thm:discriminatory_pdim}.
\end{proof}

\end{document}